\documentclass{article}
\usepackage[margin=1in]{geometry}  %
\usepackage{listings}
\usepackage{amsmath,mathtools}
\usepackage{amssymb}
\usepackage{tikz}
\usetikzlibrary{positioning}
\usepackage{caption,subcaption}
\usepackage{array}
\usepackage{mdwmath}
\usepackage{multirow}
\usepackage{mdwtab}
\usepackage{eqparbox}
\usepackage{tikz}
\usepackage{multicol}
\usepackage{amsfonts}
\usepackage{multirow,bigstrut,threeparttable}
\usepackage{array}
\usepackage{bbm}
\usepackage{epstopdf}
\usepackage{mdwmath}
\usepackage{mdwtab}
\usepackage{eqparbox}
\usepackage{latexsym}
\usepackage{amssymb}
\usepackage{bm}
\usepackage{amssymb}
\usepackage{graphicx}
\usepackage{mathrsfs}
\usepackage{psfrag}
\usepackage{setspace}
\usepackage{tikz-cd}

\usepackage[%
CJKbookmarks=true,
bookmarksnumbered=true,
bookmarksopen=true,
colorlinks=true,
citecolor=red,
linkcolor=blue,
anchorcolor=red,
urlcolor=blue
]{hyperref}
\usepackage{algpseudocode}
\usepackage{amsmath}
\usepackage{amsthm}

\usepackage{algorithm}
\usepackage[utf8]{inputenc}
\usepackage{cleveref}
\usepackage{stfloats}
\usepackage[
autocite    = superscript,
backend     = bibtex, %
sortcites   = true,
style       = authoryear %
]{biblatex} %
\usepackage{nameref}
\usepackage[normalem]{ulem}

\usepackage{soul}

\input xy
\xyoption{all}
\usetikzlibrary{shapes.geometric, arrows, positioning}
\tikzset{
  startstop/.style={rectangle, rounded corners, minimum width=3cm, minimum height=1cm, text centered, draw=black, fill=red!30},
  process/.style={rectangle, minimum width=3cm, minimum height=1cm, text centered, draw=black, fill=blue!30},
  decision/.style={diamond, minimum width=3cm, minimum height=1cm, text centered, draw=black, fill=green!30},
  arrow/.style={thick,->,>=stealth}
}

\newtheorem{lemma}{Lemma}

\newtheorem{assumption}{Assumption}[section]

\crefname{theorem}{Theorem}{Theorems}
\crefname{proposition}{Proposition}{Propositions}

\newtheorem{theorem}{Theorem}           %
\newtheorem{corollary}{Corollary}       %
\newtheorem{proposition}{Proposition}   %
\newtheorem{definition}{Definition}     %

\crefname{theorem}{Theorem}{Theorems}
\crefname{corollary}{Corollary}{Corollaries}
\crefname{proposition}{Proposition}{Propositions}
\crefname{definition}{Definition}{Definitions}

\def\EE{\mathbb{E}}

\def\PP{\mathbb{P}}

\def\RR{\mathbb{R}}

\def\calI{\mathcal{I}}

\def\calN{\mathcal{N}}

\def\calP{\mathcal{P}}

\def\calS{\mathcal{S}}

\newcommand{\ind}{\textup{ind}}

\def\1{\mathbbm{1}}
\def\var{\mathsf{Var}}

\newcommand\independent{\protect\mathpalette{\protect\independenT}{\perp}}
\def\independenT#1#2{\mathrel{\rlap{$#1#2$}\mkern2mu{#1#2}}}
\usepackage{booktabs}
\usepackage{adjustbox}
\usepackage{multirow}
\usepackage{listings}

\newtheorem{remark}{Remark}

\def \var {\mathsf{Var}}

\usepackage{xspace}

\def\independenT#1#2{\mathrel{\rlap{$#1#2$}\mkern2mu{#1#2}}}

\definecolor{myblue}{rgb}{.8, .8, 1}
\definecolor{mathblue}{rgb}{0.2472, 0.24, 0.6} %
\definecolor{mathred}{rgb}{0.6, 0.24, 0.442893}
\definecolor{mathyellow}{rgb}{0.6, 0.547014, 0.24}

\usepackage{enumitem}

\newcommand{\EIF}{{\text{EIF}}}
\newcommand{\IF}{{\text{IF}}}
\newcommand{\np}{{\text{np}}}

\usepackage{pgfplots}
\pgfplotsset{compat=1.17}

\title{Estimation and Inference for Causal Explainability}
\usepackage{times}
\usepackage{authblk}

\author[1]{Weihan Zhang}
\author[2]{Zijun Gao}

\affil[1]{Division of Biostatistics, University of California, Berkeley\\
\texttt{weihan\_zhang2001@berkeley.edu}}

\affil[2]{Department of Data Science and Operations, University of Southern California\\
\texttt{zijungao@marshall.usc.edu}}

\begin{document}

\maketitle

\begin{abstract}
Understanding how much each variable contributes to an outcome is a central question across disciplines.
A causal view of explainability is favorable for its ability in uncovering underlying mechanisms and generalizing to new contexts. 
Based on a family of causal explainability quantities, we develop methods for their estimation and inference.
In particular, we construct a one-step correction estimator using semi-parametric efficiency theory, which explicitly leverages the independence structure of variables to reduce the asymptotic variance.
For a null hypothesis on the boundary, i.e., zero explainability, we show its equivalence to Fisher's sharp null, which motivates a randomization-based inference procedure. 
Finally, we illustrate the empirical efficacy of our approach through simulations as well as an immigration experiment dataset, where we investigate how features and their interactions shape public opinion toward admitting immigrants.
\end{abstract}


\section{INTRODUCTION}\label{sec:introduction}

\textit{``The truth is rarely pure and never simple.'' --- Oscar Wilde.} 
Outcomes typically rely on multiple factors rather than a single cause: 
customer ratings are influenced by various attributes of an product;
biological traits are influenced by multiple genes \parencite{fisher1918correlation,visscher2017tenyears, watanabe2019global, sivakumaran2011abundant}.
In the presence of multiple factors, it is practically relevant to quantify their explainability, that is how much a set of factors and their interactions account for the variability in the outcome.

Compared with association-based measures, causality-based explainability can illuminate mechanisms and transfer across contexts reliably, and is sometimes regarded as the foundation of scientific explanation
\parencite{pearl2009,vanderweele2015explanation}.
This motivates us to adopt the Causal ANOVA \parencite{pmlr-v275-gao25a}, a set of quantities generalizing the functional ANOVA \parencite{hooker2007generalized} and Sobol indices \parencite{sobol2001global} to the counterfactual perspective, to measure the explainability.
For the statistical analysis of causal ANOVA quantities, prior work has primarily focused on settings in which an oracle, typically a black-box machine learning model, is available to evaluate the outcome at arbitrarily chosen input values, with the resulting explainabilities used to assess feature importance.
However, in many domains including social science and genetics, researchers rely on collected data samples and have no access to such oracles.

Semi-parametric methods target the estimation and inference for finite-dimensional parameters while allowing nuisance components to be non-parametric, and operate on pre-collected datasets requiring no oracle.
This aligns with our scope: a scalar target (explainability summarized by a number), minimal modeling assumptions, and no oracle access. 
Therefore, we develop semi-parametric methods for the estimation and inference of Causal ANOVA quantities.
Particularly, we derive the efficient influence function and construct the associated one-step correction estimator, and explicitly show how the independence structure in Causal ANOVA can be leveraged to gain efficiency.
Degenerate nulls (the parameter of interest lies on the boundary of its space) are a known challenge in semi-parametric statistics \parencite{verdinelli2024feature}, and we design a valid randomization-based procedure for such nulls in contrast to the common practice based on sample splitting or noise injection.

\paragraph{Contributions.}
We investigate the estimation and inference for the explainability attributable to individual factors and their interactions through the lens of semi-parametric statistics.
\begin{itemize}    
    \item Our approach can be used to identify influential factors and interactions by testing for non-zero explainability, and further quantify their importance by providing confidence intervals.
   Applied to an immigration-experiment dataset, multiple factors are shown to exhibit nonzero causal explainability for public opinions toward admitting immigrants, with interaction results compared against CRT and ACIE analyses in \Cref{app:rda}.

   \item  We derive the efficient influence function that explicitly leverages the structure, i.e., (conditional) independence, in the probability model and construct the associated one-step correction estimator, which attains no larger asymptotic variance compared to the counterpart that ignores the structure. 
    The benefits of exploiting independence structure as well as our derivation extend to other applications with independent components, such as factorial A/B tests. 
 
    \item We propose a randomization-based procedure to test whether the null is degenerate,
    which requires neither sample splitting nor noise injection.
    Further combined with the confidence interval using the asymptotic distribution of one-step correction estimators (valid under non-degenerate nulls), we introduce a sequential procedure that returns valid confidence sets regardless of the null degeneracy.
\end{itemize}

\paragraph{Organization.}
In \Cref{sec:formulation}, we review Causal ANOVA, semi-parametric estimation theory, and randomization test.
In \Cref{sec:estimation}, we exploit the independence structure in the probability model to derive a more efficient influence function and provide the associated one-step estimator.
In \Cref{sec:inference}, we derive the one-step estimator's asymptotic distribution under non-degenerate nulls, introduce a randomization-based test for degenerate nulls, and combine them into a sequential procedure.
In \Cref{sec:simulation}, \Cref{sec:real.data}, we demonstrate our procedure in simulations and a real world dataset.
Proofs, additional simulations are provided in the appendix.

\paragraph{Notation.}
Let $\xi$ denote the causal estimand of interest.
Let $Y$ be the outcome and $\bm W=(W_1,\ldots,W_K)$ be the vector of $K$ factors; write $[K]=\{1,\ldots,K\}$ and, for $\mathcal S\subseteq[K]$, write $\bm W_{\mathcal S}$ for the corresponding subvector.
For $\calS\subseteq[K]$, write $-\calS=[K]\setminus\calS$ for its complement.
Let $\mathbb P$ denote the true joint law of $(Y,\bm W)$ on the product measurable space $(\mathcal Y\times\mathcal W,\ \mathcal B(\mathcal Y)\otimes\mathcal B(\mathcal W))$.
Let $\mu(\bm w)=\EE[Y\mid\bm W=\bm w]$ and $\upsilon(\bm w)=\EE[Y^2\mid\bm W=\bm w]$ denote the conditional mean and second moment.
Denote the marginal law of $\bm W$ by $\mathbb P_{\bm W}$ and the marginal law of $W_k$ by $\mathbb P_{W_k}$.
When we impose the product structure for $\bm W$, we write $\mathbb P_{\bm W} = \bigotimes_{k=1}^K \mathbb P_{W_k}$.
To introduce densities, fix $\sigma$-finite dominating measures $\lambda_k$ on $(\mathcal W_k,\mathcal B(\mathcal W_k))$ (e.g.\ counting measure for discrete $W_k$ or Lebesgue measure for continuous $W_k$), and let $\lambda_{\bm W}:=\bigotimes_{k=1}^K \lambda_k$.
When $\mathbb P_{W_k}\ll \lambda_k$, we write $p_{W_k}(w_k) := d\mathbb P_{W_k}/d\lambda_k(w_k)$ and $p_{\bm W}(\bm w) := d\mathbb P_{\bm W}/d\lambda_{\bm W}(\bm w)$.
Under the product structure $\mathbb P_{\bm W}=\bigotimes_{k=1}^K \mathbb P_{W_k}$ and marginal absolute continuity, the joint density factorizes as $p_{\bm W}(\bm w) \;=\; \prod_{k=1}^K p_{W_k}(w_k) \quad \lambda_{\bm W}\text{-a.e.}$ We denote the conditional law of $Y$ given $\bm W$ by $\mathbb P_{Y\mid \bm W}$; when dominated by a measure $\lambda_Y$ (e.g.\ Lebesgue for continuous $Y$ or counting measure for discrete $Y$), we write $p_{Y\mid \bm W}(y\mid \bm w) := d\mathbb P_{Y\mid \bm W=\bm w}/d\lambda_Y(y)$.
Estimated counterparts are denoted by $\widehat{\mathbb P}$, $\widehat{\mu}$ and $\widehat{\upsilon}$, with $\widehat{\mathbb P}_{\bm W}$, $\widehat{\mathbb P}_{W_k}$, and $\widehat{\mathbb P}_{Y\mid \bm W}$ as well as densities defined analogously sharing the same dominating measure.
For a probability model $\mathcal P$ and $\mathbb P\in\mathcal P$, let $\dot{\mathcal P}_{\mathbb P}$ denote the tangent space at $\mathbb P$ (subscript $\mathbb P$ dropped when unambiguous).
We write $\varphi_{\IF}^\xi$, $\varphi_{\EIF}^\xi$, or simply $\varphi^\xi$ for influence functions of the estimand $\xi$ (superscript $\xi$ dropped when unambiguous).
For a comprehensive summary of notation, see \Cref{tab:notation}.

\section{PROBLEM FORMULATION}\label{sec:formulation}

\subsection{Causal ANOVA}\label{sec:causal.ANOVA}

Causal ANOVA is set in the potential outcome framework~\parencite{rubin1974estimating}, with $K$ randomized treatments $W_k\in\mathcal W_k$ and $n$ units.
Each unit $i$ has potential outcomes $\bm Y_i(\cdot)$, but only $Y_i = Y_i(\bm W_i)$ is observed for the realized assignment $\bm W_i = (W_{1,i},\ldots,W_{K,i})$.
To ensure that Causal ANOVA quantities generalize to future observations, we assume the observed units are drawn i.i.d.\ from a super-population model. Formally,
\begin{align}\label{eq:super.population}
    \bm W_i \sim \mathbb P_{\bm W}, 
    \ 
    \bm Y_i(\cdot) \sim \mathbb P_{Y(\cdot)}, \ \bm W_i \independent Y_i(\cdot),
\end{align}
independently across units,
where $\mathbb P_{\bm W}$ is the (unknown) distribution of treatment assignments and $\mathbb P_{Y(\cdot)}$ is the unknown distribution of potential outcomes.

Causal ANOVA defines explainability via contrasts among potential outcomes, i.e., comparing factuals and counterfactuals.
While Causal ANOVA can accommodate any set of factors $\bm W$ that respect a directed acyclic graph (DAG), identification of the explainabilities under dependent factors generally requires additional assumptions (e.g., comonotonicity of potential outcomes).
As a result, in this paper we focus on independent treatments specified below.
\begin{assumption}[Independent treatments]\label{assu:independent}
The factors $W_k$, $k\in[K]$, are mutually independent.
\end{assumption}
\noindent Let $\mathcal P_{\ind}$ denote the model satisfying \Cref{assu:independent}, a proper subset of the non-parametric model $\calP_{\np}$ in \eqref{eq:super.population}.
\Cref{assu:independent} is satisfied in many experiment designs, for example, conjoint analysis randomizes the value of attributes independently.
More generally, settings with block-wise independence of treatment groups can be analyzed in a similar manner, where each block of treatments is regarded
as a single combined factor.
In \Cref{appe:sec:causal.ANOVA}, we show \Cref{assu:independent} can be relaxed to the conditional independence: $W_k$ are mutually independent given some covariates $\bm X$.

Let $\bm W'$ have the same law as $\bm W$ and be independent of $(\bm W,Y(\cdot))$. Throughout, assume $Y\in L_2$ and $0<\var(Y)<\infty$.
Under \Cref{assu:independent}, Causal ANOVA uses the variance of the difference obtained by replacing $\bm W_\calS$ with $\bm W'_\calS$ to define the total explainability of $\bm W_\calS$, $\calS \subseteq [K]$.
\begin{definition}[Total explainability] \label{defi:total.independent}
    For $\calS \subseteq [K]$,\begin{align}\label{eq:total.variance}
        \xi(\vee_{k \in \calS} W_k)
        := \frac{\var\left(Y(\bm W) - Y(\bm W'_{\calS}, \bm W_{-\calS})\right)}{2\var\left(Y(\bm W)\right)}.
    \end{align}
\end{definition}
\noindent The explainability of interaction is defined as the variance of a difference-in-difference contrast.
\begin{definition}[Interaction explainability] \label{defi:interaction.explainability}
For $k \neq k' \in [K]$,
\begin{align}\label{eq:xi-3}\begin{split}
    &\xi(W_{k}\wedge W_{k'})
  :=\frac{\var\left(I(W_k^\prime,W^\prime_{k^\prime})\right)}{4\,\var\left(Y(\bm W)\right)},
\end{split}
\end{align} where $I(W_k^\prime,W^\prime_{k^\prime}) = Y(W_{k}',W_{k'}')-Y(W_{k}',W_{k'})-Y(W_{k},W_{k'}')+Y(W_{k},W_{k'})$.
Here the remaining treatments are held at $\bm W_{-\{k,k'\}}$.
\end{definition}
\noindent 
\Cref{defi:interaction.explainability} defines the first-order interaction, and can be generalized to higher-order interactions. However, we omit them due to the reduced usefulness and increased complexity.

When $Y$ is not fully determined by $\bm W$, the Causal ANOVA quantities are generally not identifiable. 
The non-parametric structural equation model with additive, independent errors, a.k.a. NPSEM-IE \parencite{pearl2009}, offers a useful balance between identifiability and generality.
\begin{assumption}[NPSEM-IE]\label{assu:additive.independent.error}
    The potential outcomes satisfy $Y(\bm w)=f_Y(\bm w)+E_Y$, with the same unit-level error $E_Y$ across treatment levels and $E_Y$ independent of $\bm W$. The observed outcome is $Y=Y(\bm W)$.
\end{assumption}
\noindent NPSEM-IE is a commonly used model for causal discovery and inference \parencite{hoyer2008nonlinear,peters2014causal}. We discuss its extensions in \Cref{sec:NPSEM.extension}. 

\begin{proposition}\label{lemma:definition}
   Under \Cref{assu:independent} and \Cref{assu:additive.independent.error}, for any $\calS \subseteq [K]$, \Cref{defi:total.independent}~is identifiable and admits the form, \begin{align}\label{eq:total.variance.2}
    \xi(\vee_{k \in \calS} W_k)
    = \frac{\mathbb{E}[\mathbb{E}^2\left[Y\mid \bm W]\right]}{\var\left(Y\right)}-\frac{\mathbb{E}\left[\mathbb{E}^2[Y\mid \bm W_{-\mathcal{S}}]\right]}{\var\left(Y\right)}.
\end{align}
\end{proposition}
\noindent By the inclusion-exclusion principle for causal explainabilities, that is
$ \xi(W_{k} \wedge W_{k'})
     =  \xi(W_{k})+ \xi(W_{k'})- \xi(W_{k} \vee W_{k'})$,
we obtain an analogue of \Cref{lemma:definition} for the interaction explainabilities.
\begin{corollary}\label{coro:definition.interaction}
   Under the assumption of \Cref{lemma:definition}, for $k \neq k' \in [K]$, \Cref{defi:interaction.explainability} is identifiable and admits the form,  \begin{align}\label{eq:interaction.regression}
   \begin{split}
    \xi(W_{k} \wedge W_{k'})
    =& \frac{\mathbb{E}[\mathbb{E}^2\left[Y\mid \bm W]\right]}{\var\left(Y\right)}-\frac{\mathbb{E}\left[\mathbb{E}^2[Y\mid \bm W_{-k}]\right]}{\var\left(Y\right)}\\
    &-\frac{\mathbb{E}\left[\mathbb{E}^2[Y\mid \bm W_{-k'}]\right]}{\var\left(Y\right)}+ \frac{\mathbb{E}\left[\mathbb{E}^2[Y\mid \bm W_{-\{k,k'\}}]\right]}{\var\left(Y\right)}.
    \end{split}
\end{align}
\end{corollary}

\subsection{Influence function and one-step correction}\label{sec:IF}

Influence functions are critical in semi-parametric inference for constructing
approximately unbiased and efficient estimators robust to nuisance
estimation~\parencite{bickel1993efficient}. 
Let the quantity of interest $\xi$ be some functional of $\PP$, e.g., the causal explainability for a set of factors.
Let $\mathcal{P}$ be a class of probabilities for the law $\PP$ of $(Y_i,\bm W_i)$, and define its tangent space $\dot{\calP}$ as the closure of the linear span of its scores.
Suppose the functional $\xi$ admits the von Mises expansion with remainder $R_2$:
\begin{align}\label{eq:von.mises}
    \xi(\PP_{\varepsilon}) - \xi(\PP) 
    = \int \varphi(y,\bm w; \PP_{\varepsilon}) \, d(\PP_{\varepsilon} - \PP)(y,\bm w) 
    + R_2(\PP_{\varepsilon}, \PP),
\end{align}
for a perturbed distribution $\PP_{\varepsilon}$ of $\PP$, then $\varphi(y,\bm w; \PP)$ is a mean-zero, 
finite-variance function satisfying 
\[\int \varphi(y,\bm{w}; \mathbb{P}) d\mathbb{P}(y,\bm{w}) = 0,\] typically called an influence function.
Influence functions are not unique; however, the projection of any influence function onto the tangent space $\dot{\calP}$ is unique, which is called the \emph{efficient influence function}, denoted by $\varphi_{\EIF}(y,\bm w; \PP)$, and the semi-parametric efficiency bound of $\xi$ is $\var(\varphi_{\EIF}(Y,\bm W; \PP))$.
We emphasize that for the same functional, the efficient influence function depends on the chosen model class $\calP$. 
A smaller class $\calP$ yields a smaller tangent space $\dot{\calP}$ and hence likely a lower semi-parametric efficiency bound.

Influence functions provide a principled way to correct plug-in estimators and construct less biased, more efficient estimators. 
Particularly, let $\widehat{\PP}$ be an estimator of the true distribution $\PP$. 
A natural plug-in estimator of $\xi(\PP)$ is $\xi(\widehat{\PP})$. 
Given an influence function $\varphi$, the associated one-step corrected estimator is
\begin{align}\label{eq:IF.corrected}
    \widehat{\xi} 
    = \xi(\widehat{\PP}) + \PP_n \widehat{\varphi}, \quad \PP_n \widehat{\varphi} := \frac{1}{n}\sum_{i=1}^n 
      \varphi(Y_i,\bm W_i;\widehat{\PP}).
\end{align}
If $\xi(\widehat{\PP})$ is regular and efficient, then the one-step corrected 
estimator $\widehat{\xi}$ is also efficient, in the sense of attaining the 
semi-parametric efficiency bound.

We also review the literature on degenerate nulls (the parameter of interest lies on the boundary of the parameter space).
For degenerate nulls, the influence functions of the estimator vanish. As a result, the estimator typically converges to a degenerate distribution, which can not be used to construct valid confidence intervals.
Currently there are three major approaches to addressing the above problem \parencite{dai2022significance,williamson2023general,verdinelli2024feature}. As noted by \cite{verdinelli2024feature}, all these approaches rely on $O(n^{-1/2})$-level expansions to preserve validity under the null, at the cost of efficiency. In \cite{hudson2023nonparametric}, an estimator is proposed, which converges at rate $n^{-1}$ under the null and at rate $n^{-1/2}$ away from the null. However, the method does not yield valid confidence intervals in the null hypothesis's neighborhood.

\subsection{Fisher's randomization test}\label{sec:randomization.test}

Fisher's randomization test originally evaluates a test statistic over the randomization distribution induced by the assignment mechanism under some sharp null \parencite{fisher1918correlation}.
Explicitly, given a test statistic $T$, we first compute $T(Y,\bm W)$ under the observed assignment $\bm W$.
Under the sharp null, we impute the potential outcomes under the partially sharp null, draw assignments $\bm W^{(b)}$ independently from the known conditional assignment mechanism, using the same conditioning as for the observed assignment, evaluate $T(Y,\bm W^{(b)})$ to form the randomization distribution, and compute the p-value
\begin{align}\label{eq:randomization.p.value}
P := \frac{1+\sum_{b=1}^B \mathbf 1\{\,T(Y,\bm W^{(b)})\ge T(Y,\bm W)\,\}}{1+B}.
\end{align}
When the observed and resampled statistics are jointly exchangeable under the null, any finite $B$ gives a finite-sample valid $p\text{-value}$ in the sense that $\PP(P \leq \alpha) \leq \alpha$ for all $\alpha \in (0,1)$ \parencite{PhipsonSmyth2010pvalue}.
In Causal ANOVA, the degenerate null (explainability being zero) implies Fisher's sharp null (\Cref{sec:inference}), motivating the use of randomization-based inference.

\section{ESTIMATION}\label{sec:estimation}

We build on the semi-parametric efficiency theory to derive the influence function for the Causal ANOVA quantities~\eqref{defi:total.independent} and then incorporate it into the standard cross-fitting procedure (\Cref{algo:cross.fitting}) \parencite{schick1986asymptotically,klaassen1987consistent} to obtain the one-step correction estimator \parencite{pfanzagl1990estimation,bickel1993efficient}.
Among the one-step correction estimators satisfying the asymptotic distribution established in \Cref{sec:inference} below, the one with a smaller asymptotic variance, i.e., higher efficiency, is preferred.
In particular, the influence function that yields an estimator attaining the semi-parametric efficiency bound regarding a model $\calP$ ($\calP$ denoting a class of candidate joint distributions of $(Y,\bm W)$) is optimal and referred to as the efficient influence function (for $\calP$).

According to \Cref{lemma:definition} and \Cref{coro:definition.interaction}, for the total or interaction explainabilities, it suffices to estimate quantities of the form ${\mathbb{E}\left[\mathbb{E}[Y\mid \bm W_{-\calS}]^2\right]}/{\var\left(Y\right)}$, $\calS \subseteq [K]$, which is pathwise differentiable \parencite{Williamson2021}.

\subsection{Baseline estimator}\label{sec:EIF.np}
To begin with, consider the non-parametric model $\mathcal{P}_{\np}$ for $(Y, \bm{W})$. Under this model, the efficient influence function for $\mathbb{E}[\mathbb{E}(Y \mid \bm{W}_{-\mathcal{S}})^2] / \var(Y)$ takes the form 
(derivation in \Cref{appe:sec:proof:lemma:influence.function.conditional.expectation.squared}),
\begin{align}\label{eq:IF.np}
\begin{split}   
&\varphi_{\IF}(Y,\bm W; \PP)
:= \frac{2Y\mathbb{E}[Y\mid \bm W_{-\calS}] - \mathbb{E}[Y\mid \bm W_{-\calS}]^2}{\var\left(Y\right)}-\mathbb{E}\left[\mathbb{E}[Y\mid \bm W_{-\calS}]^2\right] \left(\frac{Y - \EE\left[Y\right]}{\var\left(Y\right)}\right)^2.
\end{split}
\end{align}

\subsection{Estimator leveraging independence}\label{sec:EIF.independent}

We leverage the independence among $W_k$ \eqref{assu:independent} to refine $\varphi_{\IF}$ in \eqref{eq:IF.np} to gain efficiency. The key step is to project $\varphi_{\IF}$ onto the tangent space of the restricted model $\dot{\mathcal{P}}_{\mathrm{ind}}$, a subspace of $\dot{\mathcal{P}}_{\mathrm{np}}$, to get the efficient influence function $\varphi_{\EIF}$. The difference $\varphi_{\IF} - \varphi_{\EIF}$, while may be influential under $\mathcal{P}_{\np}$, has zero influence on the target estimand when $W_k$ are independent, and variation along these directions can inflate the variability. 
By removing them, we do not lose any information of the target estimand under the independence~\eqref{assu:independent} and obtain no larger asymptotic variance. 
\Cref{exam:project} provides an alternative explanation that demonstrates the projection procedure and the resulting efficiency gain. The next result shows the efficient influence function incorporating the independence~\eqref{assu:independent}.

\begin{proposition}\label{lemm:influence.function.conditional.expectation.squared}
Under the assumptions of \Cref{lemma:definition}, fix any $\calS \subseteq [K]$.
Assume $0<\var(Y)<\infty$ and that $2Y\EE[Y\mid\bm W_{-\calS}]-\EE[Y\mid\bm W_{-\calS}]^2-\EE[\EE[Y\mid\bm W_{-\calS}]^2]$ and $(Y-\EE Y)^2-\var(Y)$ belong to $L_2(\PP)$.
The square-integrability conditions hold, in particular, if $\EE[|Y|^4]<\infty$.
The efficient influence function of $\frac{\mathbb{E}\left[\mathbb{E}\left[Y\mid \bm W_{-\mathcal{S}}\right]^2\right]}{\var\left(Y\right)}$ in the model with unrestricted conditional law $\PP_{Y\mid\bm W}$ and independent, otherwise unrestricted treatment marginals, i.e., $\PP_{\bm W} = \prod_{k \in [K]} \PP_{W_k}$, takes the form
\begin{align}\label{eq:EIF}
\begin{split}
    \varphi_{\EIF}(Y,\bm W;\PP) =& \frac{2\left(Y-\mathbb{E}[Y\mid \bm W]\right)\cdot \mathbb{E}[Y\mid \bm W_{-\calS}]}{\var\left(Y\right)}+ \frac{2\sum_{k\in \calS}\mathbb{E}\left[\mathbb{E}[Y\mid \bm W]\cdot\mathbb{E}[Y\mid \bm W_{-\calS}]\mid W_k\right]}{\var\left(Y\right)}\\
&+\frac{\sum_{k\in -\calS}\mathbb{E}\left[\mathbb{E}[Y\mid \bm W_{-\calS}]^2\mid W_k\right]}{\var\left(Y\right)}- \frac{(2|\calS|+|-\calS|)\cdot\mathbb{E}\left[\mathbb{E}[Y\mid \bm W_{-\calS}]^2\right]}{\var(Y)}\\
&-\frac{\varphi^{\var(Y)}_{\EIF}(Y,\bm W;\PP)\cdot\mathbb{E}\left[\mathbb{E}\left[Y\mid \bm W_{-\calS}\right]^2\right]}{\var\left(Y\right)^2},
        \end{split}
        \\
\begin{split}
   \varphi^{\var(Y)}_{\EIF}(Y,\bm W;\PP) = &\left(Y-\EE\left[Y\right]\right)^2-\var\left(Y\mid \bm W\right)-\left(\EE[Y\mid\bm W]-\EE[Y]\right)^2\\
&+\sum_{k\in[K]}\EE\left[(Y-\EE[Y])^2\mid W_k\right]-K\cdot\var\left(Y\right).
\end{split}
\end{align}
\end{proposition}
The proof of \Cref{lemm:influence.function.conditional.expectation.squared} relies on explicitly characterizing the tangent space $\dot{\calP}_{\ind}$ under~\eqref{assu:independent} and then performing the projection.
Details are provided in \Cref{appe:sec:proof}.
In \Cref{sec:simulation}, we demonstrate the efficiency gain of $\varphi_{\EIF}$ over $\varphi_{\IF}$ in simulated datasets.

\begin{remark}\label{exam:project}According to \textit{Proposition 1.} in Chapter 3 of \cite{bickel1993efficient}, the efficient influence function can be decomposed as $\varphi_{\EIF} = \varphi - \Pi\{\varphi\mid \dot{\mathcal{P}}_{\ind}^\bot\}$, where $\varphi$ is any influence function and $\varphi_{\EIF}$ is orthogonal to $ \Pi\{\varphi\mid \dot{\mathcal{P}}_{\ind}^\bot\}$. 
    The estimator based on $\varphi_{\IF}$, $\varphi_{\EIF}$ has asymptotic variance $\var(\varphi_{\IF})$, $\var(\varphi_{\EIF})$, respectively; for $\PP\in\mathcal P_{\ind}$, $\var(\varphi_{\IF})
        = \var(\varphi_{\EIF}) + \var\left(\Pi\{\varphi_{\IF}\mid \dot{\mathcal{P}}_{\ind}^\bot\}\right)
        \ge \var(\varphi_{\EIF})$.

\end{remark}

\subsection{Estimators leveraging additional structure}\label{sec:EIF.additional}

We briefly discuss the possibility of further improving the efficiency of $\varphi_{\EIF}$ in \eqref{eq:EIF} by leveraging additional structures of $\calP$.

If we further include the additive and independent noise assumption~\eqref{assu:additive.independent.error}, known as the location-shift regression model \parencite{tsiatis2006semiparametric}, $\varphi_{\EIF}$ can be refined accordingly and the improved influence function requires the distribution of the error term $E_Y$ (in \eqref{assu:additive.independent.error}).
However, $E_Y$ is unobserved, and its distribution is difficult to estimate with sufficient accuracy to ensure the asymptotic distribution in \Cref{sec:inference}.
Therefore, the efficiency gain from incorporating~\eqref{assu:additive.independent.error} is primarily conceptual rather than practical, and we therefore continue with $\varphi_{\EIF}$ in what follows (see Chapter 5 of \cite{tsiatis2006semiparametric} for a more detailed discussion).

If $\PP_{\bm W}$ is known, $\varphi_{\EIF}$ can be further improved, and we provide the derivation of this efficient influence function in \Cref{remark: propensity.2}.
However, in many real-data scenarios, including the application in \Cref{sec:real.data}, the distribution $\PP_{\bm W}$ is unavailable, and thus we regard $\varphi_{\EIF}$ as more relevant and useful.

For cross-fitting, let $L\ge2$ be the number of folds and let $\{\calI_l\}_{l=1}^L$ partition $\{1,\ldots,n\}$.
Write $n_l=|\calI_l|$ for the fold sizes and $\calI_{-l}=\{1,\ldots,n\}\setminus\calI_l$ for the training indices.

The functions in \eqref{eq:IF.np} and \eqref{eq:EIF} refer to a regression-square component $\EE[\EE[Y\mid\bm W_{-\calS}]^2]/\var(Y)$.
For a total or interaction explainability $\xi$, its influence functions $\varphi_{\IF}^{\xi}$ and $\varphi_{\EIF}^{\xi}$ are the corresponding linear combinations specified by \eqref{eq:total.variance.2} and \eqref{eq:interaction.regression}.
In the inference statements and algorithms, $\varphi_{\IF}$ and $\varphi_{\EIF}$ denote these full-target influence functions, with the superscript $\xi$ omitted.
In these inference statements and algorithms, $\sigma^2=\var\{\varphi(Y,\bm W;\PP)\}$ denotes the variance of the chosen influence function.

\begin{algorithm}[tbp]
\caption{Cross-fitting of one-step correction estimator}
\label{algo:cross.fitting}
\begin{algorithmic}
\State\textbf{Input:} Data $\{\bm W_i,Y_i\}_{i=1}^n$, number of folds $L \ge 2$, an influence function $\varphi$, nuisance learners for ${\PP}$, ${\mu}(\bm w)=\EE[Y\mid \bm W=\bm w]$, ${\upsilon}(\bm w)=\EE[Y^2\mid \bm W=\bm w]$ (needed only for $\varphi_{\EIF}$). 

\State\textbf{Partition:} Partition the sample into $L$ folds $\{\calI_\ell\}_{\ell=1}^L$ with (approximately) equal sizes $n_\ell=|\calI_\ell|$.

\For{$\ell=1,\dots,L$}

\State \textbf{Nuisance estimation:} On $\calI_{-\ell}$, obtain nuisance estimators $\widehat\mu_{-\ell}(\bm w)$ for $\mu(\bm w)$, $\widehat\upsilon_{-\ell}(\bm w)$ for $\upsilon(\bm w)$ when using $\varphi_{\EIF}$, $\widehat{p}_{W_k,-\ell}$ for $p_{W_k}$, $k=1,\dots,K$, and compute $\widehat{p}_{\bm W,-\ell}:=\prod_{k=1}^K \widehat{p}_{W_k,-\ell}$.

\State For $\varphi_{\IF}$, use the training mean $\widehat{\EE}_{-\ell}[Y]=(n-n_\ell)^{-1}\sum_{i\in\calI_{-\ell}}Y_i$ and variance $\var(Y)_{\widehat{\PP}_{-\ell}}=(n-n_\ell)^{-1}\sum_{i\in\calI_{-\ell}}(Y_i-\widehat{\EE}_{-\ell}[Y])^2$.
For $\varphi_{\EIF}$, use the fitted integrated moments in \Cref{tab:EIF}.
Here $\widehat\PP_{-\ell}$ denotes the collection of fitted components.
Define the estimator arbitrarily if a fitted denominator is zero.

\State \textbf{One-step correction:} On the held-out fold $\calI_\ell$, compute
\begin{align}\label{eq:crossfit}
\widehat\xi_{\ell}
:=
\frac{1}{n_\ell}\sum_{i\in \calI_\ell}\widehat{\xi}_{i},\quad
\widehat{\xi}_{i}:=\xi\!\big(\widehat{\PP}_{-\ell}\big)
\;+\;
\varphi\!\big(Y_i,\bm W_i;\widehat{\PP}_{-\ell}\big),
\end{align}
where the details $\xi(\widehat{\PP}_{-\ell})$ and
$\varphi(Y_i, \bm W_i; \widehat{\PP}_{-\ell})$ are given in \Cref{tab:EIF}.

\EndFor

\State\textbf{Aggregation.} Average the estimates over folds $\widehat{\xi} := \frac{1}{n}\sum_{\ell=1}^L n_\ell\widehat\xi_{\ell}$.
Compute the variance estimator $\widehat{\sigma}^2 :=\frac{1}{n}\sum_{i=1}^n \big(\widehat{\xi}_{i} - \widehat{\xi}\big)^2$.

\State Set $\widehat\sigma=\sqrt{\widehat\sigma^2}$.

\State \textbf{Output.} Estimator $\widehat{\xi}$ and its estimated standard error $\widehat{\sigma}/\sqrt{n}$.
\end{algorithmic}
\end{algorithm}

\section{INFERENCE}\label{sec:inference}

\subsection{Asymptotic distribution of one-step correction estimator}\label{sec:asymptotic.distribution}

Write $\widehat{\varphi}_{\IF}=\varphi_{\IF}(Y,\bm W;\widehat{\PP})$ for the influence function with fitted nuisance components, and similarly for $\widehat{\varphi}_{\EIF}$.

Throughout these asymptotic statements, $K$ and the number of folds $L$ are fixed, $\min_l n_l/n$ is bounded away from zero, and the partition is chosen independently of the observations.
Each nuisance estimator is measurable with respect to its training fold and any fitting randomness independent of the held-out observations.
Assume $\var(Y)>0$.
Unqualified norms are $L_2(\PP)$ norms, with functions of a subvector lifted to the full observation; marginal-density norms are $L_2(\PP_{W_k})$ norms.
For probability laws, write $\|Q-P\|_{\mathrm{var}}:=\int|dQ-dP|$.
All conditions below hold for each fold and each of the finitely many regression-square components entering the estimand.
Write $\widehat{\mu}_{-\calS,-l}$ for the fitted reduced conditional mean $\EE[Y\mid\bm W_{-\calS}]$.
Use the training-sample moments for $\varphi_{\IF}$ and the fitted integrated moments for $\varphi_{\EIF}$ specified in \Cref{algo:cross.fitting,tab:EIF}.

\begin{proposition}[One-step correction estimator based on $\varphi_{\IF}$]\label{prop:limiting.np}
We assume the following.
\begin{enumerate}
\item[(C1)] (Boundedness) There exists $C<\infty$ such that $|Y|,|\mu|,|\EE[Y\mid\bm W_{-\calS}]|,|\widehat\mu_{-l}|,|\widehat{\mu}_{-\calS,-l}|\le C$ almost surely.
The regression bounds hold on the supports of both the true and fitted laws.
\item[(C2)] (Consistency) The regression errors $\|\widehat\mu_{-l}-\mu\|$, $\|\widehat{\mu}_{-\calS,-l}-\EE[Y\mid\bm W_{-\calS}]\|$ and $\|\widehat{\PP}_{\bm W,-l}-\PP_{\bm W}\|_{\mathrm{var}}$ converge to zero in probability.
\item[(C3)] (Rates) These three errors are $o_\PP(n^{-1/4})$.
\end{enumerate}
Under C1, the empirical training mean and variance automatically have $O_\PP(n^{-1/2})=o_\PP(n^{-1/4})$ errors.
Then the estimator in \Cref{algo:cross.fitting} satisfies
\[
\sqrt n(\widehat\xi_{\IF}-\xi)
=\frac1{\sqrt n}\sum_{i=1}^n\varphi_{\IF}(Y_i,\bm W_i;\PP)+o_\PP(1)
\ \stackrel{d}{\longrightarrow}\
\mathcal N\!\left(0,\var\{\varphi_{\IF}(Y,\bm W;\PP)\}\right).
\]
Thus it is asymptotically linear at $\PP$.
\end{proposition}

\begin{proposition}[One-step correction estimator based on $\varphi_{\EIF}$]\label{prop:limiting}
Assume \Cref{assu:independent}, the preceding sampling conventions, and C1 for $Y,\mu,\widehat\mu_{-l}$.
Replace the joint-law condition by $\|\widehat p_{W_k,-l}-p_{W_k}\|_{L_2(\PP_{W_k})}=o_\PP(n^{-1/4})$ for every $k$, and assume $\|\widehat\mu_{-l}-\mu\|=o_\PP(n^{-1/4})$ and $\|\widehat\upsilon_{-l}-\upsilon\|=o_\PP(n^{-1/4})$.
Use the fitted product law and the reduced regression defined in \Cref{tab:EIF}.
Further assume
\begin{enumerate}
\item[(C4)] (Positivity) For each coordinate, the true and fitted marginal densities are supported on the same fixed support and have a common dominating measure.
For some $\underline c>0$, $p_{W_k},\widehat p_{W_k,-l}\ge\underline c$ almost everywhere on that coordinate's support, for every $k,l$.
\item[(C5)] (Donsker class) There is a fixed $\PP$-Donsker class $\mathcal F$ such that, for each $k,l$, $\PP\{(p_{W_k}-\widehat p_{W_k,-l})^2\in\mathcal F\}\to1$.
\end{enumerate}
For this efficient construction, the reduced-regression and coherent fitted-variance quarter rates follow from these assumptions and need not be imposed separately.
Then
\[
\sqrt n(\widehat\xi_{\EIF}-\xi)
=\frac1{\sqrt n}\sum_{i=1}^n\varphi_{\EIF}(Y_i,\bm W_i;\PP)+o_\PP(1)
\ \stackrel{d}{\longrightarrow}\
\mathcal N\!\left(0,\var\{\varphi_{\EIF}(Y,\bm W;\PP)\}\right).
\]
\end{proposition}

Proofs of \Cref{prop:limiting.np,prop:limiting} are provided in \Cref{appe:sec:proof}.

\subsection{Non-degenerate null}\label{sec:inference.non.degenerate.null}

For $\var(\varphi) > 0$, where $\varphi$ may denote either $\varphi_{\IF}$ or $\varphi_{\EIF}$, 
\Cref{prop:limiting.np,prop:limiting} imply that we can construct the confidence interval for any significance level $\alpha \in (0,1)$ as
\begin{align}\label{eq:confidence.interval}
    \Big[\,\widehat{\xi}-\Phi^{-1}(1-\alpha/2)\frac{\widehat{\sigma}}{\sqrt{n}},\ 
          \widehat{\xi}+\Phi^{-1}(1-\alpha/2)\frac{\widehat{\sigma}}{\sqrt{n}}\,\Big].
\end{align}

Here $\Phi$ denotes the standard normal distribution function, and $\widehat\sigma$ is the square root of the consistent plug-in variance estimator in \Cref{algo:cross.fitting}.

\begin{corollary}\label{coro:coverage}
Under the assumptions of \Cref{prop:limiting}, if
$\var\{\varphi_{\EIF}(Y,\bm W;\PP)\}>0$, the variance estimator in
\Cref{algo:cross.fitting} is consistent for this variance, and the confidence
interval in \eqref{eq:confidence.interval} has coverage tending to $1-\alpha$.
The same conclusion holds for $\widehat\xi_{\IF}$ under
\Cref{prop:limiting.np} and positive variance of its true influence function.
\end{corollary}

\subsection{Degenerate null}\label{sec:inference.degenerate.null}

When $\xi(\vee_{k\in\calS}W_k)=0$ for some $\mathcal{S} \subseteq [K]$, the leading term in the expansion of the estimation error vanishes (\Cref{prop:limiting.np,prop:limiting}), and the asymptotic distribution cannot be directly used to construct confidence interval\footnote{
For interaction nulls, zero interaction explainability need not imply a zero unrestricted-model influence function $\varphi_{\IF}$.
Normal inference applies only when the chosen influence function has positive variance and the corresponding regularity conditions hold; the efficient influence function can still be degenerate.
Example: consider independent $W_1, W_2,E_Y \sim \mathcal{N}(0,1)$ and $Y = W_1 + W_2 + E_Y$, we have $\xi(W_1)$, $\xi(W_2)$, $\xi(W_1 \vee W_2) > 0$ while $\xi(W_1 \wedge W_2) = 0$.
Then $\varphi_{\IF} = {-2W_1W_2}/{\var(Y)}$ for $\xi(W_1 \wedge W_2)$, which does not vanish.
}.
Here, we exploit the structure of the null and the independence~\eqref{assu:independent} to introduce a randomization-based inference procedure.
Compared with existing approaches to null degeneracy (\Cref{sec:IF}), our method is finite-sample valid and avoids sample splitting or added noise (and preventing the associated loss of power).

As the first step, we establish an equivalent statement of the zero total explainability null.
\begin{proposition}\label{prop:sharp.null}
Assume \Cref{assu:independent}, the null $\xi(\vee_{k\in\calS}W_k) = 0$ is equivalent to
\[
Y(w_{\mathcal{S}}, w_{-\mathcal{S}}) = Y(w_{\mathcal{S}}', w_{-\mathcal{S}}), 
\]
almost surely over potential outcomes and assignment pairs $(\bm W_{\calS},\bm W'_{\calS},\bm W_{-\calS})$.
For discrete assignments, this holds almost surely over potential outcomes at each pair with positive assignment probability.
\end{proposition}
\noindent
Proof of \Cref{prop:sharp.null} is provided in \Cref{appe:sec:proof}.
\Cref{prop:sharp.null} implies that, holding components $W_{-\mathcal{S}}$ fixed, changing $W_{\mathcal{S}}$ does not affect the outcome $Y$.
In finite samples, this yields, almost surely under supported assignments, the partially sharp null \parencite{fisher1918correlation,zhang_2023_randomization_test},
\begin{align}\label{prop:sharp.null.finite}
    Y_i(w_{\mathcal{S}}, w_{-\mathcal{S}}) = Y_i(w_{\mathcal{S}}', w_{-\mathcal{S}}), 
    \quad \forall i \in [n].
\end{align}
Under the null~\eqref{prop:sharp.null.finite}, for reassignments $W_{\calS,i}^*$ drawn from the specified assignment law, we can impute the corresponding potential outcome almost surely as $Y_i\big(W_{\calS,i}^*, W_{-\calS,i}\big) = Y_i$.

Randomization tests (\Cref{sec:randomization.test}) are standard for the partially sharp null~\eqref{prop:sharp.null.finite}.
By \Cref{prop:sharp.null}, we propose to adopt randomization tests for degenerate nulls.
Explicitly, we choose a computable test statistic that tends to be large under the alternative $\xi(\vee_{k\in\calS}W_k)>0$, such as an estimator of $\xi(\vee_{k\in\calS}W_k)$.
We then generate $W_{\calS,i}^*$ from the same distribution as $W_{\calS,i}$ conditional on $W_{-\calS,i}$, impute the potential outcomes using \eqref{prop:sharp.null.finite}, and recompute the test statistic based on $W_{\calS,i}^*$, the imputed outcomes, and compare the observed test statistic to those using $W_{\calS,i}^*$ to compute the p-value.

Note that standard randomization tests require knowing the distribution of $\bm{W}$ in order to generate $W_{\calS,i}^*$.
However, the joint distribution of $\bm W$ may not be available in our setting.
Fortunately, under the independence assumption~\eqref{assu:independent} and the i.i.d.\ nature of the observations, we have
\begin{align}\label{eq:permutation.distribution}
    \big(\bm W_{\mathcal{S},\rho(i)}, \bm W_{-\calS,i}\big)_{i=1}^n \overset{d}{=} 
    \big(\bm W_{\calS,i}, \bm W_{-\calS,i}\big)_{i=1}^n,
\end{align}
where $\rho$ denotes a permutation of the indices $i \in [n]$.
For the randomization test, draw permutations independently and uniformly from the full permutation group, including the identity, and permute $\bm W_{\calS,i}$ while keeping $\bm W_{-\calS,i}$ fixed.
Apply the same statistic and fitting rule to every dataset, with any auxiliary randomness fixed independently of assignments.
The weak inequality in the $p$-value treats ties conservatively.
We summarize the procedure in \Cref{algo:randomization.test}.
The result below establishes the validity of the randomization test, and the proof is provided in \Cref{appe:sec:proof}.
\begin{proposition}\label{prop:coverage.degenerate}
Assume~\Cref{assu:independent}, under the null $\xi(\vee_{k\in\calS}W_k)=0$, the randomization test in \Cref{algo:randomization.test} controls type~I error.
\end{proposition}

\subsection{Sequential procedure adapting to null degeneracy}\label{sec:sequential test}

Since it is unknown a priori whether $\xi(\vee_{k\in\calS}W_k)=0$, a procedure is needed to distinguish between the degenerate and non-degenerate cases in order to apply the suitable inference procedure above.
Explicitly, we propose to use the following sequential approach.
\begin{enumerate}
\item Run the randomization test~\Cref{algo:randomization.test} for the degenerate null as in \Cref{sec:inference.degenerate.null} at level $\alpha$.
If the null is not rejected, output $\{0\}$; otherwise, proceed to Step 2.
\item Compute~Eq.~\eqref{eq:confidence.interval} in \Cref{sec:inference.non.degenerate.null} to construct a $(1-\alpha)$ confidence interval and output.
\end{enumerate}
A detailed algorithm of the sequential procedure is provided in \Cref{algorithm:general}.
We emphasize that although two tests are conducted in the sequential procedure, there is no need to split the significance level $\alpha$.
At the null, noncoverage is contained in first-stage rejection; at a fixed positive alternative, failure to reach the second-stage interval has probability tending to zero.
These event bounds, rather than independence of the two stages, justify the following result; its proof is in \Cref{appe:sec:proof}.
\begin{proposition}\label{prop:coverage.sequential}
The sequential procedure has coverage at least $1-\alpha$ at $\xi(\vee_{k\in\calS}W_k)=0$ under \Cref{prop:coverage.degenerate}.
At each fixed $\xi(\vee_{k\in\calS}W_k)>0$, it has asymptotic coverage at least $1-\alpha$ if \Cref{coro:coverage} applies, including positive influence-function variance, and the randomization test has power tending to one.

Replacing the nonrejection output $\{0\}$ by $[0,\infty)$ removes the power requirement; the corresponding second-stage coverage assumptions remain necessary at positive alternatives.
\end{proposition}

\section{SIMULATION}\label{sec:simulation}

We consider the following data generating mechanism and estimand: independent treatments $W_k \sim \calN(0, 1)$, $K = 3$, $E_Y \sim \mathcal{N}(0,0.5)$. The response follows $Y = W_1^3+W_2^3+W_3^3+W_1W_2^2+\sigma \cdot W_1^2W_3+E_Y$. In these simulations, $\sigma$ is the interaction coefficient, labelled ``noise magnitude'' in the plots. The estimands of interest are $\xi(W_3)$, $\xi(W_1\vee W_3)$ and $\xi(W_1 \wedge W_3)$. 
We evaluate two different estimators based on influence functions: (1) one-step correction estimator accounting for the independence among $\bm W$ ($\varphi_{\EIF}$-based estimator); (2) one-step correction estimator ignoring the independence among $\bm W$ ($\varphi_{\IF}$-based estimator).

\begin{figure}[htbp]
    \centering

    \begin{minipage}{\linewidth}
        \centering
        \begin{minipage}{0.3\textwidth}
            \centering
            \includegraphics[clip, trim = 0cm 0cm 0cm 1cm, width=\linewidth]{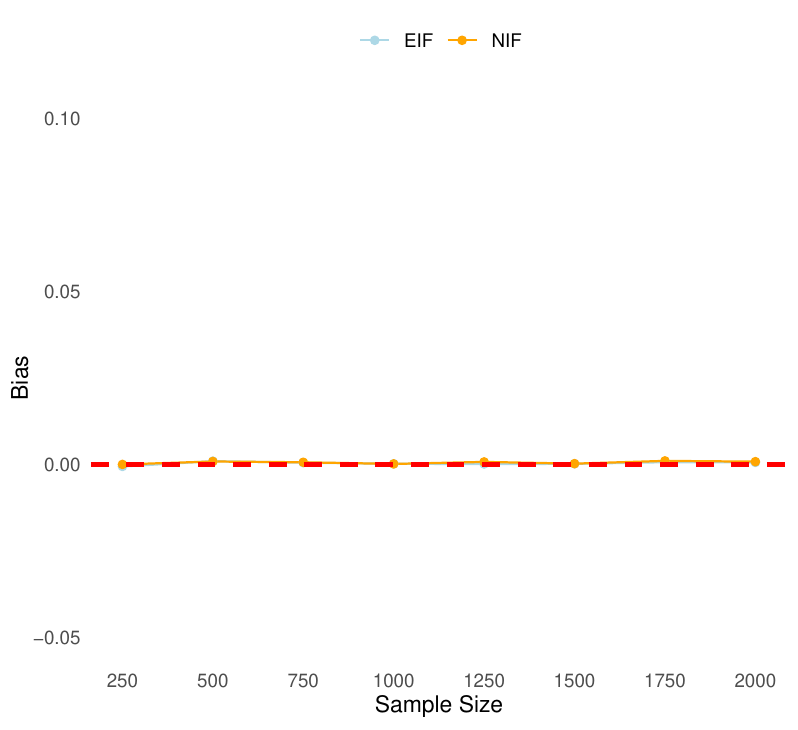}
        \end{minipage}
        \hfill
        \begin{minipage}{0.3\textwidth}
            \centering
            \includegraphics[clip, trim = 0cm 0cm 0cm 1cm, width=\linewidth]{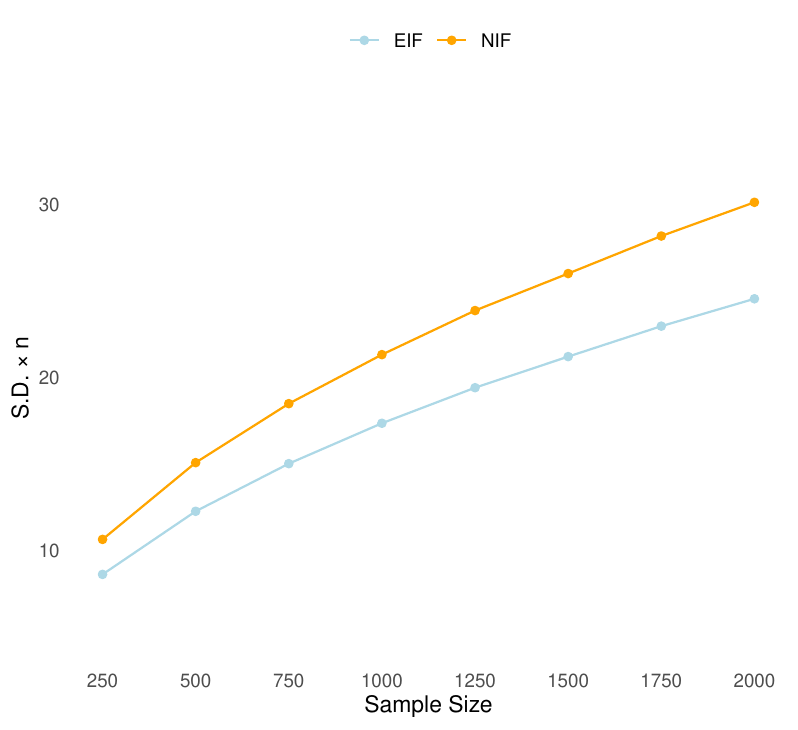}
        \end{minipage}
        \hfill
        \begin{minipage}{0.3\textwidth}
            \centering
            \includegraphics[clip, trim = 0cm 0cm 0cm 1cm, width=\linewidth]{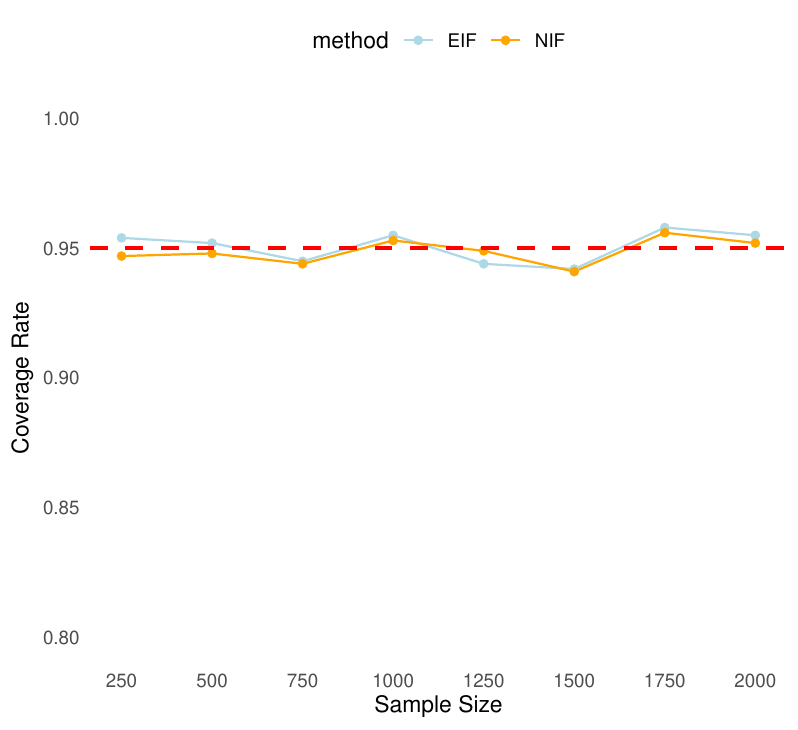}
        \end{minipage}
        \par\smallskip
        \centering
        \small{(a) $\xi(W_3)$; True nuisance functions; Varying sample size; $\sigma=1$}
    \end{minipage}
    \vspace{1em} %

    \begin{minipage}{\linewidth}
        \centering
        \begin{minipage}{0.3\textwidth}
            \centering
            \includegraphics[clip, trim = 0cm 0cm 0cm 1cm, width=\linewidth]{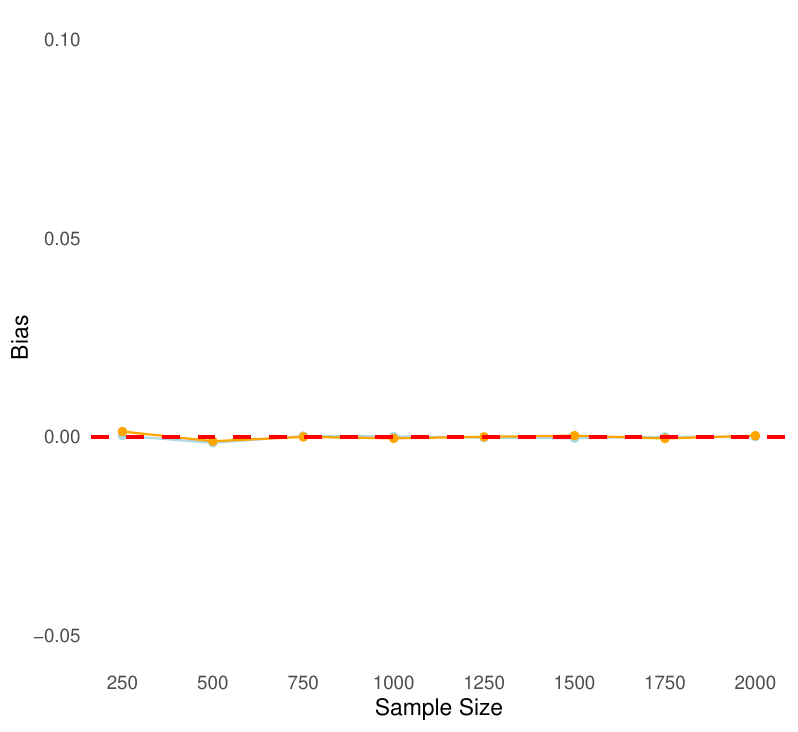}
        \end{minipage}
        \hfill
        \begin{minipage}{0.3\textwidth}
            \centering
            \includegraphics[clip, trim = 0cm 0cm 0cm 1cm, width=\linewidth]{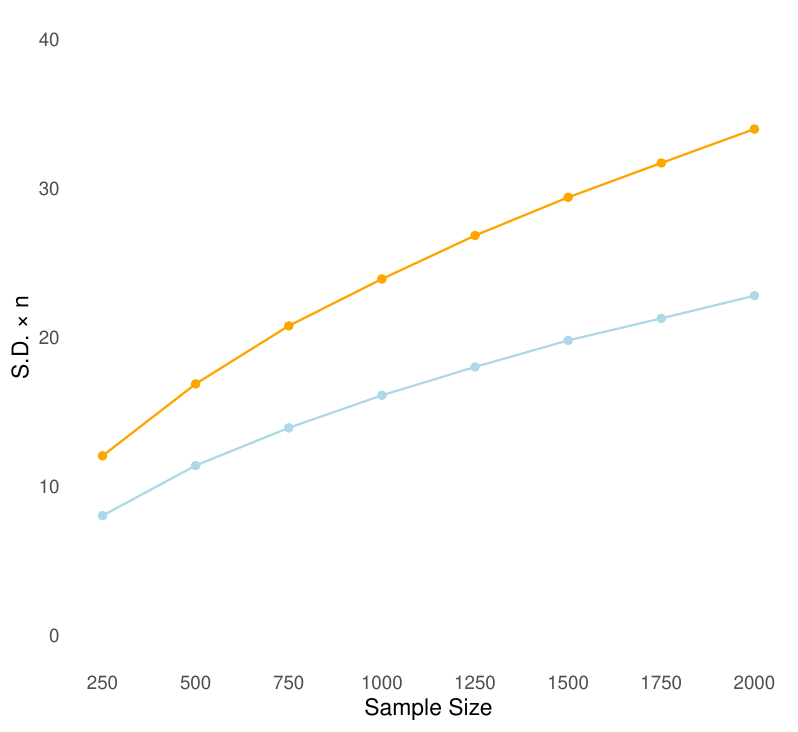}
        \end{minipage}
        \hfill
        \begin{minipage}{0.3\textwidth}
            \centering
            \includegraphics[clip, trim = 0cm 0cm 0cm 1cm, width=\linewidth]{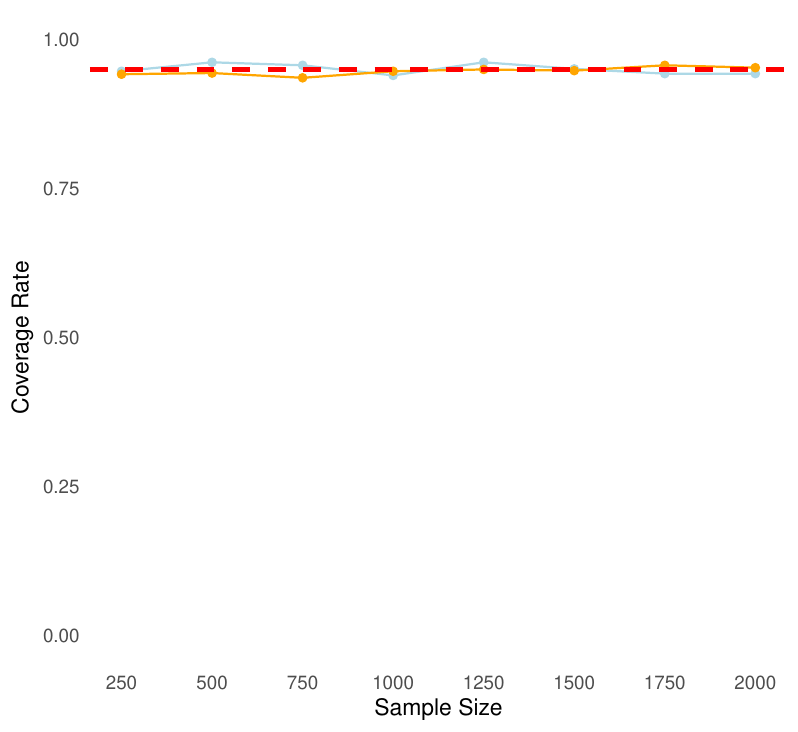}
        \end{minipage}
        \par\smallskip
        \centering
        \small{(b) $\xi(W_1 \vee W_3)$; True nuisance functions; Varying sample size; $\sigma=1$}
    \end{minipage}
    \vspace{1em}

    \begin{minipage}{\linewidth}
        \centering
        \begin{minipage}{0.3\textwidth}
            \centering
            \includegraphics[clip, trim = 0cm 0cm 0cm 1cm, width=\linewidth]{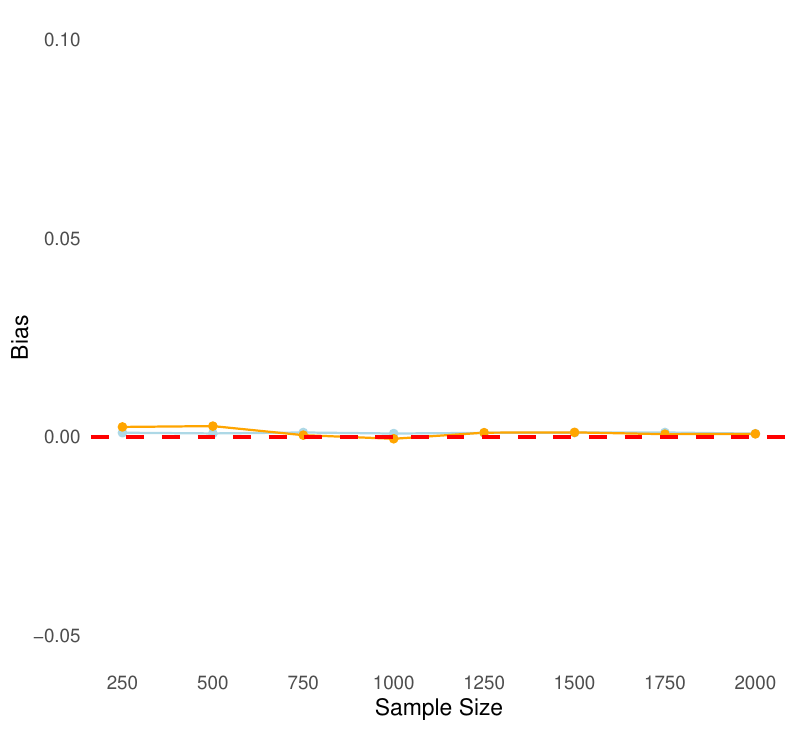}
        \end{minipage}
        \hfill
        \begin{minipage}{0.3\textwidth}
            \centering
            \includegraphics[clip, trim = 0cm 0cm 0cm 1cm, width=\linewidth]{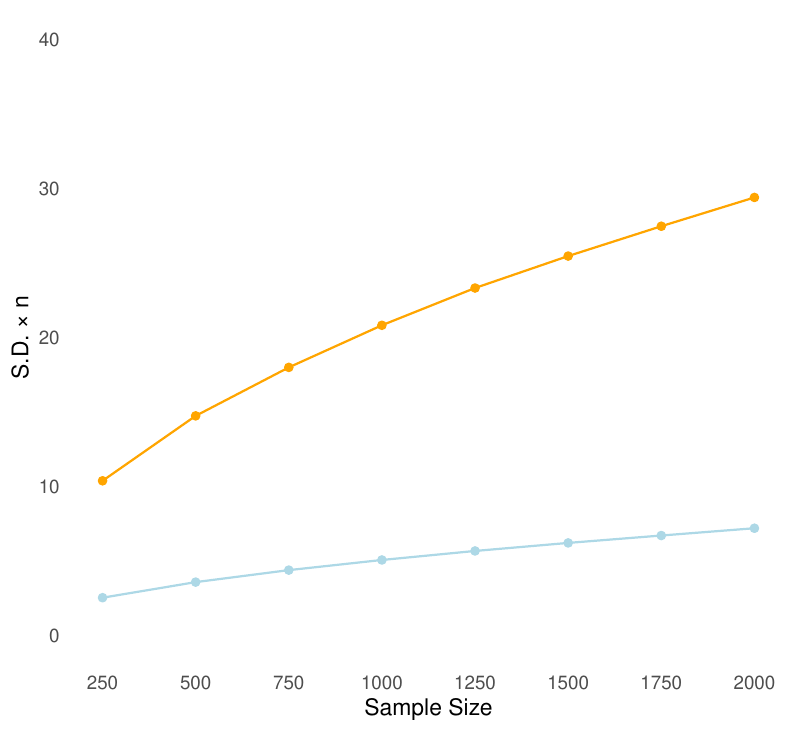}
        \end{minipage}
        \hfill
        \begin{minipage}{0.3\textwidth}
            \centering
            \includegraphics[clip, trim = 0cm 0cm 0cm 1cm, width=\linewidth]{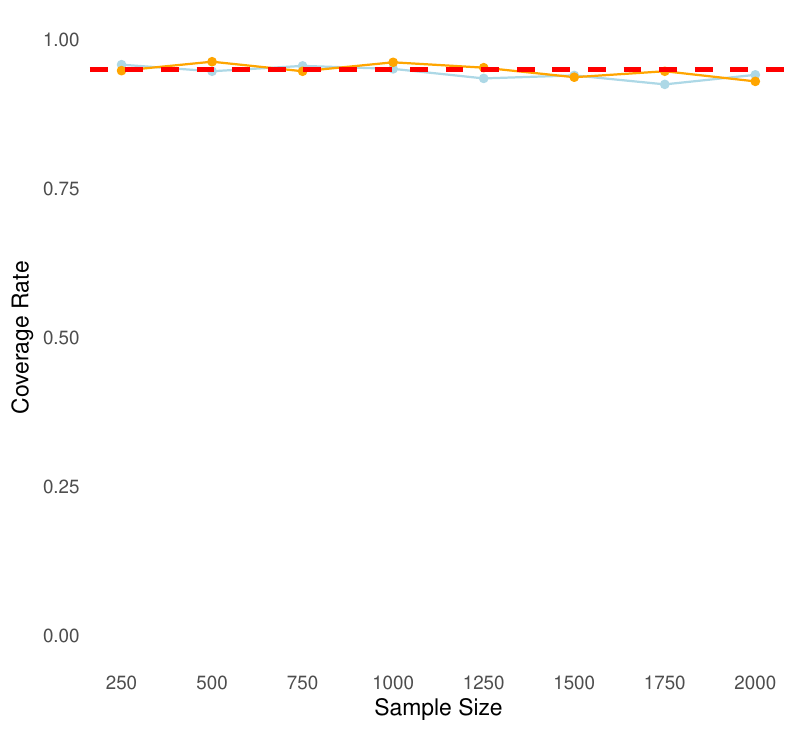}
        \end{minipage}
        \par\smallskip
        \centering
        \small{(c) $\xi(W_1 \wedge W_3)$; True nuisance functions; Varying sample size; $\sigma=1$}
    \end{minipage}
    \caption{Comparison of bias (left panel), estimated standard deviation times sample size (mid panel), and coverage rate (right panel; significance level $\alpha = 0.05$).
    $\varphi_{\EIF}$-based method in blue, $\varphi_{\IF}$-based method in gold.
    True nuisance functions are used. Results aggregated over $1000$ trials.}
    \label{fig:simulation1}
\end{figure}

To start, we implement the two methods using the true nuisance functions. 
In \Cref{fig:simulation1}, we vary the sample size. 
Both methods achieve exact coverage.
The $\varphi_{\EIF}$-based method exhibits lower variance, validating that exploiting independence can improve efficiency.
We provide the results with varying $\sigma$ in \Cref{app:simulation.noise.magnitude}.

\begin{figure}[htbp]
    \centering

    \begin{minipage}{\linewidth}
        \centering
        \begin{minipage}{0.32\textwidth}
            \centering
            \includegraphics[clip, trim = 0cm 0cm 0cm 1cm, width=\linewidth]{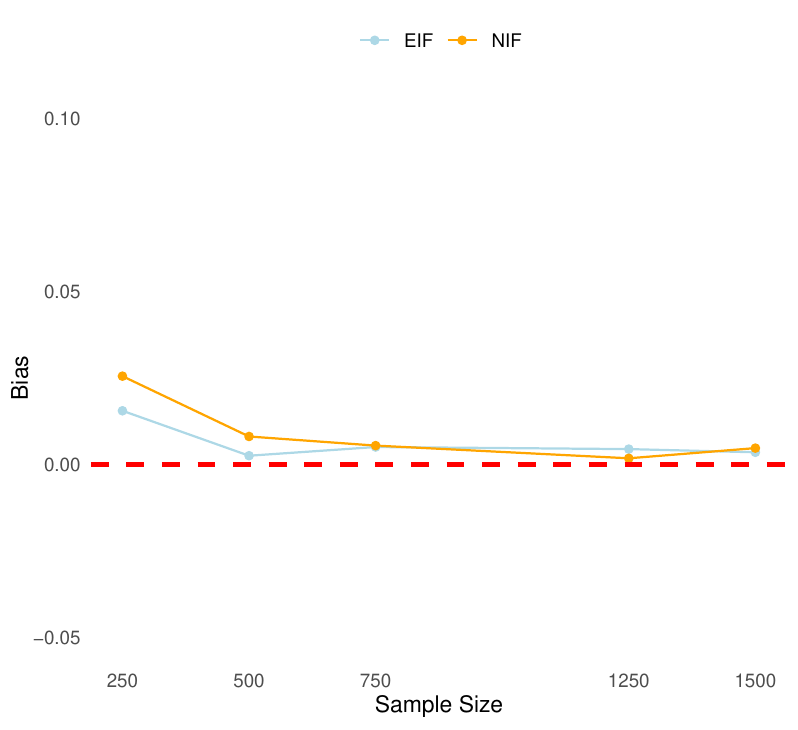}
        \end{minipage}
        \hfill
        \begin{minipage}{0.32\textwidth}
            \centering
            \includegraphics[clip, trim = 0cm 0cm 0cm 1cm, width=\linewidth]{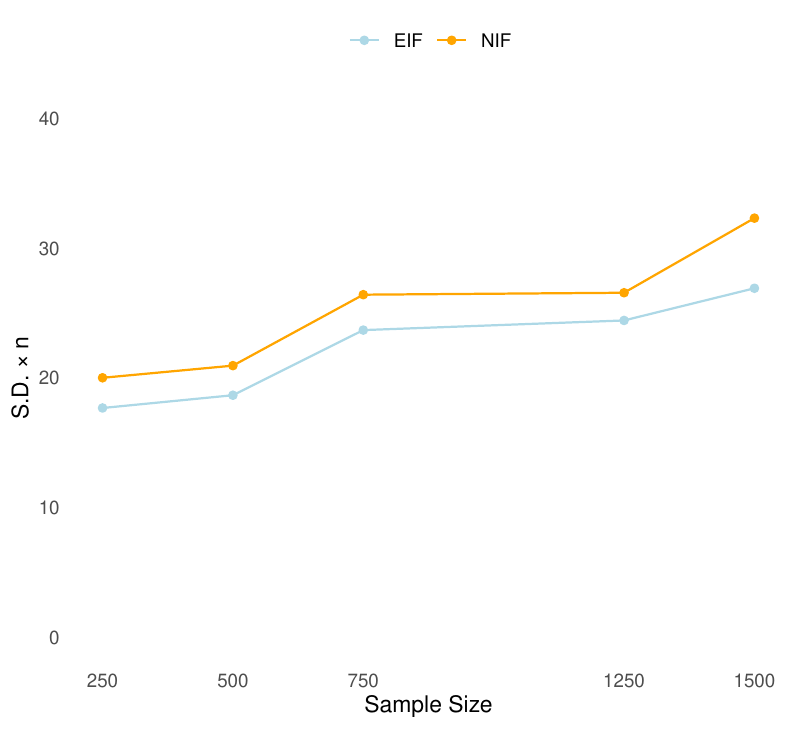}
        \end{minipage}
        \hfill
        \begin{minipage}{0.32\textwidth}
            \centering
            \includegraphics[clip, trim = 0cm 0cm 0cm 1cm, width=\linewidth]{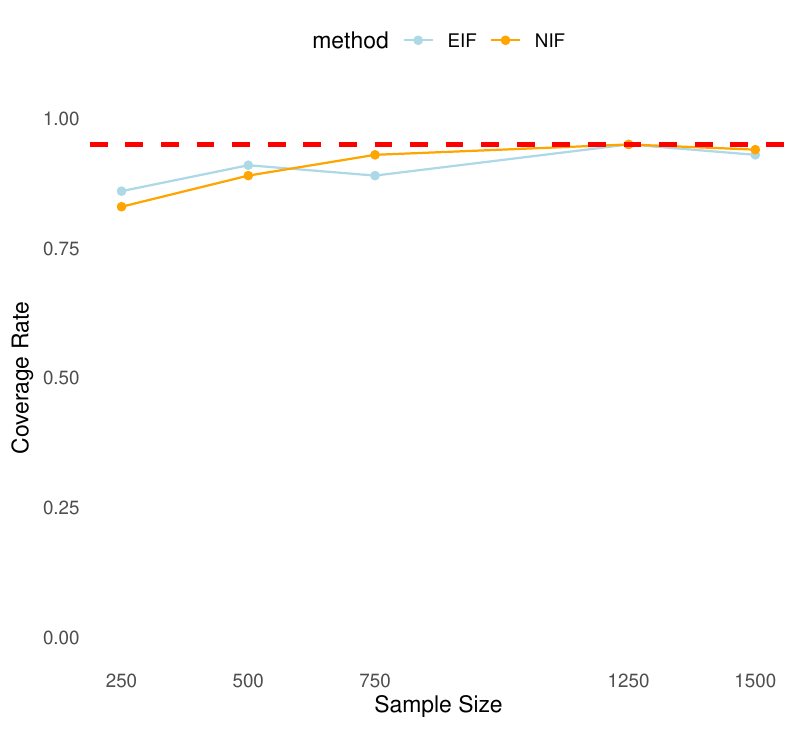}
        \end{minipage}
        \par\smallskip
        \centering
        \small{ $\xi(W_3)$; Estimated nuisance functions; Varying sample size; $\sigma=1$}
    \end{minipage}
    \vspace{1em} %

    \caption{Comparison using estimated nuisance functions (details in the caption of \Cref{fig:simulation1}). Results aggregated over $100$ trials.}\label{fig:simulation.estimated.nusiance}
\end{figure}

In \Cref{fig:simulation.estimated.nusiance}, we implement the two methods with nuisance functions estimated by the highly adaptive LASSO (HAL) and super learner (package \textsc{hal9001} and \textsc{Super Learner})
\parencite{vanderlaan2007super,hal9001}.
Note that $\varphi_{\EIF}$ involves more nuisance functions than $\varphi_{\IF}$, for example, $\varphi_{\EIF}$ additionally requires $\mathbb{E}\left[\mathbb{E}[Y\mid \bm W]\cdot\mathbb{E}[Y\mid \bm W_{-\calS}]\mid W_k\right]$ for all $k \in [K]$.
With estimated nuisance functions, the $\varphi_{\EIF}$-based method still exhibits variance reduction, although these gains are less substantial than those obtained with true nuisance functions.
In \Cref{app:estimated}, we provide the comparison for interaction
explainabilities, where the nuisance estimation is even more complex and the $\varphi_{\IF}$-based estimator appears competitive. We noticed that the efficient influence function requires more nuisance components, which would result in unstable finite-sample performance (especially under coverage) of the one-step bias-correcting estimator that could offset theoretical efficiency gains. We recommend using the nonparametric influence function (NIF, or IF mentioned in our paper) based one-step bias-correcting estimator in finite samples, especially when the practitioners are not so confident about their estimated nuisances, since it only contains two nuisances. As shown in \Cref{app:estimated}, the performance of the IF-based estimator is more robust compared to the $\varphi_{\EIF}$-based estimator in finite samples using the estimated nuisances. The $\varphi_{\EIF}$-based estimator is preferred when the nuisances are well-estimated to gain a narrower confidence interval.

\section{REAL DATA ANALYSIS}\label{sec:real.data}

Conjoint analysis is a factorial survey-based experiment in which participants choose between pairs of candidate profiles with several randomly selected characteristics. One of the most cited applications of conjoint analysis is \cite{hainmueller2015conjoint}, who examined the role of immigrants' attributes in shaping support for
their admission to the United States. In the experiment, $1396$ respondents were asked to act as immigration officials
and to decide which of a pair of immigrants they would choose for admission. Respondents evaluated a total of five pairings. The attributes were then randomly chosen to form immigrant profiles. 

We apply the explainability metric and the associated estimation and inference tools proposed to this real dataset. 
We use the $\varphi_{\IF}$-based estimator as it requires a simpler set of nuisance functions to estimate\footnote{Its interaction influence function need not vanish at the null, but normal inference still requires positive influence-function variance.}. 
For the nulls of zero total explainability, we do not rule out the degenerate null issue, and thus apply the sequential procedure in \Cref{sec:sequential test}.
Since the data have a cluster structure, we treat the respondent as the independent unit; in this analysis, $n$ is the number of respondents, $m_i$ is the number of rows for respondent $i$, and $N=\sum_{i=1}^n m_i$ is the total number of rows.
Let $O_{ir}=(Y_{ir},\bm W_{ir})$ denote row $r$ of respondent $i$, and let $O_i=\{O_{ir}\}^{m_i}_{r=1}$ denote all observations from respondent $i$.
We assume $\{O_i\}^n_{i=1}$ are i.i.d. across $i$, and allow arbitrary dependence within $i$.
We compute fitted row-level influence-function contributions $\phi_{ir}$ (with cross-fitting done by respondent folds), and aggregate to respondent-level scores $U_i=\sum_{r=1}^{m_i}\phi_{ir}$.
The cluster-robust variance for $\widehat\xi=\xi(\widehat{\PP})+N^{-1}\sum_{i,r}\phi_{ir}$ is then
\[
\widehat{\mathrm{Var}}(\widehat\xi) = \frac{1}{N^2}\cdot\frac{n}{n-1}\sum_{i=1}^n (U_i-\overline U)^2,
\]
where $\overline U=\frac{1}{n}\sum_{i=1}^n U_i$, and we use the $1-\alpha/2$ quantile $t_{n-1,\,1-\alpha/2}$ of Student's $t$ distribution with $n-1$ degrees of freedom for confidence intervals.
This is an asymptotic cluster-robust approximation, requiring a cluster-level asymptotic linearization and finite positive variance of the limiting respondent score.
Here each respondent contributes ten profile rows; with variable or response-dependent cluster sizes, row normalization instead requires the corresponding row-weighted target and centering and moment conditions.
Any randomization test must preserve the actual assignment mechanism.
Under clustering, we resample/permutate $\bm W_{\mathcal{S}}$ within the respondent (and within blocks/strata if present), preserving task structure and any design constraints (e.g., paired-profile choice sets, restricted attribute combinations).
We then recompute the test statistic on each resampled dataset using the same respondent-level cross-fitting and cluster-robust studentization.
Finite-sample randomization validity requires that the resampling scheme match the true constrained assignment law and that the observed and all resampled statistics be jointly exchangeable.
Because the full assignment mechanism is unavailable, design-exact validity of the reported empirical permutation results is conditional on this requirement.
\Cref{tab:total-explainability} shows that the total explainabilities of \textit{Language}, \textit{Job Experience}, \textit{Job Plan}, \textit{Prior Trips to the U.S.}, \textit{Education \& Profession}, and \textit{Origin \& Application} reason are all significantly positive, except \textit{Gender}. 
Among these factors, \textit{Job Plan}, \textit{Education \& Profession} exhibit the two highest explainabilities, suggesting an encouraging pattern that immigrants' skills and qualifications are seriously evaluated.
None of the reported interaction-explainability intervals in \Cref{tab:interaction-explainability} excludes zero, including that for \textit{Job Plan} and \textit{Job Experience}. The separate CRT result in \Cref{app:rda} is evidence under that test's null and should not be read as significance of the explainability interval.

\begin{table}[t] 
\centering 
\caption{{Total explainability}}
\label{tab:total-explainability} 
\vspace{-1em} 
\begin{tabular}{l r l} 
\hline Variable & Total explainability & 95\% CI \\ 
\hline Gender & 0.0012 & [$-0.0003$, $0.0026$] \\
Language & 0.0319 & [$0.0240$, $0.0398$] \\ 
Job experience & 0.0164 & [$0.0109$, $0.0219$] \\ 
Job plans & 0.0884 & [$0.0751$, $0.1016$] \\ 
Prior trips to U.S. & 0.0366 & [$0.0275$, $0.0457$] \\ 
Education \& Profession & 0.0579 & [$0.0471$, $0.0687$] \\ 
Origin \& Application reason & 0.0155 & [$0.0101$, $0.0210$] \\ 
\hline 
\end{tabular}
\\\footnotesize 
\emph{Note:} The permutation test for Gender returns a p-value of 1, failing to reject the null hypothesis. Its displayed Wald interval is supplementary; the sequential procedure returns $\{0\}$ on nonrejection (or $[0,\infty)$ under the fallback). 
\end{table} 
\begin{table}[t] \centering \caption{{Interaction explainability}} \label{tab:interaction-explainability} \vspace{-1em} 
\begin{tabular}{l l r l} \hline Factor 1 & Factor 2 & Interaction explainability & 95\% CI \\ \hline Job plan & Prior trips to U.S. & -0.0002& [$-0.0032$, $0.0029$] \\ Job plan & Job experience & 0.0014 & [$-0.0006$, $0.0034$] \\ Job plan & Language & 0.0010 & [$-0.0017$, $0.0038$] \\ \hline \end{tabular} \end{table}

\section{DISCUSSION}\label{sec:discussion}

Using semi-parametric efficiency theory, we develop one-step corrected estimators for causal ANOVA quantities, both with and without incorporating independence across factors.
The estimator that incorporates this structure achieves no larger asymptotic variance theoretically and when nuisance components can be accurately estimated, whereas the estimator that does not explicitly leverage independence has a simpler form and may exhibit greater robustness in small samples.
To address the issue of null degeneracy, we develop a randomization-based inference procedure and integrate it with the semi-parametric analysis.
Applied to a real immigration dataset, our methods identify multiple factors with nonzero explainability. The reported interval for the interaction between \textit{Job Plan} and \textit{Job Experience} includes zero; the separate CRT comparison in \Cref{app:rda} reports evidence under its own null.

We outline a few future directions.
When testing a set of factors and their interactions for nonzero explainability simultaneously, multiplicity arises, and it is of interest to explore proper adjustment.
    Causal ANOVA quantities exhibit a hierarchical structure (for example, if the total explainability of a factor is zero, then any interaction involving this factor must be zero), which can be exploited via closed testing or gatekeeping to improve power while maintaining error control. 
    In addition, in sequential decision problems (Markov decision processes), the Markov property yields a natural conditional independence. A direction for future research is to investigate the influence functions incorporating this conditional independence to reduce estimator's asymptotic variance.
   Moreover, when factors are dependent, Causal ANOVA quantities are generally not point-identified. A natural direction is to characterize the identified set and extend the semi-parametric analysis to the estimation and inference for this set.

\section*{Acknowledgment of AI use}
The authors used AI tools during the preparation of this manuscript for
language editing, literature exploration and reference cross-checking,
assistance with drafting and reviewing the exposition and mathematical
arguments, and assistance with proof formalization. All mathematical
results, proofs, interpretations, and cited references were independently
checked and verified by the authors, who take full responsibility for the
contents of the manuscript.

\printbibliography

@inproceedings{hoyer2008nonlinear,
  author = {Hoyer, Patrik O. and Janzing, Dominik and Mooij, Joris M. and Peters, Jonas and Sch{\"o}lkopf, Bernhard},
  title = {Nonlinear causal discovery with additive noise models},
  booktitle = {Advances in Neural Information Processing Systems},
  volume = {21},
  year = {2008},
  publisher = {Curran Associates, Inc.}
}

@article{peters2014causal,
  author = {Peters, Jonas and Mooij, Joris M. and Janzing, Dominik and Sch{\"o}lkopf, Bernhard},
  title = {Causal discovery with continuous additive noise models},
  journal = {Journal of Machine Learning Research},
  volume = {15},
  number = {1},
  pages = {2009--2053},
  year = {2014}
}

@article{vanderlaan2007super,
  title   = {Super Learner},
  author  = {van der Laan, Mark J. and Polley, Eric C. and Hubbard, Alan E.},
  journal = {Statistical Applications in Genetics and Molecular Biology},
  volume  = {6},
  number  = {1},
  pages   = {Article 25},
  year    = {2007}
}

@article{hal9001,
    title = {{hal9001}: Scalable highly adaptive lasso regression in
      {R}},
    author = {Nima S Hejazi and Jeremy R Coyle and Mark J {van der
      Laan}},
    year = {2020},
    journal = {Journal of Open Source Software},
    publisher = {The Open Journal}
  }

@article{zhang_2023_randomization_test,
	title = {What is a {Randomization} {Test}?},
	volume = {0},
	number = {0},
	journal = {Journal of the American Statistical Association},
	author = {Zhang, Yao and Zhao, Qingyuan},
	year = {2023},
	pages = {1--15},
}

@article{rubin1974estimating,
	title={Estimating causal effects of treatments in randomized and nonrandomized studies},
	author={Rubin, Donald B.},
	journal={Journal of Educational Psychology},
	volume={66},
	number={5},
	pages={688-701},
	year={1974},
	publisher={American Psychological Association}
}

@article{hooker2007generalized,
  title={Generalized functional {ANOVA} diagnostics for high-dimensional functions of dependent variables},
  author={Hooker, Giles},
  journal={Journal of Computational and Graphical Statistics},
  volume={16},
  number={3},
  pages={709--732},
  year={2007},
  publisher={Taylor \& Francis}
}

@article{sobol2001global,
  title={Global sensitivity indices for nonlinear mathematical models and their Monte Carlo estimates},
  author={Sobol', Ilya M.},
  journal={Mathematics and Computers in Simulation},
  volume={55},
  number={1-3},
  pages={271--280},
  year={2001},
  publisher={Elsevier}
}

@InProceedings{pmlr-v275-gao25a,
  title     = {Counterfactual explanability of black-box prediction models},
  author    = {Gao, Zijun and Zhao, Qingyuan},
  booktitle = {Proceedings of the Fourth Conference on Causal Learning and Reasoning},
  editor    = {Huang, Biwei and Drton, Mathias},
  series    = {Proceedings of Machine Learning Research},
  volume    = {275},
  pages     = {1174--1174},
  year      = {2025},
  month     = {07--09 May},
  publisher = {PMLR}
}

@book{vanderweele2015explanation,
  title={Explanation in Causal Inference: Methods for Mediation and Interaction},
  author={VanderWeele, Tyler J.},
  year={2015},
  publisher={Oxford University Press}
}

@book{pearl2009,
  title = {Causality: {{Models}}, Reasoning, and Inference},
  author = {Pearl, Judea},
  year = {2009},
  edition = {2},
  publisher = {Cambridge University Press},
  address = {New York}
}

@article{fisher1918correlation,
  author  = {Fisher, R. A.},
  title   = {The Correlation Between Relatives on the Supposition of Mendelian Inheritance},
  journal = {Transactions of the Royal Society of Edinburgh},
  year    = {1918},
  volume  = {52},
  pages   = {399--433}
}

@article{visscher2017tenyears,
  author  = {Visscher, Peter M. and Wray, Naomi R. and Zhang, Qian and Sklar, Pamela and McCarthy, Mark I. and Brown, Matthew A. and Yang, Jian},
  title   = {10 Years of GWAS Discovery: Biology, Function, and Translation},
  journal = {The American Journal of Human Genetics},
  year    = {2017},
  volume  = {101},
  number  = {1},
  pages   = {5--22}
}

@article{watanabe2019global,
  author  = {Watanabe, Kyoko and Stringer, Sven and Frei, Oleksandr and Umi{\'c}evi{\'c} Mirkov, Marija and de Leeuw, Christiaan A. and Polderman, Tinca J. C. and van der Sluis, Sophie and Andreassen, Ole A. and Neale, Benjamin M. and Posthuma, Danielle},
  title   = {A global view of pleiotropy and genetic architecture across complex traits},
  journal = {Nature Genetics},
  year    = {2019},
  volume  = {51},
  pages   = {1339--1348}
}

@article{sivakumaran2011abundant,
  author  = {Sivakumaran, Shonali and Agakov, Felix and Theodoratou, Evropi and Prendergast, James G. and Zgaga, Lina and Manolio, Teri and Rudan, Igor and McKeigue, Paul and Wilson, James F. and Campbell, Harry},
  title   = {Abundant pleiotropy in human complex diseases and traits},
  journal = {Proceedings of the National Academy of Sciences},
  year    = {2011},
  volume  = {108},
  number  = {21},
  pages   = {13924--13929}
}

@article{vanderweele2014attributing,
  title={Attributing effects to interactions},
  author={VanderWeele, Tyler J. and Tchetgen, Eric J. Tchetgen},
  journal={Epidemiology},
  volume={25},
  number={5},
  pages={711--722},
  year={2014}
}

@article{williamson2023general,
  title={A general framework for inference on algorithm-agnostic variable importance},
  author={Williamson, Brian D. and Gilbert, Peter B. and Simon, Noah R. and Carone, Marco},
  journal={Journal of the American Statistical Association},
  volume={118},
  number={543},
  pages={1645--1658},
  year={2023},
  publisher={Taylor \& Francis}
}

@article{PhipsonSmyth2010pvalue,
title = {Permutation P-values Should Never Be Zero: Calculating Exact P-values When Permutations Are Randomly Drawn},
author = {Phipson, Belinda and Smyth, Gordon K.},
volume = {9},
number = {1},
journal = {Statistical Applications in Genetics and Molecular Biology},
year = {2010},
publisher={ De Gruyter}
}

@article{kennedy2020ins,
author = {Kennedy, Edward H. and Balakrishnan, Sivaraman and G’Sell, Max},
title = {Sharp Instruments for Classifying Compliers and Generalizing Causal Effects},
volume = {48},
journal = {The Annals of Statistics},
number = {4},
publisher = {Institute of Mathematical Statistics},
pages = {2008 -- 2030},
year = {2020}
}

@article{hainmueller2015conjoint,
    author = {Hainmueller, Jens and Hopkins, Daniel J. and Yamamoto, Teppei},
    title = {Causal Inference in Conjoint Analysis: Understanding Multidimensional Choices via Stated Preference Experiments},
    journal = {Political Analysis},
    volume = {22},
    number = {1},
    pages = {1--30},
    year = {2015},
    publisher={Cambridge University Press}
}

@article{kennedy2024semiparametric,
  title={Semiparametric doubly robust targeted double machine learning: a review},
  author={Kennedy, Edward H.},
  journal={Handbook of Statistical Methods for Precision Medicine},
  pages={207--236},
  year={2024},
  publisher={Chapman and Hall/CRC}
}

@article{Williamson2021,
author = {Williamson, Brian D. and Gilbert, Peter B. and Carone, Marco and Simon, Noah},
title = {Nonparametric variable importance assessment using machine learning techniques},
journal = {Biometrics},
volume = {77},
number = {1},
pages = {9-22},
year = {2021}
}

@book{tsiatis2006semiparametric,
  title={Semiparametric Theory and Missing Data},
  author={Tsiatis, Anastasios A.},
  volume={4},
  year={2006},
  publisher={Springer}
}

@article{ham2024using,
  title={Using machine learning to test causal hypotheses in conjoint analysis},
  author={Ham, Dae Woong and Imai, Kosuke and Janson, Lucas},
  journal={Political Analysis},
  volume={32},
  number={3},
  pages={329--344},
  year={2024},
  publisher={Cambridge University Press}
}

@article{verdinelli2024feature,
  title={Feature importance: A closer look at shapley values and loco},
  author={Verdinelli, Isabella and Wasserman, Larry},
  journal={Statistical Science},
  volume={39},
  number={4},
  pages={623--636},
  year={2024},
  publisher={Institute of Mathematical Statistics}
}

@article{dai2022significance,
  title={Significance tests of feature relevance for a black-box learner},
  author={Dai, Ben and Shen, Xiaotong and Pan, Wei},
  journal={IEEE transactions on neural networks and learning systems},
  volume={35},
  number={2},
  pages={1898--1911},
  year={2022},
  publisher={IEEE}
}

@article{hudson2023nonparametric,
  title={Nonparametric inference on non-negative dissimilarity measures at the boundary of the parameter space},
  author={Hudson, Aaron},
  journal={arXiv preprint arXiv:2306.07492},
  year={2023}
}

@book{bickel1993efficient,
  title={Efficient and Adaptive Estimation for Semiparametric Models},
  author={Bickel, Peter J. and Ritov, Ya’acov and Klaassen, Chris A. J. and Wellner, Jon A.},
  volume={4},
  year={1993},
  publisher={Springer}
}

@incollection{pfanzagl1990estimation,
  title={Estimation in semiparametric models},
  author={Pfanzagl, Johann},
  booktitle={Estimation in Semiparametric Models: Some Recent Developments},
  pages={17--22},
  year={1990},
  publisher={Springer}
}

@article{schick1986asymptotically,
  title={On asymptotically efficient estimation in semiparametric models},
  author={Schick, Anton},
  journal={The Annals of Statistics},
  pages={1139--1151},
  year={1986},
  publisher={JSTOR}
}

@article{klaassen1987consistent,
  title={Consistent estimation of the influence function of locally asymptotically linear estimators},
  author={Klaassen, Chris AJ},
  journal={The Annals of Statistics},
  volume={15},
  number={4},
  pages={1548--1562},
  year={1987},
  publisher={Institute of Mathematical Statistics}
}

@article{van1991differentiable,
  title={On differentiable functionals},
  author={{van der Vaart}, Aad},
  journal={The Annals of Statistics},
  pages={178--204},
  year={1991},
  publisher={JSTOR}
}

@article{averbukh1967theory,
  title={The theory of differentiation in linear topological spaces},
  author={Averbukh, Vladimir I. and Smolyanov, Oleg Georgievich},
  journal={Russian Mathematical Surveys},
  volume={22},
  number={6},
  pages={201},
  year={1967},
  publisher={IOP Publishing}
}

\clearpage

\begin{center}
{\LARGE\bfseries Appendix for ``Estimation and Inference for Causal Explainability''}
\end{center}

\vspace{1cm}
\paragraph{Organization.} 
In \Cref{sec:estimation.interaction.explainability}, we provide additional definitions of causal
ANOVA quantities and discuss their hierarchical structure.
In \Cref{covariates}, we extend the independence condition~\eqref{assu:independent} to conditional independence.
In \Cref{appe:sec:algorithm}, we present additional algorithms.
In \Cref{appe:sec:lemma}, we collect supplementary lemmas and corollaries.
In \Cref{appe:sec:proof}, we provide detailed proofs.
In \Cref{appe:sec:table}, we give additional tables.  
In \Cref{app:simulation}, we report additional empirical studies.

\section{Method extension}\label{appe:sec:causal.ANOVA}

\subsection{Interaction explainability}\label{sec:estimation.interaction.explainability}

We define the interaction effect of a set of treatments iteratively.

\begin{definition}[Interaction effect]\label{defi:interaction.effect}
    For $\bm w$, $\bm w' \in \RR^K$, let $I_{\emptyset, \bm w'} (\bm w) := Y(\bm w')$, $I_{k, \bm w'} (\bm w) :=  Y(\bm w_{k}, \bm w'_{-k}) -  Y(\bm w')$ for $k \in [K]$. For $\calS \subseteq [K]$,
\begin{align*}
    I_{\calS, \bm w'} (\bm w)
    &:= Y(\bm w_{\calS}, \bm w'_{-\calS}) - \sum_{\calS' \subsetneq
      \calS} I_{\calS',\bm w'}(\bm w).
\end{align*}
\end{definition}
\noindent Equivalently, the interaction term can be represented as a linear combination of the potential outcomes evaluated at a combination of $\bm w_{\calS'}, \bm w'_{-\calS'}$ for $\calS' \subseteq \calS$ with coefficients in $\{1,-1\}$,
\begin{align*}
   I_{\calS, \bm w'} (\bm w)
   = \sum_{\calS' \subseteq \calS} (-1)^{|\calS - \calS'|} Y(\bm w_{\calS'}, \bm w'_{-\calS'}).
\end{align*}
The interaction effect of $W_k$ and $W_{k'}$ takes the form
$I_{\{k,k'\}, \bm w'} (\bm w) = Y(w_k, w_{k'}, \bm w'_{-\{k,k'\}}) - Y(w_k, \bm w'_{-k}) - Y(w_{k'}, \bm w'_{-k'}) + Y(\bm w')$.
For $K = 2$ and binary $W_1$, $W_2$, \cite{vanderweele2014attributing} uses $I_{\{1,2\}, [0,0]} ([1,1])$ as the interactive effect of $W_1$ and $W_2$.

The anchored contrast holds treatments outside $\calS$ at the common anchor $\bm W'_{-\calS}$. Under the product assignment law, replacing this common complement by $\bm W_{-\calS}$ gives a local contrast with the same variance, by exchangeability of the independent copies. This variance equivalence permits computing interaction explainability using a copy only of the treatments in $\calS$ (``self-sufficiency''); it is not a pointwise identity of contrasts. The common full anchor is retained in \eqref{eq:anchor.decomposition}.

We define the explainability of the interaction of a set of treatments $W_\calS$, $\calS \subseteq [K]$.
\begin{definition}
    Let $\bm W'$ be an independent copy of $\bm W$.
    For $\calS \subseteq [K]$,
    \begin{align}\label{eq:interaction.variance}
        \xi(\wedge_{k \in \calS} W_{k}) = \frac{\var\left(I_{\calS, \bm W'} (\bm W) \right)}{2^{|\calS|}\var\left(Y\right)}.
    \end{align}
\end{definition}

The difference $Y - Y(\bm W')$ can be decomposed into the sum of interaction effects with $\calS \neq \emptyset$,
\begin{align}\label{eq:anchor.decomposition}
    Y - Y(\bm W') =
\sum_{\emptyset \neq \calS \subseteq [K]} I_{\calS, \bm W'}(\bm W).
\end{align}
Note that $I_{\calS, \bm W'}(\bm W)$, $I_{\calS', \bm W'}(\bm W)$ are correlated, therefore $\var(Y - Y(\bm W'))$ typically does not equal the sum of the variance of $I_{\calS, \bm W'}(\bm W)$.
However, we show for independent treatments, $\var(Y - Y(\bm W'))$ can be represented as a linear combination of the variance of $I_{\calS, \bm W'}(\bm W)$.

\begin{proposition}[Anchored variance decomposition]\label{prop:variance.decomposition}
Under \Cref{assu:independent}, for every nonempty $\calS\subseteq[K]$,
\begin{align}\label{eq:inclusion.exclusion}
        &\xi(\vee_{k \in \calS} W_{k})
       = \sum_{\emptyset \neq \calS' \subseteq \calS}
       \left(-1\right)^{|\calS'|-1} \xi(\wedge_{k \in \calS'} W_{k}),\\
       &\xi(\wedge_{k \in \calS} W_{k})
       = \sum_{\emptyset \neq \calS' \subseteq \calS}
       \left(-1\right)^{|\calS'|-1} \xi(\vee_{k \in \calS'} W_{k}).
    \end{align}
\end{proposition}

The proof is provided in \textit{Appendix A.1.} in \cite{pmlr-v275-gao25a}.

\Cref{prop:variance.decomposition} indicates the total explainability and the explainability of interactions admits a relationship resembling the inclusion-exclusion principle in its structural form. Based on it, the influence functions for the explainability of interactions among treatments can be derived from the inclusion-exclusion principle and the linearity property of influence functions (\Cref{fig:influence.function.dependence}).

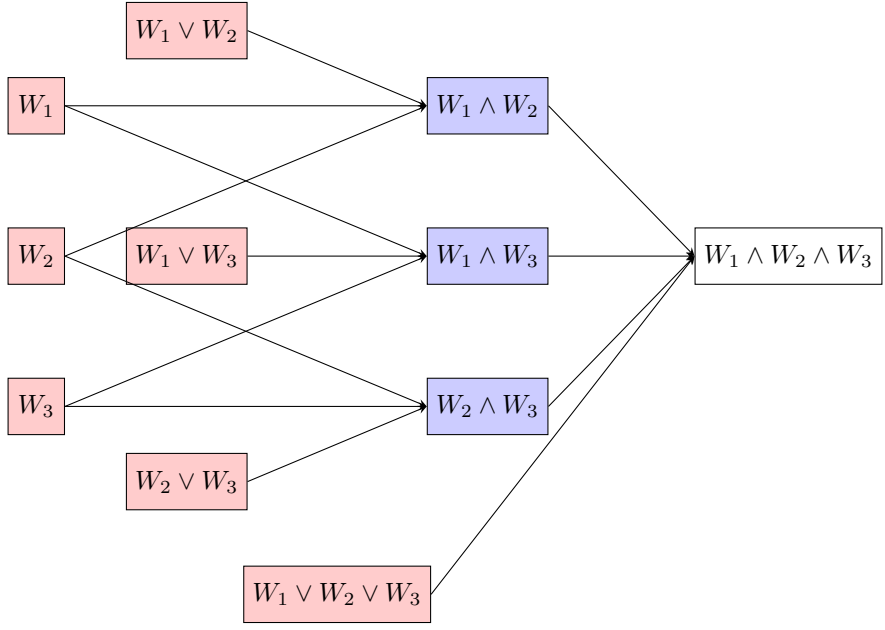
\begin{figure}[tbp]
    \centering
    \begin{tikzpicture}[>=stealth,
    node distance=1.5cm,
    every node/.style={rectangle, draw, minimum size=0.75cm}]

        \node (w1) at (-2,-1) [fill=red!20] {$W_1$};
        \node (w2) at (-2,-3) [fill=red!20] {$W_2$};
        \node (w3) at (-2,-5) [fill=red!20] {$W_3$};
        \node (w1v2) at (0,0) [fill=red!20] {$W_1 \vee W_2$};
        \node (w1v3) at (0,-3) [fill=red!20] {$W_1 \vee W_3$};
        \node (w2v3) at (0,-6) [fill=red!20] {$W_2 \vee W_3$};
        \node (w1v2v3) at (2,-7.5) [fill=red!20] {$W_1 \vee W_2 \vee W_3$};

        \node (w12) at (4,-1) [fill=blue!20] {$W_1 \wedge W_2$};
        \node (w13) at (4,-3) [fill=blue!20] {$W_1 \wedge W_3$};
        \node (w23) at (4,-5) [fill=blue!20] {$W_2 \wedge W_3$};

        \node (w123) at (8,-3) {$W_1 \wedge W_2 \wedge W_3$};

        \draw[->] (w1.east) -- (w12.west);
        \draw[->] (w1.east) -- (w13.west);
        \draw[->] (w1v2.east) -- (w12.west);
        \draw[->] (w2.east) -- (w12.west);
        \draw[->] (w2.east) -- (w23.west);
        \draw[->] (w2v3.east) -- (w23.west);
        \draw[->] (w3.east) -- (w13.west);
        \draw[->] (w3.east) -- (w23.west);
        \draw[->] (w1v3.east) -- (w13.west);

        \draw[->] (w12.east) -- (w123.west);
        \draw[->] (w13.east) -- (w123.west);
        \draw[->] (w23.east) -- (w123.west);
        \draw[->] (w1v2v3.east) -- (w123.west);
    \end{tikzpicture}
    \caption{Influence functions of total and interaction explainabilities: zeroth-order, first-order, second-order interaction terms are highlighted in red, blue, and white, respectively. The dependence is derived from the inclusion-exclusion principle and the linearity of influence functions.}
    \label{fig:influence.function.dependence}
\end{figure}
\setcounter{assumption}{2}

\subsection{Extension of NPSEMs}\label{sec:NPSEM.extension}

We next discuss three NPSEM models of increasing generality, with NPSEM-IE~\eqref{assu:additive.independent.error} as the starting point.

\paragraph{NPSEM-IE.} NPSEM-IE is a commonly used framework for causal discovery and inference \parencite{hoyer2008nonlinear, peters2014causal}. Identification of an unknown causal structure requires additional conditions. Here, the shared-error restriction identifies the causal explainability quantities under the stated assignment assumptions.

\paragraph{NPSEM with additive noise.} 
Consider the NPSEM model where the noise remains additive, but its distribution may depend on the treatment level 
, i.e., heteroskedastic errors. Additivity and heteroskedasticity alone do not imply co-monotonicity. If the counterfactuals are additionally nondecreasing functions of a common latent rank, the model satisfies the co-monotone property \parencite{pmlr-v275-gao25a} under which the joint distribution of all potential outcomes is identifiable. In principle, one could estimate this joint distribution, construct an oracle from it, and sample from this oracle to evaluate the Causal ANOVA quantities. However, learning this joint distribution is statistically challenging, and the inferential properties of such a procedure remain largely unexplored.

\paragraph{General NPSEM.} Finally, when the noise is not necessarily additive, the co-monotonicity structure may be violated, and without further restrictions the Causal ANOVA quantities are generally only partially identifiable. Characterizing and conducting inference for the resulting partial identification set is important but beyond the scope of this paper.

As an observable diagnostic, one can compare residual distributions $Y-\EE[Y\mid\bm W]$ across treatment levels. Approximate invariance supports the observed additive independent-error restriction, but does not verify the shared-error coupling of potential outcomes. Variation across levels motivates a more general error model.

\subsection{Covariates adjustment}\label{covariates}

If there exists a set of covariates $\bm X$ such that the treatments are unconfounded given $\bm X$, that is $ Y(\cdot)_i \independent \bm W_i \mid \bm X_i$, the above analysis can be extended by conditioning on $\bm X$.\begin{assumption}[Unconfounded treatments]\label{assu:unconfounded}
    The treatments are independent of the potential outcomes conditional on the covariates,
    \begin{align*}
        Y_i(\cdot) \independent \bm W_i \mid \bm X_i.
    \end{align*}
\end{assumption}

\begin{assumption}[Conditionally independent treatments]\label{assu:conindependent}
    Treatments $W_k$, $k \in [K]$ are mutually independent conditional on the covariates.
\end{assumption}

\begin{definition}[Total explainability with covariates]\label{defi:total.conindependent}
    Let $\bm W'$ have the conditional law of $\bm W$ given $\bm X$ and be independent of $(\bm W,Y(\cdot))$ given $\bm X$.
    For $\calS \subseteq [K]$,
    \begin{align}\label{eq:total.variance.conditional}
        \xi(\vee_{k \in \calS} W_k)
        := \frac{\mathbb{E}\left[ \mathsf{Var }\left( Y(\bm{W}) - Y(\bm{W}_\calS',\bm{W}_{-\calS}) \mid \bm{X} \right) \right]}{2\var\left(Y(\bm W)\right)}.
    \end{align}
\end{definition}

\begin{lemma}[Law of total conditional variance]\label{LTCV}
Let \((\Omega,\mathcal F,\mathbb P)\) be a probability space, let \(Y\in L_2(\mathbb P)\),
and let \(X\) and \(W\) be random elements. Define $\mathcal G := \sigma(X)$ and $\mathcal H := \mathcal G \vee \sigma(W)$. Then, \(\mathbb P\)-a.s.,
\begin{equation*}
\var(Y \mid \mathcal G)
= \mathbb{E}\!\left[\,\var(Y \mid \mathcal H)\,\middle|\, \mathcal G\right]
\;+\;
\var\!\left(\mathbb{E}[Y \mid \mathcal H] \,\middle|\, \mathcal G\right).
\end{equation*}
Equivalently,
\begin{equation*}
\var(Y \mid X)
= \mathbb{E}\!\left[\,\var(Y \mid X,W)\,\middle|\, X\right]
\;+\;
\var\!\left(\mathbb{E}[Y \mid X,W] \,\middle|\, X\right)
\quad \PP\text{-a.s.}
\end{equation*}
\end{lemma}

\begin{proof}
Let $\mathcal G=\sigma(X)$ and $\mathcal H=\mathcal G\vee\sigma(W)$ with $\mathcal G\subseteq\mathcal H$.
Recall the conditional variance definition
\[
\var(Y\mid\mathcal A):=\EE\!\left[Y^2\mid\mathcal A\right]-\Big(\EE[Y\mid\mathcal A]\Big)^2,
\ \mathcal A\subseteq\mathcal F,
\]
which is well-defined since $Y\in L_2(\PP)$.

Start from the right-hand side:
\begin{align*}
&\EE\!\left[\var(Y\mid\mathcal H)\mid\mathcal G\right]
+\var\!\left(\EE[Y\mid\mathcal H]\mid\mathcal G\right) \\
&=\EE\!\left[\EE[Y^2\mid\mathcal H]-\big(\EE[Y\mid\mathcal H]\big)^2 \,\middle|\,\mathcal G\right]
+\Big\{\EE\!\left[\big(\EE[Y\mid\mathcal H]\big)^2\mid\mathcal G\right]
-\big(\EE[\EE[Y\mid\mathcal H]\mid\mathcal G]\big)^2\Big\}.
\end{align*}
The middle terms cancel, leaving
\[
\EE\!\left[\EE[Y^2\mid\mathcal H]\mid\mathcal G\right]
-\big(\EE[\EE[Y\mid\mathcal H]\mid\mathcal G]\big)^2.
\]
Now apply the tower property twice (since $\mathcal G\subseteq\mathcal H$):
\[
\EE\!\left[\EE[Y^2\mid\mathcal H]\mid\mathcal G\right]=\EE[Y^2\mid\mathcal G],
\
\EE[\EE[Y\mid\mathcal H]\mid\mathcal G]=\EE[Y\mid\mathcal G].
\]
Therefore,
\[
\EE\!\left[\var(Y\mid\mathcal H)\mid\mathcal G\right]
+\var\!\left(\EE[Y\mid\mathcal H]\mid\mathcal G\right)
=\EE[Y^2\mid\mathcal G]-\big(\EE[Y\mid\mathcal G]\big)^2
=\var(Y\mid\mathcal G) \quad \PP\text{-a.s.}
\]

\end{proof}

\begin{lemma}\label{lemma: definition}
   Under \Cref{assu:unconfounded} and \Cref{assu:conindependent}, additionally assume $Y(\bm w)=f_Y(\bm w,\bm X)+E_Y$, with the same error across treatment levels and $E_Y\independent\bm W\mid\bm X$. Then \Cref{defi:total.conindependent} admits the equivalent form 
\begin{align}\label{eq:total.variance.conditional.regression}
    \xi(\vee_{k \in \calS} W_k)
    =  \frac{\mathbb{E}[\mathbb{E}^2\left[Y\mid \bm W, \bm X]\right]}{\var\left(Y\right)}-\frac{\mathbb{E}\left[\mathbb{E}^2[Y\mid \bm W_{-\mathcal{S}},\bm X]\right]}{\var\left(Y\right)}.
\end{align}
\end{lemma} 
Indeed, the shared error cancels from the treatment contrast.
Conditional independence of the two draws given $(\bm W_{-\calS},\bm X)$ gives $\EE[\var\{Y(\bm W)-Y(\bm W'_{\calS},\bm W_{-\calS})\mid\bm X\}]=2\EE[\var\{\EE[Y\mid\bm W,\bm X]\mid\bm W_{-\calS},\bm X\}]$.
Expanding the conditional variance yields the displayed formula.

\begin{definition} [Interaction explainability with covariates]
    Let $\bm W'$ have the conditional law of $\bm W$ given $\bm X$ and be independent of $(\bm W,Y(\cdot))$ given $\bm X$.
    For $\calS \subseteq [K]$,
    \begin{align}\label{eq:interaction.convariance}
        \xi(\wedge_{k \in \calS} W_{k}) = \frac{\EE\left[\var\left(I_{\calS, \bm W'} (\bm W) \mid \bm X\right)\right]}{2^{|\calS|}\var\left(Y(\bm W)\right)}.
    \end{align}
\end{definition}

\section{Additional algorithms}\label{appe:sec:algorithm}

\begin{algorithm}[H]
\caption{Randomization test for degenerate null}
\label{algo:randomization.test}
\begin{algorithmic}
\State\textbf{Input:} Data $\{\bm W_i,Y_i\}_{i=1}^n$, test statistic $T(Y,\bm W)$ (e.g. arbitrary estimator of $\xi(\vee_{k \in \calS} W_k)$), number of permutations $B$.
\For{$b = 1,\ldots,B$}

   \State{Draw $\rho_b$ independently and uniformly from all permutations of $[n]$, including the identity; set $\bm W_{\calS,i}^{(b)}=\bm W_{\calS,\rho_b(i)}$, keeping $\bm W_{-\calS,i}$ fixed.}

   \State Impute $Y_i(\bm W_i^{(b)}) = Y_i$ and compute $T(Y,\bm W^{(b)})$.

   \EndFor

   \State{\textbf{Output:} 
   $$P := \frac{1}{B+1}\Big[1+\sum_{b=1}^B \mathbbm{1}\{T(Y,\bm W^{(b)}) \geq T(Y,\bm W)\}\Big].$$   
   }
\end{algorithmic}
\end{algorithm}

\clearpage

\begin{algorithm}\label{algorithm:general}
\caption{General algorithm for independent treatments}

\begin{tikzpicture}[
    node distance=1.2cm,
    every node/.style={font=\footnotesize, align=center, inner sep=3pt},
    >=stealth
]

\tikzset{
    startstop/.style={rectangle, rounded corners, minimum width=2.2cm, minimum height=0.7cm, draw},
    process/.style={rectangle, minimum width=2.5cm, minimum height=0.8cm, draw},
    decision/.style={diamond, aspect=1.8, minimum width=1.8cm, draw}
}

\node (start) [startstop] {Start};
\node (test) [process, below=of start] {For $k \in [K]$, test $H_0:\xi(W_k) = 0$};
\node (reject) [decision, below=1cm of test] {Reject?};

\node (xi0) [process, below left=0.5cm and 0.5cm of reject] {
    Output $\{0\}$\\
    (or $[0,\infty)$ as the fallback).\\
    Nonrejection does not establish\\
    that the null is true.};

\node (xi1) [process, below right=0.5cm and 0.5cm of reject] {
    Construct the interval \eqref{eq:confidence.interval}.\\
    For nonempty $\calS\subseteq[K]\setminus\{k\}$, use\\
    $\xi((\vee_{k'\in\calS}W_{k'})\wedge W_k)$\\
    $=\xi(\vee_{k'\in\calS}W_{k'})+\xi(W_k)$\\
    $-\xi((\vee_{k'\in\calS}W_{k'})\vee W_k)$.\footnotemark};

\draw [->] (start) -- (test);
\draw [->] (test) -- (reject);
\draw [->] (reject) -| node[pos=0.25, above] {No} (xi0);
\draw [->] (reject) -| node[pos=0.25, above] {Yes} (xi1);

\end{tikzpicture}
\end{algorithm}

\footnotetext{Interaction inference also requires positive influence-function variance; degeneracy may occur, as in \Cref{sec:simulation}. Coverage here is for each specified estimand; simultaneous conclusions across factors require multiplicity adjustment.}

\section{Additional lemmas and corollaries}\label{appe:sec:lemma}

\begin{lemma}[\textit{Theorem 5.6.}\label{56T} in \cite{tsiatis2006semiparametric}]
All notation in this lemma is local to the lemma.
Here $W$ and $Z$ are generic random elements, and $\beta$ and $\eta$ are parameters indexing the marginal model $p_Z$.
If no restrictions are put on the conditional density $p_{W|Z}(w|z)$, where the marginal density of $Z$ is assumed to be from the semi-parametric model $p_Z(z, \beta, \eta)$, then the orthogonal complement of the tangent space $\mathscr{T}^{WZ}$ for the semi-parametric model for the joint distribution of $(W, Z)$ (i.e., $(\mathscr{T}^{WZ})^\perp$) is equal to the orthogonal complement of the tangent space $\mathscr{T}^Z$ for the semi-parametric model of the marginal distribution for $Z$ alone (i.e., $(\mathscr{T}^{Z})^\perp$).
    
\end{lemma}

\begin{lemma}\label{auxiliary1}
Fix $\calS\subseteq [K]$ such that $0<|-\calS|<K$.
Assume that each of the two blocks $\bm W_\calS$ and $\bm W_{-\calS}$ admits a centered square-integrable function with positive variance.
Consider the semiparametric model $\mathcal P_{\text{blk}}$ for $(Y,\bm W)=(Y,\bm W_{-\calS},\bm W_\calS)$ defined by
\[
p(y,\bm w_{-\calS},\bm w_\calS)=p_{Y\mid \bm W}\cdot p_{\bm W_{-\calS}}\cdot p_{\bm W_\calS},
\]
where we put no restrictions on
$p_{Y\mid \bm W}$, $p_{\bm W_{-\calS}}$, or $p_{\bm W_{\calS}}$ beyond domination and square integrability of scores.

Let $\dot{\mathcal P}_{\mathrm{blk}}$ be the tangent space of $\mathcal P_{\text{blk}}$ at $\PP$.
Let $\mathcal P_{\text{np}}(Y,\bm W_{-\calS})$ be the unrestricted model for the marginal law of $(Y,\bm W_{-\calS})$, with tangent space $\dot{\mathcal P}_{Y,\bm W_{-\calS}} = L_2^0(\PP_{Y,\bm W_{-\calS}})$.

Then
\[
\dot{\mathcal P}^\perp_{Y,\bm W_{-\calS}}=\{0\},
\ \text{but}\
\dot{\mathcal P}^\perp_{\mathrm{blk}}\neq \{0\}.
\]
\end{lemma}

\begin{proof}
Since $\mathcal P_{\text{np}}(Y,\bm W_{-\calS})$ is unrestricted, its tangent space is $L_2^0(\PP_{Y,\bm W_{-\calS}})$ and hence
$\dot{\mathcal P}^\perp_{Y,\bm W_{-\calS}}=\{0\}$.

Under $\mathcal P_{\text{blk}}$, the score contributions from $p_{Y\mid \bm W}$, $p_{\bm W_{-\calS}}$, and $p_{\bm W_{\calS}}$ vary
independently, so the tangent space decomposes orthogonally as
\[
\dot{\mathcal P}_{\mathrm{blk}}
=
\dot{\mathcal P}_{Y\mid \bm W}\ \oplus^\perp\ \dot{\mathcal P}_{\bm W_{-\calS}}\ \oplus^\perp\ \dot{\mathcal P}_{\bm W_\calS},
\]
where
\[
\dot{\mathcal P}_{Y\mid \bm W}=\{s(Y,\bm W):\EE[s\mid \bm W]=0\},\quad
\dot{\mathcal P}_{\bm W_{-\calS}}=\{s(\bm W_{-\calS}):\EE[s]=0\},\quad
\dot{\mathcal P}_{\bm W_\calS}=\{s(\bm W_\calS):\EE[s]=0\}.
\]

Hence $h(Y,\bm W)\in \dot{\mathcal P}^\perp_{\mathrm{blk}} = \dot{\mathcal P}^\perp_{Y\mid \bm W}\ \cap\ \dot{\mathcal P}^\perp_{\bm W_{-\calS}}\ \cap\ \dot{\mathcal P}^\perp_{\bm W_\calS}$ iff it is orthogonal to each component.
First, any $h(\bm W)$ is orthogonal to $\dot{\mathcal P}_{Y\mid \bm W}$ because for $s\in\dot{\mathcal P}_{Y\mid \bm W}$,
\[
\EE[h(\bm W)s(Y,\bm W)] = \EE\{h(\bm W)\EE[s(Y,\bm W)\mid \bm W]\}=0.
\]
So it suffices to find a nonzero $h(\bm W)$ orthogonal to both $\dot{\mathcal P}_{\bm W_{-\calS}}$ and $\dot{\mathcal P}_{\bm W_\calS}$.

In this proof, choose square-integrable functions $u(\bm W_\calS)$ and $v(\bm W_{-\calS})$ such that
$\EE[u(\bm W_\calS)]=\EE[v(\bm W_{-\calS})]=0$ and neither is a.s.\ zero. Define
\[
h(\bm W):=u(\bm W_\calS)\,v(\bm W_{-\calS}).
\]
Then $h(\bm W)\not\equiv 0$. Moreover, for any mean-zero $s(\bm W_\calS)$,
using block independence and $\EE[v]=0$,
\[
\EE[h(\bm W)\,s(\bm W_\calS)]
=
\EE\!\left[ v(\bm W_{-\calS}) \right]\,
\EE\!\left[u(\bm W_\calS)s(\bm W_\calS)\right]
=0,
\]
and similarly for any mean-zero $s(\bm W_{-\calS})$,
\[
\EE[h(\bm W)\,s(\bm W_{-\calS})]
=
\EE\!\left[ u(\bm W_\calS)\right]\,
\EE\!\left[v(\bm W_{-\calS})s(\bm W_{-\calS})\right]
=0.
\]
Thus $h\in \dot{\mathcal P}^\perp_{\mathrm{blk}}$ and the orthocomplement is nontrivial.
\end{proof}

\begin{lemma}\label{auxiliary2}
Assume $K\ge 2$ and the full-independence model $\mathcal P_{\text{ind}}$ for $(Y,\bm W)$:
\[
p(y,\bm w)=p_{Y\mid \bm W}\prod_{k=1}^K p_{W_k},
\]
with no restrictions on $p_{Y\mid \bm W}$ and each $p_{W_k}$ beyond domination and square-integrable scores.
Assume that at least two distinct coordinates admit centered square-integrable functions with positive variance.
Let $\dot{\mathcal P}_{\text{ind}}$ be its tangent space at $\PP$.
Let $\mathcal P_{\text{np}}(Y)$ be the unrestricted model for the marginal law of $Y$ alone,
with tangent space $\dot{\mathcal P}_Y=L_2^0(\PP_Y)$.

Then
\[
\dot{\mathcal P}_Y^\perp=\{0\},
\ \text{but}\
\dot{\mathcal P}_{\text{ind}}^\perp\neq \{0\}.
\]
\end{lemma}

\begin{proof}
Since $\mathcal P_{\text{np}}(Y)$ is unrestricted, $\dot{\mathcal P}_Y=L_2^0(\PP_Y)$ and thus $\dot{\mathcal P}_Y^\perp=\{0\}$.

Under $\mathcal P_{\mathrm{ind}}$, the tangent space decomposes orthogonally as
\[
\dot{\mathcal P}_{\text{ind}}
=
\bigoplus_{k\in[K]}^\perp\dot{\mathcal P}_{W_k}\oplus^\perp \dot{\mathcal P}_{Y\mid \bm W},
\]
where $\dot{\mathcal P}_{Y\mid \bm W}=\{s(Y,\bm W):\EE[s\mid \bm W]=0\}$ and
$\dot{\mathcal P}_{W_k}=\{s(W_k):\EE[s]=0\}$.

In this proof, fix distinct indices $i\neq j$ and choose square-integrable $u(W_i),v(W_j)$ with
$\EE[u(W_i)]=\EE[v(W_j)]=0$ and not a.s.\ zero. Let
\[
h(\bm W):=u(W_i)v(W_j).
\]
Then $h(\bm W)\not\equiv 0$. As in \Cref{auxiliary1}, $h(\bm W)\perp \dot{\mathcal P}_{Y\mid \bm W}$.
For any $s(W_k)\in \dot{\mathcal P}_{W_k}$, mutual independence and centering imply
\[
\EE[h(\bm W)s(W_i)] = \EE[v(W_j)]\,\EE[u(W_i)s(W_i)]=0,\
\EE[h(\bm W)s(W_j)] = \EE[u(W_i)]\,\EE[v(W_j)s(W_j)]=0,
\]
and for $k\notin\{i,j\}$,
\[
\EE[h(\bm W)s(W_k)] = \EE[h(\bm W)]\,\EE[s(W_k)]=0.
\]
Thus $h\in \dot{\mathcal P}_{\text{ind}}^\perp = \bigcap_{k\in[K]}\dot{\mathcal P}_{W_k}^\perp\cap \dot{\mathcal P}_{Y\mid \bm W}^\perp$, proving the orthocomplement is nontrivial.
\end{proof}

\begin{corollary}\label{remark: propensity.2}Under the moment and positive-variance conditions of \Cref{lemm:influence.function.conditional.expectation.squared}, if $\PP_{\bm W}$ is known and $\PP_{Y\mid\bm W}$ is unrestricted, 
the efficient influence function takes the form
\begin{align*}
 &\frac{2\left(Y-\mathbb{E}[Y\mid\bm W]\right)\mathbb{E}[Y\mid \bm W_{-\calS}]}{\var\left(Y\right)}\\
 &\quad-\frac{\Big\{\left(Y-\EE[Y]\right)^2-\EE\left[(Y-\EE[Y])^2\mid\bm W\right]\Big\}\cdot\mathbb{E}\left[\mathbb{E}\left[Y\mid \bm W_{-\calS}\right]^2\right]}{\var\left(Y\right)^2}.
\end{align*}

\end{corollary}

\begin{proof}
With $\PP_{\bm W}$ fixed, the tangent space is $\{s\in L_2^0(\PP):\EE[s\mid\bm W]=0\}$.
The efficient influence function is therefore the unrestricted influence function in \eqref{eq:IF.np} minus its conditional mean given $\bm W$:
\[
\varphi_{\EIF}
=\varphi_{\IF}-\EE[\varphi_{\IF}\mid\bm W].
\]
Substitution of \eqref{eq:IF.np} gives the displayed expression.
\end{proof}

\begin{lemma}[\textit{Lemma 2.} in \cite{kennedy2020ins}]\label{l2k}Let $\widehat{f}(\mathbf{o})$ be a function estimated from a sample $\mathbf{O}^N = (\mathbf{O}_{n+1}, \dots, \mathbf{O}_N)$, and let $\mathbb{P}_n$ denote the empirical measure over $(\mathbf{O}_1, \dots, \mathbf{O}_n)$, which is independent of $\mathbf{O}^N$.
In this lemma, $\mathbb P$ denotes the common observation law, $f\in L_2(\mathbb P)$ is fixed, and the norm is in $L_2(\mathbb P)$.
Then
\[
(\mathbb{P}_n - \mathbb{P})(\widehat{f} - f) = O_{\mathbb{P}}\left( \frac{\| \widehat{f} - f \|}{\sqrt{n}} \right).
\]

\end{lemma}

\begin{lemma}\label{L2swap-train}
In this lemma, let $\bm W_1,\ldots,\bm W_n$ be i.i.d.\ from $\PP_{\bm W}$, and fix a fold index $l$.
Let $\calI_{-l}=[n]\setminus\calI_l$ be its training indices, with training size $n_{-l}=n-n_l\to\infty$.
Let $h_{-l}$ be training-measurable and suppose
$\PP\{h_{-l}\in\mathcal H\}\to1$ for a fixed class $\mathcal H$ whose
square class $\mathcal F=\{h^2:h\in\mathcal H\}$ is $\PP_{\bm W}$-Donsker.
Then
\[
\left|\frac{1}{n_{-l}}\sum_{i\in\calI_{-l}}h_{-l}(\bm W_i)^2
-\|h_{-l}\|_{L_2(\PP_{\bm W})}^2\right|
=O_\PP(n_{-l}^{-1/2}).
\]
If additionally $\|h_{-l}\|_{L_2(\PP_{\bm W})}\to_\PP0$, the bound improves to
$o_\PP(n_{-l}^{-1/2})$. For fixed proportional folds this is
$o_\PP(n^{-1/2})$.
\end{lemma}

\begin{proof}
The difference is the absolute centered training average of $h_{-l}^2$.
Donsker tightness gives the stated big-$O$ bound uniformly over $\mathcal F$.
For the strengthened conclusion, in this proof put $f_{-l}=h_{-l}^2\ge0$, so
$\PP_{\bm W} f_{-l}=\|h_{-l}\|_2^2\to_\PP0$.
A Donsker class is totally bounded in the covariance semimetric
$\|(f-\PP_{\bm W}f)-(g-\PP_{\bm W}g)\|_{L_2(\PP_{\bm W})}$.
On a class of nonnegative functions this implies that, for every
$\varepsilon>0$, some $\delta>0$ satisfies
$\PP_{\bm W} f<\delta\Rightarrow\|f\|_2<\varepsilon$.
Indeed, otherwise a sequence with $\PP_{\bm W} f_j\to0$ and
$\|f_j\|_2\ge\varepsilon$ has a subsequence whose centered functions
converge in $L_2$. The uncentered functions have the same limit, which is
nonnegative and has mean zero, hence is zero, a contradiction.
This deterministic implication gives $\|f_{-l}\|_2\to_\PP0$.
Adjoin zero to $\mathcal F$ and use asymptotic equicontinuity to obtain
$n_{-l}^{-1/2}\sum_{i\in\calI_{-l}}\{f_{-l}(\bm W_i)-\PP_{\bm W}f_{-l}\}=o_\PP(1)$.
The probability of being outside the class tends to zero and does not change
the conclusion. No independence between $h_{-l}$ and its training sample is
needed, and no extra $L_4$ or sup-norm assumption has been added.
\end{proof}

\section{Omitted proofs}\label{appe:sec:proof}

\subsection{Proof of \Cref{lemma:definition}}\label{appe:sec:derivation:IF.np}

    \begin{proof} \begin{align*}
    \xi(\vee_{k \in \calS} W_k)  := & \frac{\var\left(Y - Y(\bm W'_{\calS}, \bm W_{-\calS})\right)}{2\var\left(Y\right)}\\
     = & \frac{\EE\left[\var\left(Y - Y(\bm W'_{\calS}, \bm W_{-\calS})\mid Y(\cdot), \bm W_{-\calS}\right)\right]}{2\var\left(Y\right)}\\
     &+ \frac{\var\left[\EE\left(Y - Y(\bm W'_{\calS}, \bm W_{-\calS})\mid  Y(\cdot), \bm W_{-\calS}\right)\right]}{2\var\left(Y\right)} \  (\text{\Cref{eq:super.population}}) \\
     = & \frac{\EE\left[\var\left(Y \mid Y(\cdot), \bm W_{-\calS}\right)\right]+\EE\left[\var\left( Y(\bm W'_{\calS}, \bm W_{-\calS}) \mid Y(\cdot), \bm W_{-\calS}\right)\right]}{2\var\left(Y\right)} \  (\text{\Cref{assu:independent}})\\
     = & \frac{\EE\left[\var\left(Y \mid Y(\cdot), \bm W_{-\calS}\right)\right]}{\var\left(Y\right)} \  (\text{independent copy}) \\ 
         =& \frac{\mathbb{E}[\var(Y\mid \bm W_{-\mathcal{S}}, E_Y)]}{\var\left(Y\right)} \quad (\text{\Cref{defi:total.independent}})\\
         =& \frac{\mathbb{E}[\var(\EE[Y \mid \bm W] + E_Y-\EE[E_Y] \mid \bm W_{-\mathcal{S}}, E_Y)]}{\var\left(Y\right)}\quad(\text{\Cref{assu:additive.independent.error}})\\
         =&\frac{\mathbb{E}\left[\var(\mathbb{E}[Y\mid \bm W] \mid \bm W_{-\mathcal{S}})\right]}{\var\left(Y\right)} \quad (\bm W \independent E_Y)\\
         =& \frac{\mathbb{E}[\mathbb{E}^2\left[Y\mid \bm W]\right]}{\var\left(Y\right)}-\frac{\mathbb{E}\left[\mathbb{E}^2[Y\mid \bm W_{-\mathcal{S}}]\right]}{\var\left(Y\right)}. \quad (\text{Definition of variance})
\end{align*}

\end{proof}

\subsection{Proof of \Cref{lemm:influence.function.conditional.expectation.squared}}\label{appe:sec:proof:lemma:influence.function.conditional.expectation.squared}

\begin{proof}

In this proof, we write the parameter of interest as $\Psi:=\psi\circ T$, where $\Psi:\mathcal P \to \mathbb R$, $\psi: \mathbb D \to \mathbb R$ and $T:\mathcal P \to \mathbb D$. Let $T_1:=\EE\!\left[\EE\!\left\{Y(\bm W)\mid \bm W_{-\mathcal S}\right\}^2\right]$, $T_2:= \var(Y)$, and define $T := (T_1,T_2)$,
$\psi(x,y):=x/y$. Assume that \(T_2>0\), so that \(T\in\mathbb D:=\{(x,y)\in\mathbb R^2:y>0\}\). Since \(\psi\) is \(C^1\) on the open set \(\mathbb D\), it is Fr\'echet differentiable there, and hence Hadamard differentiable. By \textit{Lemma 1.} in \cite{Williamson2021}, $T$ is pathwise differentiable, i.e. Hadamard differentiable along the tangent space $\dot{\mathcal P}$ \parencite{van1991differentiable,bickel1993efficient}. Since Hadamard differentiability allows chain rule \parencite{averbukh1967theory}, $\Psi=\psi\circ T$ is Hadamard differentiable along $\dot{\mathcal P}$, i.e. pathwise differentiable, with derivative
\[
\dot\Psi(S) = \dot\psi_{T}\circ \dot T = \frac{\dot T_{1}(S)}{T_2}-\frac{T_1}{T_2^2}\dot T_{2}(S) \in \mathbb R,
\quad S\in\dot{\mathcal P}.
\]Then by Riesz-Fr\'echet representation theorem, there exists a unique $\varphi_{\mathrm{EIF}} \in \dot{\mathcal P}$ such that\[\dot\Psi(S) = \langle \varphi_{\mathrm{EIF}}, S\rangle_{L_2} = \langle \frac{\varphi_{T_1,\mathrm{EIF}}}{T_2}-\frac{T_1}{T_2^2}\varphi_{T_2,\mathrm{EIF}}, S\rangle_{L_2}.\]

If we ignore the independence among treatments, then the model we consider is a non-parametric model $\mathcal{P}_{\text{np}}$ with tangent space  $\dot{\mathcal{P}}_{\text{np}} = L_2^0(\PP)$. Under the assumption of discrete treatments, for $-\calS \subseteq [K]$, we use the operator $\mathbb{IF}$ for pathwise differentiable functionals \parencite{kennedy2024semiparametric} to get the non-parametric influence function of the numerator, where, in this discrete calculation, $\lambda$ is counting measure and $p_{\bm W_{-\calS}}$ is the corresponding probability mass function:
\begin{align*}
        \varphi_{T_1,\mathrm{IF}} & = \mathbb{IF}\Big\{\mathbb{E}\left[\mathbb{E}\left[Y(\bm W)\mid \bm W_{-\calS}\right]^2\right]\Big\} \\
        & = \mathbb{IF}\Big\{\int \mathbb{E}\left[Y(\bm W)\mid \bm W_{-\calS}=\bm w_{-\calS}\right]^2 p_{\bm W_{-\calS}}(\bm w_{-\calS}) \lambda(d\bm w_{-\calS})\Big\}\\
        & = \int\mathbb{IF}\{\mathbb{E}\left[Y(\bm W)\mid \bm W_{-\calS}=\bm w_{-\calS}\right]^2\}p_{\bm W_{-\calS}}(\bm w_{-\calS}) \lambda(d\bm w_{-\calS})\\
        &\quad~+ \int \mathbb{E}\left[Y(\bm W)\mid \bm W_{-\calS}=\bm w_{-\calS}\right]^2\mathbb{IF}(p_{\bm W_{-\calS}}(\bm w_{-\calS})) \lambda(d\bm w_{-\calS})\\
        & =\int 2\mathbb{E}[Y(\bm W)\mid \bm W_{-\calS}=\bm w_{-\calS}]\frac{\1(\bm W_{-\calS}=\bm w_{-\calS})}{p_{\bm W_{-\calS}}(\bm w_{-\calS})}\\
        &\qquad{}\cdot\left(Y(\bm W)-\mathbb{E}[Y(\bm W)\mid \bm W_{-\calS}=\bm w_{-\calS}]\right)p_{\bm W_{-\calS}}(\bm w_{-\calS}) \lambda(d\bm w_{-\calS})\\
        &\quad~+ \int \mathbb{E}[Y(\bm W)\mid \bm W_{-\calS}=\bm w_{-\calS}]^2\left( \1(\bm W_{-\calS} = \bm w_{-\calS})-p_{\bm W_{-\calS}}(\bm w_{-\calS})\right) \lambda(d\bm w_{-\calS})\\
        & = 2Y(\bm W)\mathbb{E}[Y(\bm W)\mid \bm W_{-\calS}] - \mathbb{E}[Y(\bm W)\mid \bm W_{-\calS}]^2 - \mathbb{E}\left[\mathbb{E}[Y(\bm W)\mid \bm W_{-\calS}]^2\right].
    \end{align*}

The same numerator formula holds without assuming discrete treatments.
In this proof, for a bounded mean-zero score $s$, use the local path $d\PP_{\varepsilon}=(1+\varepsilon s)\,d\PP$.
The reduced regression under this path satisfies
\[
\EE_{\PP_{\varepsilon}}[Y\mid\bm W_{-\calS}]
=\frac{\EE[Y\mid\bm W_{-\calS}]+\varepsilon\EE[Ys\mid\bm W_{-\calS}]}
        {1+\varepsilon\EE[s\mid\bm W_{-\calS}]}.
\]
Consequently,
\begin{align*}
 T_1(\PP_{\varepsilon})
 &=\EE\left[\frac{(\EE[Y\mid\bm W_{-\calS}]+\varepsilon\EE[Ys\mid\bm W_{-\calS}])^2}
                      {1+\varepsilon\EE[s\mid\bm W_{-\calS}]}\right],\\
 \left.\frac{d}{d\varepsilon}T_1(\PP_{\varepsilon})\right|_{\varepsilon=0}
 &=\EE[2\EE[Y\mid\bm W_{-\calS}]\EE[Ys\mid\bm W_{-\calS}]
          -\EE[Y\mid\bm W_{-\calS}]^2\EE[s\mid\bm W_{-\calS}]]\\
 &=\EE[(2Y\EE[Y\mid\bm W_{-\calS}]-\EE[Y\mid\bm W_{-\calS}]^2-\EE[\EE[Y\mid\bm W_{-\calS}]^2])s].
\end{align*}
For sufficiently small $\varepsilon$, the denominator is bounded away from zero.
Boundedness of the score $s$ and $Y\in L_2(\PP)$ justify differentiation under the integral.
Bounded mean-zero scores are dense in $L_2^0(\PP)$, and the assumed square-integrability of the displayed gradient makes this derivative pairing continuous.
Thus the same tangent-space derivative representation applies for general treatment spaces.

    For the denominator, we have its non-parametric influence function as

    \begin{align*}
        \varphi_{T_2,\mathrm{IF}} = \mathbb{IF}\Big\{\var\left(Y(\bm W)\right)\Big\} = \left(Y(\bm W) - \EE\left(Y(\bm W)\right)\right)^2-\var\left(Y(\bm W)\right).
    \end{align*}

   So the non-parametric influence function of the parameter of interest is \begin{align*}
     \varphi_{\mathrm{IF}} 
       & = \frac{2Y(\bm W)\mathbb{E}[Y(\bm W)\mid \bm W_{-\calS}] - \mathbb{E}[Y(\bm W)\mid \bm W_{-\calS}]^2 - \mathbb{E}\left[\mathbb{E}[Y(\bm W)\mid \bm W_{-\calS}]^2\right]}{\var\left(Y(\bm W)\right)}\\
        & \quad~-\frac{\left(Y(\bm W) - \EE\left(Y(\bm W)\right)\right)^2\cdot\mathbb{E}\left[\mathbb{E}\left[Y(\bm W)\mid \bm W_{-\calS}\right]^2\right]-\var\left(Y(\bm W)\right)\cdot\mathbb{E}\left[\mathbb{E}\left[Y(\bm W)\mid \bm W_{-\calS}\right]^2\right]}{\var\left(Y(\bm W)\right)^2}\\
        & = \frac{2Y(\bm W)\mathbb{E}[Y(\bm W)\mid \bm W_{-\calS}] - \mathbb{E}[Y(\bm W)\mid \bm W_{-\calS}]^2}{\var\left(Y(\bm W)\right)}\\
        &\quad~-\mathbb{E}\left[\mathbb{E}[Y(\bm W)\mid \bm W_{-\calS}]^2\right]\cdot \Big\{\frac{Y(\bm W) - \EE\left(Y(\bm W)\right)}{\var\left(Y(\bm W)\right)}\Big\}^2.
   \end{align*}

    Given $-\calS \subseteq [K]$ and fully independent treatments, the joint probability measure of observed data can be factorized as $\mathbb{P}_{Y,\bm W} = \mathbb{P}_{Y\mid\bm W}\cdot\prod_{j\in -\calS} \mathbb{P}_{W_j}\cdot\prod_{i\in \mathcal{S}} \mathbb{P}_{W_i}$. Let $\mathcal{P}_{\text{ind}}$ be the set of all regular densities $p_{Y,\bm W}$ such that \begin{align*}
        p_{Y,\bm W} = p_{Y\mid \bm W}\cdot\prod_{j \in -\calS} p_{W_j} \cdot \prod_{i \in \calS} p_{W_i}.
    \end{align*} The reason why we have $\{W_i: i \in \calS\}$ here is that, even though they are not needed to define our parameter of interest, fully independent treatments implies that $p_{\bm W_{-\calS}\mid \bm W_{\calS}} = p_{\bm W_{-\calS}}$, then, when both blocks satisfy the nondegeneracy conditions of \Cref{auxiliary1}, the orthogonal complement of the tangent space $\dot{\mathcal{P}}_{Y,\bm W_{-\calS}, \bm W_{\calS}}$ for the joint distribution of $(Y,\bm W_{-\calS}, \bm W_{\calS})$, $\dot{\mathcal{P}}_{Y,\bm W_{-\calS}, \bm W_{\calS}}^\bot$, is not the same as the orthogonal complement of the tangent space $\dot{\mathcal{P}}_{Y,\bm W_{-\calS}}$ for the joint distribution of $(Y,\bm W_{-\calS})$, $\dot{\mathcal{P}}_{Y,\bm W_{-\calS}}^\bot$. Including these auxiliary variables can therefore improve efficiency, with strict improvement only when the discarded gradient component is nonzero. Then for this semi-parametric model, we can project the non-parametric influence function of the numerator onto the corresponding tangent space,\begin{align*}
        \dot{\mathcal{P}}_{\text{ind}} 
        &= \bigoplus^{\perp}_{k\in [K]} \dot{\mathcal{P}}_{W_k} \oplus^{\perp} \dot{\mathcal{P}}_{Y\mid \bm W} \\
        &= \bigoplus^{\perp}_{k\in [K]}\{S(W_k): \mathbb{E}(S) = 0\}\oplus^{\perp}\{S(Y, \bm W): \mathbb{E}(S\mid \bm W) = 0\},\end{align*}
        to get the efficient influence function
        \begin{align*}
        \varphi_{T_1,\mathrm{EIF}}
        =& \sum_{k\in [K]}\Pi\left(\mathbb{IF}\Big\{\mathbb{E}\left[\mathbb{E}[Y(\bm W)\mid \bm W_{-\calS}]^2\right]\Big\}\mid \dot{\mathcal{P}}_{W_k}\right)\\
        &+\Pi\left(\mathbb{IF}\Big\{\mathbb{E}\left[\mathbb{E}[Y(\bm W)\mid \bm W_{-\calS}]^2\right]\Big\}\mid \dot{\mathcal{P}}_{Y\mid \bm W}\right)\\
 = & \sum_{k \in [K]}\EE\left(\mathbb{IF}\Big\{\mathbb{E}\left[\mathbb{E}[Y(\bm W)\mid \bm W_{-\calS}]^2\right]\Big\}\mid W_k\right)+\mathbb{IF}\Big\{\mathbb{E}\left[\mathbb{E}[Y(\bm W)\mid \bm W_{-\calS}]^2\right]\Big\}\\
 &-\EE\left(\mathbb{IF}\Big\{\mathbb{E}\left[\mathbb{E}[Y(\bm W)\mid \bm W_{-\calS}]^2\right]\Big\}\mid \bm W\right)\\
  = & 2\left(Y(\bm W)-\mathbb{E}[Y(\bm W)\mid \bm W]\right)\cdot \mathbb{E}[Y(\bm W)\mid \bm W_{-\calS}]\\
  &+ \sum_{j\in -\calS}\Big\{\mathbb{E}\left[\mathbb{E}[Y(\bm W)\mid \bm W_{-\calS}]^2\mid W_j\right]-\mathbb{E}\left[\mathbb{E}[Y(\bm W)\mid \bm W_{-\calS}]^2\right]\Big\}\\
 & + 2\sum_{i\in \calS}\Big\{\mathbb{E}\left[\mathbb{E}[Y(\bm W)\mid \bm W]\cdot\mathbb{E}[Y(\bm W)\mid \bm W_{-\calS}]\mid W_i\right]-\mathbb{E}\left[\mathbb{E}[Y(\bm W)\mid \bm W_{-\calS}]^2\right]\Big\}.
\end{align*}

For the denominator, we can also project its non-parametric influence function onto $\dot{\mathcal{P}}_{\text{ind}}$. Because constraining the marginal distribution of $\bm W$ to be factorizable (mutually independent) implicitly restricts the conditional density $p_{\bm W\mid Y}$ to only those specific forms that preserve this independence after integrating out $Y$, which violates the condition of \Cref{56T}. Projecting onto the displayed tangent space gives\begin{align*}
    \varphi_{T_2,\mathrm{EIF}} = &\left(Y-\EE\left[Y\right]\right)^2-\var\left(Y\mid \bm W\right)-\left(\EE[Y\mid\bm W]-\EE[Y]\right)^2\\
&+\sum_{k\in[K]}\Big\{\EE\left[(Y-\EE[Y])^2\mid W_k\right]-\var\left(Y\right)\Big\}.
\end{align*} 

So the efficient influence function of the parameter of interest for $\mathcal{P}_{\text{ind}}$ is\begin{align*}
   \varphi_{\mathrm{EIF}} &= \frac{2\left(Y(\bm W)-\mathbb{E}[Y(\bm W)\mid \bm W]\right)\cdot \mathbb{E}[Y(\bm W)\mid \bm W_{-\calS}]}{\var\left(Y(\bm W)\right)} \\
    & \quad~+   \frac{2\sum_{i\in \calS}\Big\{\mathbb{E}\left[\mathbb{E}[Y(\bm W)\mid \bm W]\cdot\mathbb{E}[Y(\bm W)\mid \bm W_{-\calS}]\mid W_i\right]-\mathbb{E}\left[\mathbb{E}[Y(\bm W)\mid \bm W_{-\calS}]^2\right]\Big\}}{\var\left(Y(\bm W)\right)}\\
&\quad~+\frac{\sum_{j\in -\calS}\Big\{\mathbb{E}\left[\mathbb{E}[Y(\bm W)\mid \bm W_{-\calS}]^2\mid W_j\right]-\mathbb{E}\left[\mathbb{E}[Y(\bm W)\mid \bm W_{-\calS}]^2\right]\Big\}}{\var\left(Y(\bm W)\right)}\\
        & \quad~-\frac{\Big\{\left(Y-\EE\left[Y\right]\right)^2-\var\left(Y\mid \bm W\right)-\left(\EE[Y\mid\bm W]-\EE[Y]\right)^2\Big\}\cdot\mathbb{E}\left[\mathbb{E}\left[Y(\bm W)\mid \bm W_{-\calS}\right]^2\right]}{\var\left(Y(\bm W)\right)^2}\\
        & \quad~-\frac{\Big\{\sum_{k\in[K]}\Big\{\EE\left[(Y-\EE[Y])^2\mid W_k\right]-\var\left(Y\right)\Big\}\Big\}\cdot\mathbb{E}\left[\mathbb{E}\left[Y(\bm W)\mid \bm W_{-\calS}\right]^2\right]}{\var\left(Y(\bm W)\right)^2}.
\end{align*}

\end{proof}

\subsection{Proof of \Cref{prop:limiting.np}}

\begin{proof}
Apply Theorem~1 of \cite{Williamson2021} and its cross-fitted formulation,
Theorem~2 in
Section~3.3 of \cite{williamson2023general}.
In this proof, take $X=\bm W$; the full and reduced regressions are $\mu$ and $\EE[Y\mid\bm W_{-\calS}]$:
total explainability is their $R^2$ contrast, and interaction explainabilities
are fixed finite linear combinations of these contrasts.
Use the empirical training moments specified above; the direct correction
is first-order equivalent to the corrected component ratio by the comparison
at the end of the following proof.
\end{proof}

\subsection{Proof of \Cref{prop:limiting}}\label{appe:sec:proof:prop: limiting}

The quarter rates make the second-order remainder $o_\PP(n^{-1/2})$; the estimator's error is generally $O_\PP(n^{-1/2})$ when the limiting variance is positive.
The projection comparison is $\var\{\varphi_{\EIF}(Y,\bm W;\PP)\}\le \var\{\varphi_{\IF}(Y,\bm W;\PP)\}$ for the true influence functions.
\Cref{prop:limiting} uses marginal-density rates in place of a joint-law rate and the conditional-second-moment rate to control the fitted variance.
Its proof does not require a uniform bound on $\widehat\upsilon$.

\begin{proof}
All auxiliary notation introduced below is local to this proof.
In each decomposition below, the labels $a$, $b$, and $c$ refer only to the terms displayed in that decomposition.
For each fold $l$ and integrable function $f$ of $(Y,\bm W)$, define the held-out empirical operator by $\PP_n^l f:=n_l^{-1}\sum_{i\in\calI_l}f(Y_i,\bm W_i)$.
Define the remainder term based on probability measures $\overline{\mathbb{P}}$ and $\mathbb{P}$ as $R(\overline{\mathbb{P}},\mathbb{P}) = \psi_{\overline{\mathbb{P}}}-\psi +  \mathbb{P}\{\varphi_{\overline{\mathbb{P}}}\}$, where $\psi$ is any estimand.  We begin with the total-explainability numerator $\Theta=\EE[\mu(\bm W)^2]-\EE[\EE[Y\mid\bm W_{-\calS}]^2]$ and denote its efficient influence function by $\theta$. For $l = 1,\ldots,L$,\begin{align*}
    \widehat{\Theta}_l-\Theta &= \Theta_{\widehat{\mathbb{P}}_{-l}} + \mathbb{P}_n^l\{\theta_{\widehat{\mathbb{P}}_{-l}}\}-\Theta\\
    & = (\mathbb{P}_n^l-\mathbb{P})\{\theta_{\widehat{\mathbb{P}}_{-l}}\}+R(\widehat{\mathbb{P}}_{-l},\mathbb{P})\\
    & = (\mathbb{P}_n^l-\mathbb{P})\{\theta\}+(\mathbb{P}_n^l-\mathbb{P})\{\theta_{\widehat{\mathbb{P}}_{-l}}-\theta\}+R(\widehat{\mathbb{P}}_{-l},\mathbb{P}).
\end{align*}

For a held-out fold, the true-score average has variance
$\var(\theta)/n_l$. Its normal limit follows from the CLT; no coupling to a
normal random variable is asserted. By \Cref{l2k}, the centered fitted-score
error is $O_\PP(\|\theta_{\widehat\PP_{-l}}-\theta\|_2/\sqrt{n_l})$.
It therefore suffices to prove $L_2(\PP)$ consistency of the fitted score.
The fitted product law is $\widehat\PP_{\bm W,-l}=\prod_{k=1}^K\widehat\PP_{W_k,-l}$, and the true product law is $\PP_{\bm W}=\prod_{k=1}^K\PP_{W_k}$.
For functions of $\bm W$, $\widehat\EE_{-l}$ denotes integration under this fitted product law, and its conditional expectations integrate the other coordinates.
This fitted product law is distinct from the empirical training law.
For the remainder of this proof, let $\lambda_k$ be the common dominating measure in C4 and let $t_k$ denote the marginal total-variation error.
By C4 and the specified density norms,
\begin{align*}
t_k&:=\int|p_{W_k}-\widehat p_{W_k,-l}|\,d\lambda_k\\
&\le\underline c^{-1}\|p_{W_k}-\widehat p_{W_k,-l}\|_{L_2(\PP_{W_k})}\\
&=o_\PP(n^{-1/4}).
\end{align*}
The exact form of the first half of
$\theta_{\widehat\PP_{-l}}-\theta$ is\begin{align*}
   & (\theta_{\widehat{\mathbb{P}}_{-l}}-\theta)_{\text{first half}} =  \ 2Y(\bm W)\widehat{\mu}_{-l}-2\widehat{\mu}_{-l}^2+\sum_{k \in [K]}\widehat{\EE}_{-l}[\widehat{\mu}_{-l}^2\mid W_k]-K\cdot \widehat{\EE}_{-l}[\widehat{\mu}_{-l}^2]\\
    & - \Big\{2Y(\bm W) \mathbb{E}[Y(\bm W)\mid \bm W]-2\mathbb{E}[Y(\bm W)\mid \bm W]^2 + \sum_{k \in [K]}\mathbb{E}\big[\mathbb{E}[Y(\bm W)\mid \bm W]^2\mid W_k\big]-K\cdot\mathbb{E}\big[\mathbb{E}[Y(\bm W)\mid \bm W]^2\big]\Big\}\\
    = & \ \underbrace{2Y(\bm W)\widehat{\mu}_{-l}-2\widehat{\mu}_{-l}^2-2Y(\bm W) \mathbb{E}[Y(\bm W)\mid \bm W]+2\mathbb{E}[Y(\bm W)\mid \bm W]^2}_{: = a}\\
    & + \underbrace{\sum_{k \in [K]}\widehat{\EE}_{-l}[\widehat{\mu}_{-l}^2\mid W_k]-K\cdot \widehat{\EE}_{-l}[\widehat{\mu}_{-l}^2]-\sum_{k \in [K]}\mathbb{E}\big[\mathbb{E}[Y(\bm W)\mid \bm W]^2\mid W_k\big]+K\cdot\mathbb{E}\big[\mathbb{E}[Y(\bm W)\mid \bm W]^2\big]}_{: = b}.
\end{align*}

Note that, under C1. and C2., we have\begin{align*}
    & \ ||a|| = ||2Y(\bm W)\widehat{\mu}_{-l}-2\widehat{\mu}_{-l}^2-2Y(\bm W) \mathbb{E}[Y(\bm W)\mid \bm W]+2\mathbb{E}[Y(\bm W)\mid \bm W]^2||\\
  = & \ ||2Y(\bm W)\cdot\left(\widehat{\mu}_{-l}-\mathbb{E}[Y(\bm W)\mid \bm W]\right)-2\left(\widehat{\mu}_{-l}+\mathbb{E}[Y(\bm W)\mid \bm W]\right)\cdot\left(\widehat{\mu}_{-l}-\mathbb{E}[Y(\bm W)\mid \bm W]\right)||\\
  = & \ ||2\left(Y(\bm W)-\widehat{\mu}_{-l}-\mathbb{E}[Y(\bm W)\mid \bm W]\right)\cdot\left(\widehat{\mu}_{-l}-\mathbb{E}[Y(\bm W)\mid \bm W]\right)||\\
  \lesssim & \ ||\widehat{\mu}_{-l}-\mathbb{E}[Y(\bm W)\mid \bm W]||
  = o_{\mathbb{P}}(1)
\end{align*}
and\begin{align*}
\|b\|
&=\Bigg\|\sum_{k\in[K]}\widehat{\EE}_{-l}\!\big[\widehat{\mu}_{-l}^2\mid W_k\big]
-K\cdot \widehat{\EE}_{-l}\!\big[\widehat{\mu}_{-l}^2\big]
-\sum_{k\in[K]}\EE\!\Big[\EE\!\big[Y(\bm W)\mid \bm W\big]^2\mid W_k\Big]
+K\cdot \EE\!\Big[\EE\!\big[Y(\bm W)\mid \bm W\big]^2\Big]\Bigg\| \\
&\le
\sum_{k\in[K]}
\Big\|\widehat{\EE}_{-l}\!\big[\widehat{\mu}_{-l}^2\mid W_k\big]
-\EE\!\Big[\EE\!\big[Y(\bm W)\mid \bm W\big]^2\mid W_k\Big]\Big\|
+K\cdot
\Big|\widehat{\EE}_{-l}\!\big[\widehat{\mu}_{-l}^2\big]
-\EE\!\Big[\EE\!\big[Y(\bm W)\mid \bm W\big]^2\Big]\Big|.
\end{align*}

For the first term, for each fixed $k\in[K]$, we have
\begin{align*}
&\widehat{\EE}_{-l}\!\big[\widehat{\mu}_{-l}^2\mid W_k\big]
-\EE\!\Big[\EE\!\big[Y(\bm W)\mid \bm W\big]^2\mid W_k\Big] \\
&\quad=
\int \Big(\widehat{\mu}_{-l}^2-\EE\!\big[Y(\bm W)\mid \bm W\big]^2\Big)\prod_{k'\neq k} d\widehat{\mathbb P}_{W_{k'},-l}
\;+\;
\int \EE\!\big[Y(\bm W)\mid \bm W\big]^2
\Big(\prod_{k'\neq k} d\widehat{\mathbb P}_{W_{k'},-l}-\prod_{k'\neq k} d\mathbb P_{W_{k'}}\Big).
\end{align*}
Hence, by the triangle inequality and the fact that this quantity is $W_k$-measurable,
\begin{align*}
&\Big\|\widehat{\EE}_{-l}\!\big[\widehat{\mu}_{-l}^2\mid W_k\big]
-\EE\!\Big[\EE\!\big[Y(\bm W)\mid \bm W\big]^2\mid W_k\Big]\Big\|\\
&\quad=
\Big\|\widehat{\EE}_{-l}\!\big[\widehat{\mu}_{-l}^2\mid W_k\big]
-\EE\!\Big[\EE\!\big[Y(\bm W)\mid \bm W\big]^2\mid W_k\Big]\Big\|_{L_2(\mathbb P_{W_k})}\\
&\quad\le
\Bigg\|
\int \Big(\widehat{\mu}_{-l}^2-\EE\!\big[Y(\bm W)\mid \bm W\big]^2\Big)\prod_{k'\neq k} d\widehat{\mathbb P}_{W_{k'},-l}
\Bigg\|_{L_2(\mathbb P_{W_k})}\\
&\qquad +
\Bigg\|
\int \EE\!\big[Y(\bm W)\mid \bm W\big]^2
\Big(\prod_{k'\neq k} d\widehat{\mathbb P}_{W_{k'},-l}-\prod_{k'\neq k} d\mathbb P_{W_{k'}}\Big)
\Bigg\|_{L_2(\mathbb P_{W_k})}.
\end{align*}

Write $\widehat{\mu}_{-l}^2-\EE\!\big[Y(\bm W)\mid \bm W\big]^2
=
\Big(\widehat{\mu}_{-l}+\EE\!\big[Y(\bm W)\mid \bm W\big]\Big)
\Big(\widehat{\mu}_{-l}-\EE\!\big[Y(\bm W)\mid \bm W\big]\Big)$. For the left-hand side term, under C1., by Minkowski's integral inequality, 
\begin{align*}
&\Bigg\|
\int \Big(\widehat{\mu}_{-l}^2-\EE\!\big[Y(\bm W)\mid \bm W\big]^2\Big)\prod_{k'\neq k} d\widehat{\mathbb P}_{W_{k'},-l}
\Bigg\|_{L_2(\mathbb P_{W_k})}\\
&\quad\le
\int
\Big\|
\Big(\widehat{\mu}_{-l}+\EE\!\big[Y(\bm W)\mid \bm W\big]\Big)
\Big(\widehat{\mu}_{-l}-\EE\!\big[Y(\bm W)\mid \bm W\big]\Big)
\Big\|_{L_2(\mathbb P_{W_k})}
\prod_{k'\neq k} d\widehat{\mathbb P}_{W_{k'},-l}\\
&\quad\le
2C\int
\Big\|\widehat{\mu}_{-l}-\EE\!\big[Y(\bm W)\mid \bm W\big]\Big\|_{L_2(\mathbb P_{W_k})}
\prod_{k'\neq k} d\widehat{\mathbb P}_{W_{k'},-l}\\
&\quad\le
2C\Bigg(
\int
\Big\|\widehat{\mu}_{-l}-\EE\!\big[Y(\bm W)\mid \bm W\big]\Big\|_{L_2(\mathbb P_{W_k})}^2
\prod_{k'\neq k} d\widehat{\mathbb P}_{W_{k'},-l}
\Bigg)^{1/2}\\
&\quad=
2C\Bigg(
\int
\Big(\widehat{\mu}_{-l}-\EE\!\big[Y(\bm W)\mid \bm W\big]\Big)^2
\, d\mathbb P_{W_k}\prod_{k'\neq k} d\widehat{\mathbb P}_{W_{k'},-l}
\Bigg)^{1/2}.
\end{align*}
Using $0\le(\widehat{\mu}_{-l}-\EE[Y(\bm W)\mid \bm W])^2\le (2C)^2$ and the telescoping identity for products, we get
\begin{align*}
&\Bigg|
\int
\Big(\widehat{\mu}_{-l}-\EE\!\big[Y(\bm W)\mid \bm W\big]\Big)^2
\, d\mathbb P_{W_k}\prod_{k'\neq k} d\widehat{\mathbb P}_{W_{k'},-l}
-
\int
\Big(\widehat{\mu}_{-l}-\EE\!\big[Y(\bm W)\mid \bm W\big]\Big)^2
\, d\mathbb P_{\bm W}
\Bigg| \\
& \qquad= \Bigg|
\int
\Big(\widehat{\mu}_{-l}-\EE[Y(\bm W)\mid \bm W]\Big)^2
\, d\mathbb P_{W_k}
\Big(\prod_{k'\neq k} d\widehat{\mathbb P}_{W_{k'},-l}-\prod_{k'\neq k} d\mathbb P_{W_{k'}}\Big)
\Bigg| \\
&\qquad\le
\sum_{j\neq k}
\int
\Big(\widehat{\mu}_{-l}-\EE[Y(\bm W)\mid \bm W]\Big)^2
\, d\mathbb P_{W_k}
\Big|d\widehat{\mathbb P}_{W_j,-l}-d\mathbb P_{W_j}\Big|
\prod_{\substack{r\neq k\\ r<j}} d\widehat{\mathbb P}_{W_r,-l}
\prod_{\substack{r\neq k\\ r>j}} d\mathbb P_{W_r} \\
&\qquad\le
(2C)^2
\sum_{j\neq k}
\int
\Big|d\widehat{\mathbb P}_{W_j,-l}-d\mathbb P_{W_j}\Big|
\underbrace{\int d\mathbb P_{W_k}\prod_{\substack{r\neq k\\ r<j}} d\widehat{\mathbb P}_{W_r,-l}
\prod_{\substack{r\neq k\\ r>j}} d\mathbb P_{W_r}}_{=\,1} \\
&\qquad=
(2C)^2
\sum_{j\neq k}
\int
\Big|d\widehat{\mathbb P}_{W_j,-l}-d\mathbb P_{W_j}\Big|\\
&\qquad=
(2C)^2
\sum_{j\neq k}
\int \big|\widehat p_{W_j,-l}-p_{W_j}\big|\,d\lambda_j.
\end{align*}So under C2., we have
\begin{align*}
\Big|\int
\Big(\widehat{\mu}_{-l}-\EE[Y(\bm W)\mid \bm W]\Big)^2
\, d\mathbb P_{W_k}\prod_{k'\neq k} d\widehat{\mathbb P}_{W_{k'},-l}
-\int
\Big(\widehat{\mu}_{-l}-\EE[Y(\bm W)\mid \bm W]\Big)^2
\, d\mathbb P_{\bm W}
\Big| = o_{\PP}(1).
\end{align*}
Now use the inequality for $x,y\ge 0$,
\[
\big|\sqrt{x}-\sqrt{y}\big|
=\frac{|x-y|}{\sqrt{x}+\sqrt{y}}
\le \sqrt{|x-y|}.
\]
Applying this with
\[
x=\int
\Big(\widehat{\mu}_{-l}-\EE[Y(\bm W)\mid \bm W]\Big)^2
\, d\mathbb P_{W_k}\prod_{k'\neq k} d\widehat{\mathbb P}_{W_{k'},-l},
\
y=\int
\Big(\widehat{\mu}_{-l}-\EE[Y(\bm W)\mid \bm W]\Big)^2
\, d\mathbb P_{\bm W},
\]
gives
\begin{align*}
&\Bigg(
\int
\Big(\widehat{\mu}_{-l}-\EE[Y(\bm W)\mid \bm W]\Big)^2
\, d\mathbb P_{W_k}\prod_{k'\neq k} d\widehat{\mathbb P}_{W_{k'},-l}
\Bigg)^{1/2} \\
&\qquad\le
\Bigg(
\int
\Big(\widehat{\mu}_{-l}-\EE[Y(\bm W)\mid \bm W]\Big)^2
\, d\mathbb P_{\bm W}
\Bigg)^{1/2}\\
&\qquad+
\Bigg|
\int
\Big(\widehat{\mu}_{-l}-\EE[Y(\bm W)\mid \bm W]\Big)^2
\, d\mathbb P_{W_k}\prod_{k'\neq k} d\widehat{\mathbb P}_{W_{k'},-l}
-
\int
\Big(\widehat{\mu}_{-l}-\EE[Y(\bm W)\mid \bm W]\Big)^2
\, d\mathbb P_{\bm W}
\Bigg|^{1/2} \\
&\qquad=
\Big\|\widehat{\mu}_{-l}-\EE[Y(\bm W)\mid \bm W]\Big\|_{L_2(\mathbb P_{\bm W})}
+ o_{\PP}(1) = o_{\PP}(1).
\end{align*}

For the right-hand side term, we have
\begin{align*}
&\Bigg\|
\int \EE\!\big[Y(\bm W)\mid \bm W\big]^2
\Big(\prod_{k'\neq k} d\widehat{\mathbb P}_{W_{k'},-l}-\prod_{k'\neq k} d\mathbb P_{W_{k'}}\Big)
\Bigg\|_{L_2(\mathbb P_{W_k})}\\
&\quad\le
\sup_{w_k}
\Bigg|
\int \EE\!\big[Y(\bm W)\mid \bm W\big]^2
\Big(\prod_{k'\neq k} d\widehat{\mathbb P}_{W_{k'},-l}-\prod_{k'\neq k} d\mathbb P_{W_{k'}}\Big)
\Bigg|\\
&\quad\le
C^2\sum_{k'\neq k}\int \big|\widehat p_{W_{k'},-l}-p_{W_{k'}}\big|\,d\lambda_{k'} \;=\; o_{\PP}(1).
\end{align*}
Combining the two parts, for each $k$,
\[
\Big\|\widehat{\EE}_{-l}\!\big[\widehat{\mu}_{-l}^2\mid W_k\big]
-\EE\!\Big[\EE\!\big[Y(\bm W)\mid \bm W\big]^2\mid W_k\Big]\Big\|
=o_{\PP}(1).
\]

Similarly, for the second term,
\begin{align*}
&\Big|\widehat{\EE}_{-l}\!\big[\widehat{\mu}_{-l}^2\big]
-\EE\!\Big[\EE\!\big[Y(\bm W)\mid \bm W\big]^2\Big]\Big|\\
&\quad=
\Bigg|\int \Big(\widehat{\mu}_{-l}^2-\EE\!\big[Y(\bm W)\mid \bm W\big]^2\Big)\prod_{k} d\widehat{\mathbb P}_{W_k,-l}
+
\int \EE\!\big[Y(\bm W)\mid \bm W\big]^2
\Big(\prod_{k} d\widehat{\mathbb P}_{W_k,-l}-\prod_k d\mathbb P_{W_k}\Big)\Bigg|\\
&\quad\le
2C\Bigg(\int
\Big(\widehat{\mu}_{-l}-\EE\!\big[Y(\bm W)\mid \bm W\big]\Big)^2
\prod_k d\widehat{\mathbb P}_{W_k,-l}\Bigg)^{1/2}
+
C^2\sum_{k}\int \big|\widehat p_{W_k,-l}-p_{W_k}\big|\,d\lambda_k\\
&\quad\le
2C\Big\|\widehat{\mu}_{-l}-\EE\!\big[Y(\bm W)\mid \bm W\big]\Big\|_{L_2(\mathbb P_{\bm W})}
+o_{\PP}(1)
=o_{\PP}(1).
\end{align*}

Since $K$ is fixed, the decomposition term $b$ satisfies $\|b\|
\le
\sum_{k\in[K]} o_{\PP}(1) + K\cdot o_{\PP}(1)
=
o_{\PP}(1)$. Thus the decomposition term $b$ and the first-half fitted-score error satisfy $\|b\| = o_{\PP}(1)$ and $||(\theta_{\widehat{\mathbb{P}}_{-l}}-\theta)_{\text{first half}}|| \leq ||a||+||b|| = o_{\mathbb{P}}(1)$.

For the reduced regression,
\begin{align*}
\widehat{\mu}_{-\calS,-l}&=\int\widehat{\mu}_{-l}\prod_{k\in\calS}d\widehat\PP_{W_k,-l},\\
\EE[Y\mid\bm W_{-\calS}]&=\int\mu\prod_{k\in\calS}d\PP_{W_k}.
\end{align*}
Conditional Jensen and product telescoping give
\begin{align*}
&\|\widehat{\mu}_{-\calS,-l}-\EE[Y\mid\bm W_{-\calS}]\|_{L_2(\PP_{\bm W_{-\calS}})}\\
&\quad\le\|\widehat{\mu}_{-l}-\mu\|_{L_2(\PP_{\bm W})}+C\sum_{k\in\calS}t_k\\
&\quad=o_\PP(n^{-1/4}).
\end{align*}
Thus the reduced learner is bounded and has the required rate without a separate
assumption. The same bounds apply to its squares and its products with
$\widehat{\mu}_{-l}$.

  For the second half,\begin{align*}
    &(\theta_{\widehat{\mathbb{P}}_{-l}}-\theta)_{\text{second half}} =  \ 2Y(\bm W)\widehat{\mu}_{-\calS,-l}-2\widehat{\mu}_{-l}\widehat{\mu}_{-\calS,-l}+\sum_{j \in -\calS}\widehat{\EE}_{-l}[\widehat{\mu}_{-\calS,-l}^2\mid W_j]-|-\calS|\cdot \widehat{\EE}_{-l}[\widehat{\mu}_{-\calS,-l}^2]\\
    &+2\sum_{i \in \calS}\widehat{\EE}_{-l}[\widehat{\mu}_{-l}\widehat{\mu}_{-\calS,-l}\mid W_i]-2|\calS|\cdot \widehat{\EE}_{-l}[\widehat{\mu}_{-\calS,-l}^2]\\
    & - \Big\{2Y(\bm W) \mathbb{E}[Y(\bm W)\mid \bm W_{-\calS}]-2\mathbb{E}[Y(\bm W)\mid \bm W]\mathbb{E}[Y(\bm W)\mid \bm W_{-\calS}] + \sum_{j \in -\calS}\mathbb{E}\big[\mathbb{E}[Y(\bm W)\mid \bm W_{-\calS}]^2\mid W_j\big]\\
    &-|-\calS|\cdot\mathbb{E}\big[\mathbb{E}[Y(\bm W)\mid \bm W_{-\calS}]^2\big]+2\sum_{i \in \calS}\EE\left[\mathbb{E}[Y(\bm W)\mid \bm W]\mathbb{E}[Y(\bm W)\mid \bm W_{-\calS}]\mid W_i\right]\\
    &-2|\calS|\cdot \mathbb{E}\big[\mathbb{E}[Y(\bm W)\mid \bm W_{-\calS}]^2\big]\Big\}\\
    = & \ \underbrace{2Y(\bm W)\widehat{\mu}_{-\calS,-l}-2\widehat{\mu}_{-l}\widehat{\mu}_{-\calS,-l}-2Y(\bm W) \mathbb{E}[Y(\bm W)\mid \bm W_{-\calS}]+2\mathbb{E}[Y(\bm W)\mid \bm W]\mathbb{E}[Y(\bm W)\mid \bm W_{-\calS}]}_{:=a}\\
    & + \underbrace{\sum_{j \in -\calS}\widehat{\EE}_{-l}[\widehat{\mu}_{-\calS,-l}^2\mid W_j]-|-\calS|\cdot \widehat{\EE}_{-l}[\widehat{\mu}_{-\calS,-l}^2]-\sum_{j \in -\calS}\mathbb{E}\big[\mathbb{E}[Y(\bm W)\mid \bm W_{-\calS}]^2\mid W_j\big]}_{:=b}\\
    &\underbrace{+|-\calS|\cdot\mathbb{E}\big[\mathbb{E}[Y(\bm W)\mid \bm W_{-\calS}]^2\big]}_{:=b\ \text{continued}}\\
     &+  \underbrace{2\sum_{i \in \calS}\widehat{\EE}_{-l}[\widehat{\mu}_{-l}\widehat{\mu}_{-\calS,-l}\mid W_i]-2|\calS|\cdot \widehat{\EE}_{-l}[\widehat{\mu}_{-\calS,-l}^2]-2\sum_{i \in \calS}\EE\left[\mathbb{E}[Y(\bm W)\mid \bm W]\mathbb{E}[Y(\bm W)\mid \bm W_{-\calS}]\mid W_i\right]}_{:=c}\\
&+\underbrace{2|\calS|\cdot \mathbb{E}\big[\mathbb{E}[Y(\bm W)\mid \bm W_{-\calS}]^2\big]}_{:=c \ \text{continued}}.
\end{align*}

Also, under C1. and C2., we have\begin{align*}
    & \ ||a|| = ||2Y(\bm W)\widehat{\mu}_{-\calS,-l}-2\widehat{\mu}_{-l}\widehat{\mu}_{-\calS,-l}-2Y(\bm W) \mathbb{E}[Y(\bm W)\mid \bm W_{-\calS}]+2\mathbb{E}[Y(\bm W)\mid \bm W]\mathbb{E}[Y(\bm W)\mid \bm W_{-\calS}]||\\
  = & \ ||2Y(\bm W)\cdot (\widehat{\mu}_{-\calS,-l}-\mathbb{E}[Y(\bm W)\mid \bm W_{-\calS}])-2\widehat{\mu}_{-l}\widehat{\mu}_{-\calS,-l}+2\widehat{\mu}_{-l}\mathbb{E}[Y(\bm W)\mid \bm W_{-\calS}]\\
  & - 2\widehat{\mu}_{-l}\mathbb{E}[Y(\bm W)\mid \bm W_{-\calS}] +2\mathbb{E}[Y(\bm W)\mid \bm W]\mathbb{E}[Y(\bm W)\mid \bm W_{-\calS}]||\\
  = & \ ||2(Y(\bm W)-\widehat{\mu}_{-l})\cdot(\widehat{\mu}_{-\calS,-l}-\mathbb{E}[Y(\bm W)\mid \bm W_{-\calS}])-2\mathbb{E}[Y(\bm W)\mid \bm W_{-\calS}]\cdot(\widehat{\mu}_{-l}-\mathbb{E}[Y(\bm W)\mid \bm W])||\\
   \le & \ 4C\|\widehat\mu_{-\calS,-l}-\EE[Y\mid\bm W_{-\calS}]\|_2+2C\|\widehat\mu_{-l}-\EE[Y\mid\bm W]\|_2\\
   = & \ o_{\mathbb{P}}(1).
\end{align*}The remaining decomposition terms satisfy $||b|| = o_{\mathbb{P}}(1)$ and $||c|| = o_{\mathbb{P}}(1)$, where the proofs are similar to that in the first half. Thus the decomposition term $c$ and the second-half fitted-score error satisfy $||c|| = o_{\PP}(1)$ and $||(\theta_{\widehat{\mathbb{P}}_{-l}}-\theta)_{\text{second half}}|| \leq ||a||+||b||+||c||= o_{\mathbb{P}}(1)$, and the empirical process term is of the order $o_{\mathbb{P}}(1/\sqrt{n})$.

The exact form of the first half of the remainder term $R(\widehat{\mathbb{P}}_{-l},\mathbb{P})  =  \ \Theta_{\widehat{\mathbb{P}}_{-l}}-\Theta+\mathbb{P}\{\theta_{\widehat{\mathbb{P}}_{-l}}\}$ is\begin{align*}
    R(\widehat{\mathbb{P}}_{-l},\mathbb{P})_{\text{first half}}  = & \ \widehat{\EE}_{-l}[\widehat{\mu}_{-l}^2]  - \int \mathbb{E}\left[Y(\bm W)\mid \bm W\right]^2 d\mathbb{P}_{\bm W} \\
    & +\int \Big\{2Y(\bm W)\widehat{\mu}_{-l}-2\widehat{\mu}_{-l}^2+\sum_{k \in [K]}\widehat{\EE}_{-l}[\widehat{\mu}_{-l}^2\mid W_k]-K\cdot \widehat{\EE}_{-l}[\widehat{\mu}_{-l}^2]\Big\}d\mathbb{P} \\
    = & \underbrace{\int \Big\{2Y(\bm W)\widehat{\mu}_{-l}-\widehat{\mu}_{-l}^2-\mathbb{E}\left[Y(\bm W)\mid \bm W\right]^2\Big\} d\mathbb{P}}_{:=a}\\
    & +\underbrace{\int\Big\{\sum_{k \in [K]}\widehat{\EE}_{-l}[\widehat{\mu}_{-l}^2\mid W_k]-K\cdot\widehat{\EE}_{-l}[\widehat{\mu}_{-l}^2]+\widehat{\EE}_{-l}[\widehat{\mu}_{-l}^2]-\widehat{\mu}_{-l}^2\Big\}d\mathbb{P}_{\bm W}}_{:=b}.
\end{align*}

Note that, under C1. and C3.,\begin{align*}
  |a| = & \  |\int \Big\{2Y(\bm W)\widehat{\mu}_{-l}-\widehat{\mu}_{-l}^2-\mathbb{E}\left[Y(\bm W)\mid \bm W\right]^2\Big\} d\mathbb{P}|\\
   = & \ \int (\widehat{\mu}_{-l}-\mathbb{E}\left[Y(\bm W)\mid \bm W\right])^2d\mathbb{P}_{\bm W}\\
   = & \ ||\widehat{\mu}_{-l}-\mathbb{E}\left[Y(\bm W)\mid \bm W\right]||^2\\
   = & \ o_{\mathbb{P}}(1/\sqrt{n}).
\end{align*}
To bound the first-half remainder term $b$, define, locally in this proof, the product-law remainder operator for $f\in L_2(\PP_{\bm W})$ by
\[
B_f=\sum_{k=1}^K\EE\{\widehat\EE_{-l}[f\mid W_k]\}-(K-1)\widehat\EE_{-l}[f]-\EE[f].
\]
In this proof, put
$r_k=(p_{W_k}-\widehat p_{W_k,-l})/p_{W_k}$ and $e_k=\|r_k\|_{L_2(\PP_{W_k})}$.
The exact subset expansion of the product measures gives
\[
B_f=-\sum_{\substack{I\subseteq[K]\\ |I|\ge2}}
\EE\!\left[f\prod_{i\in I}r_i
                  \prod_{j\notin I}\frac{\widehat p_{W_j,-l}}{p_{W_j}}\right].
\]
Every subset occurs once, including the higher-order terms. Cauchy--Schwarz
and product independence imply
\begin{align*}
|B_f|&\le\|f\|_{L_2(\PP_{\bm W})}
\sum_{\substack{I\subseteq[K]\\ |I|\ge2}}
\prod_{i\in I}e_i\prod_{j\notin I}(1+e_j),\\
e_k&\le\underline c^{-1}\|p_{W_k}-\widehat p_{W_k,-l}\|_{L_2(\PP_{W_k})}.
\end{align*}
Here $\|\widehat p_{W_k,-l}/p_{W_k}\|_{L_2(\PP_{W_k})}\le1+e_k$. For fixed $K$, every summand has
at least two $o_\PP(n^{-1/4})$ factors. Hence $B_f=o_\PP(n^{-1/2})$
whenever $\|f\|_{L_2(\PP_{\bm W})}=O_\PP(1)$.
The first-half remainder is exactly
$-\EE[(\widehat{\mu}_{-l}-\mu)^2]+B_{\widehat{\mu}_{-l}^{\,2}}$, proving its required rate.
For a coordinate subset $J$ and a function $f$ of $\bm W_J$, write $B_f^J$ for the same definition and bound with $[K]$ replaced by $J$ and $K$ by $|J|$.

The empirical-norm statement that motivated \Cref{L2swap-train} can also be
made precise. Set $\Delta_{k,-l}=p_{W_k}-\widehat p_{W_k,-l}$ and use the lemma with $h=\Delta_{k,-l}$,
restricting C5's class to its nonnegative members.
The assumed marginal-density rate gives $\PP\Delta_{k,-l}^2\to_\PP0$, and C5 supplies membership in a fixed
Donsker class. The nonnegative compactness argument in the lemma therefore
\emph{derives}, rather than assumes,
\begin{equation}
\|\Delta_{k,-l}^2\|_{L_2(\PP)}
=\|\Delta_{k,-l}\|_{L_4(\PP)}^2\to_\PP0.
\label{eq:L4-to-0}
\end{equation}
Consequently
\begin{equation}
\frac{1}{n-n_l}\sum_{i\in\calI_{-l}}\Delta_{k,-l}((\bm W_i)_k)^2
=\|\Delta_{k,-l}\|_{L_2(\PP)}^2+o_\PP(n^{-1/2})
=o_\PP(n^{-1/2}).
\label{eq:decomp-delta}
\end{equation}
This is an identity for the empirical training measure, not for an arbitrary
fitted product law $\widehat\PP_{\bm W,-l}$. The preceding product-law bound does not use this
empirical substitution or C5.

 For the second half,\begin{align*}
    &R(\widehat{\mathbb{P}}_{-l},\mathbb{P})_{\text{second half}}\\ 
    = & \ \widehat{\EE}_{-l}[\widehat{\mu}_{-\calS,-l}^2]  - \int \mathbb{E}\left[Y(\bm W)\mid \bm W_{-\calS}\right]^2 d\mathbb{P}_{\bm W_{-\calS}} \\
    & + \int \Big\{2Y(\bm W)\widehat{\mu}_{-\calS,-l}-2\widehat{\mu}_{-l}\widehat{\mu}_{-\calS,-l} + \sum_{j\in -\calS}\widehat{\mathbb{E}}_{-l}\left[\widehat{\mu}_{-\calS,-l}^2\mid W_j\right]-|-\calS|\cdot\widehat{\mathbb{E}}_{-l}\left[\widehat{\mu}_{-\calS,-l}^2\right]\\
 & + 2\sum_{i\in \calS}\widehat{\mathbb{E}}_{-l}\left[\widehat{\mu}_{-l}\widehat{\mu}_{-\calS,-l}\mid W_i\right]-2|\calS|\cdot\widehat{\mathbb{E}}_{-l}\left[\widehat{\mu}_{-\calS,-l}^2\right]\Big\}d\mathbb{P}\\
  = & \ \int \Big\{2Y(\bm W)\widehat{\mu}_{-\calS,-l}-\widehat{\mu}_{-\calS,-l}^2-\mathbb{E}\left[Y(\bm W)\mid \bm W_{-\calS}\right]^2\Big\}d\PP\\
  & + \int \Big\{\sum_{j\in -\calS}\widehat{\mathbb{E}}_{-l}\left[\widehat{\mu}_{-\calS,-l}^2\mid W_j\right]-|-\calS|\cdot\widehat{\mathbb{E}}_{-l}\left[\widehat{\mu}_{-\calS,-l}^2\right]+\widehat{\mathbb{E}}_{-l}\left[\widehat{\mu}_{-\calS,-l}^2\right]-\widehat{\mu}_{-\calS,-l}^2\Big\}d\PP_{\bm W}\\
  & + 2\int \Big\{\sum_{k\in [K]}\widehat{\EE}_{-l}[\widehat{\mu}_{-l}\widehat{\mu}_{-\calS,-l}\mid W_k]-K\cdot \widehat{\EE}_{-l}[\widehat{\mu}_{-l}\widehat{\mu}_{-\calS,-l}]+\widehat{\EE}_{-l}[\widehat{\mu}_{-l}\widehat{\mu}_{-\calS,-l}]-\widehat{\mu}_{-l}\widehat{\mu}_{-\calS,-l}\Big\}d\PP_{\bm W}\\
  &-2\int \Big\{\sum_{j\in -\calS}\widehat{\EE}_{-l}[\widehat{\mu}^2_{-\calS,-l}\mid W_j]-|-\calS|\cdot \widehat{\EE}_{-l}[\widehat{\mu}^2_{-\calS,-l}]+\widehat{\EE}_{-l}[\widehat{\mu}^2_{-\calS,-l}]-\widehat{\mu}^2_{-\calS,-l}\Big\}d\PP_{\bm W}\\
   = & \ \underbrace{\int \Big\{2Y(\bm W)\widehat{\mu}_{-\calS,-l}-\widehat{\mu}_{-\calS,-l}^2-\mathbb{E}\left[Y(\bm W)\mid \bm W_{-\calS}\right]^2\Big\}d\PP}_{:=a}\\
  & - \underbrace{\int \Big\{\sum_{j\in -\calS}\widehat{\mathbb{E}}_{-l}\left[\widehat{\mu}_{-\calS,-l}^2\mid W_j\right]-|-\calS|\cdot\widehat{\mathbb{E}}_{-l}\left[\widehat{\mu}_{-\calS,-l}^2\right]+\widehat{\mathbb{E}}_{-l}\left[\widehat{\mu}_{-\calS,-l}^2\right]-\widehat{\mu}_{-\calS,-l}^2\Big\}d\PP_{\bm W}}_{:=b}\\
  & + \underbrace{2\int \Big\{\sum_{k\in [K]}\widehat{\EE}_{-l}[\widehat{\mu}_{-l}\widehat{\mu}_{-\calS,-l}\mid W_k]-K\cdot \widehat{\EE}_{-l}[\widehat{\mu}_{-l}\widehat{\mu}_{-\calS,-l}]+\widehat{\EE}_{-l}[\widehat{\mu}_{-l}\widehat{\mu}_{-\calS,-l}]-\widehat{\mu}_{-l}\widehat{\mu}_{-\calS,-l}\Big\}d\PP_{\bm W}}_{:=c}
 \end{align*} The second equality holds because $\widehat{\EE}_{-l}[\widehat{\mu}_{-l}\widehat{\mu}_{-\calS,-l}] = \widehat{\EE}_{-l}\left[\widehat{\EE}_{-l}[\widehat{\mu}_{-l}\widehat{\mu}_{-\calS,-l}\mid \bm W_{-\calS}]\right] = \widehat{\EE}_{-l}[\widehat{\mu}_{-\calS,-l}^2]$ and similarly $\widehat{\EE}_{-l}[\widehat{\mu}_{-l}\widehat{\mu}_{-\calS,-l}\mid W_j] = \widehat{\EE}_{-l}\left[\widehat{\EE}_{-l}[\widehat{\mu}_{-l}\widehat{\mu}_{-\calS,-l}\mid \bm W_{-\calS}]|W_j\right] = \widehat{\EE}_{-l}[\widehat{\mu}^2_{-\calS,-l}\mid W_j]$, for $j \in -\calS$. The displayed identity is equivalently
\begin{align*}
&R(\widehat{\mathbb{P}}_{-l},\mathbb{P})_{\text{second half}}\\
&\quad=-\EE[(\widehat{\mu}_{-\calS,-l}-\EE[Y\mid\bm W_{-\calS}])^2]\\
&\qquad-B_{\widehat{\mu}_{-\calS,-l}^{\,2}}^{-\calS}
  +2B_{\widehat{\mu}_{-l}\widehat{\mu}_{-\calS,-l}}^{[K]}.
\end{align*}
The reduced-regression rate proved above and the subset bound show directly that
each term is $o_\PP(n^{-1/2})$. Thus both numerator components have negligible
population remainders; no training empirical-norm substitution is needed.

For the denominator, we denote its efficient influence function by $\eta$.
For each $k\in[K]$, define the fitted marginal mean and second moment by
\begin{align*}
\widehat{\mu}_{k,-l}(w_k)&:=\int\widehat{\mu}_{-l}(\bm w)\prod_{j\in[K]\setminus\{k\}}d\widehat{\PP}_{W_j,-l}(w_j),\\
\widehat{\upsilon}_{k,-l}(w_k)&:=\int\widehat{\upsilon}_{-l}(\bm w)\prod_{j\in[K]\setminus\{k\}}d\widehat{\PP}_{W_j,-l}(w_j).
\end{align*}
The exact form of the empirical process term is
\begin{align*}
    &\eta_{\widehat{\PP}_{-l}} - \eta \\
    =& Y(\bm W)^2-2Y(\bm W)\widehat{\EE}_{-l}[Y(\bm W)]-\widehat{\upsilon}_{-l}+2\widehat{\mu}_{-l}\widehat{\EE}_{-l}[Y(\bm W)]+\sum_{k\in[K]}\widehat{\upsilon}_{k,-l}\\
    &-2\sum_{k \in [K]}\widehat{\EE}_{-l}[Y(\bm W)]\widehat{\mu}_{k,-l}-K\cdot\widehat{\EE}_{-l}[Y(\bm W)^2]+2K\cdot\widehat{\EE}_{-l}[Y(\bm W)]^2\\
    &-\Big\{Y(\bm W)^2-2Y(\bm W)\EE[Y(\bm W)]-\EE[Y(\bm W)^2\mid \bm W]+2\EE[Y(\bm W)\mid \bm W]\EE[Y(\bm W)]\Big\}\\
    &-\Big\{\sum_{k\in[K]}\EE[Y(\bm W)^2\mid W_k]-2\sum_{k\in[K]}\EE[Y(\bm W)]\EE[Y(\bm W)\mid W_k]-K\cdot\EE[Y(\bm W)^2]+2K\cdot\EE[Y(\bm W)]^2\Big\}\\
    = & 2Y(\bm W)\cdot\left(\EE[Y(\bm W)]-\widehat{\EE}_{-l}[Y(\bm W)]\right)+\left(\EE[Y(\bm W)^2\mid \bm W]-\widehat{\upsilon}_{-l}\right)\\
    &+\left(2\widehat{\mu}_{-l}\widehat{\EE}_{-l}[Y(\bm W)]-2\EE[Y(\bm W)\mid \bm W]\EE[Y(\bm W)]\right)\\
&+\sum_{k\in[K]}\left(\widehat{\upsilon}_{k,-l}-\EE[Y(\bm W)^2\mid W_k]\right)+2\sum_{k\in[K]}\left(\EE[Y(\bm W)]\EE[Y(\bm W)\mid W_k]-\widehat{\EE}_{-l}[Y(\bm W)]\widehat{\mu}_{k,-l}\right)\\
&+K\cdot\left(\EE[Y(\bm W)^2]-\widehat{\EE}_{-l}[Y(\bm W)^2]\right)+2K\cdot\left(\widehat{\EE}_{-l}[Y(\bm W)]^2-\EE[Y(\bm W)]^2\right),
\end{align*}
To justify consistency when $\widehat{\upsilon}_{-l}$ is not uniformly bounded, C4 gives,
for any coordinate subset $J$ and any square-integrable functions $f$ and $\widehat f$,
\begin{align*}
&\left\|\int(\widehat f-f)\prod_{j\in J}d\widehat\PP_{W_j,-l}\right\|_{L_2(\PP_{\bm W_{-J}})}\\
&\quad\le\|\widehat f-f\|_{L_2(\PP_{\bm W})}\prod_{j\in J}(1+e_j).
\end{align*}
For bounded $f$, product telescoping also gives
\begin{align*}
&\left\|\int f\left(\prod_{j\in J}d\widehat\PP_{W_j,-l}-\prod_{j\in J}d\PP_{W_j}\right)\right\|_2\\
&\quad\le\|f\|_\infty\sum_{j\in J}t_j.
\end{align*}
Apply these inequalities to $f=\upsilon$, which is bounded by $C^2$, and
$\widehat f=\widehat{\upsilon}_{-l}$. They prove consistency of every marginal second
moment in the variance score. The analogous bounds for the means, together
with $\widehat\EE_{-l}[Y]\to_\PP\EE[Y]$, prove
$\|\eta_{\widehat\PP_{-l}}-\eta\|_2=o_\PP(1)$.

We then move on to the remainder term:\begin{align*}
&\var(Y)_{\widehat{\PP}_{-l}} - \var(Y)+\PP\{\eta_{\widehat{\PP}_{-l}}\}\\ =&\widehat{\EE}_{-l}[Y^2]-\widehat{\EE}_{-l}[Y]^2-\EE[Y^2]+\EE[Y]^2\\
&+ \int \Big\{Y^2-2Y\cdot\widehat{\EE}_{-l}[Y]-\widehat{\upsilon}_{-l}+2\widehat{\mu}_{-l}\cdot\widehat{\EE}_{-l}[Y]+\sum_{k\in[K]}\widehat{\upsilon}_{k,-l}\\
&-2\sum_{k\in[K]}\widehat{\mu}_{k,-l}\cdot\widehat{\EE}_{-l}[Y]-K\cdot\widehat{\EE}_{-l}[Y^2]+2K\cdot\widehat{\EE}_{-l}[Y]^2\Big\} d\PP\\
= & \underbrace{\int\Big\{\EE[Y]^2-2Y\cdot\widehat{\EE}_{-l}[Y]+\widehat{\EE}_{-l}[Y]^2\Big\}d\PP}_{:=a}+\underbrace{\int\Big\{ \widehat{\EE}_{-l}[Y^2]-\widehat{\upsilon}_{-l}+\sum_{k\in[K]}\widehat{\upsilon}_{k,-l}-K\cdot\widehat{\EE}_{-l}[Y^2]\Big\}d\PP}_{:=b}\\
& + \underbrace{\int\Big\{2\widehat{\mu}_{-l}\cdot\widehat{\EE}_{-l}[Y]-2\widehat{\EE}_{-l}[Y]^2+2K\cdot\widehat{\EE}_{-l}[Y]^2-2\sum_{k\in[K]}\widehat{\mu}_{k,-l}\cdot\widehat{\EE}_{-l}[Y]\Big\} d\PP}_{:=c}.
\end{align*}

For the coherent fitted mean in \Cref{tab:EIF},
\[
\widehat\EE_{-l}[Y]=\int\widehat{\mu}_{-l}\,d\widehat\PP_{\bm W,-l},
\]
rather than an unrelated empirical mean,
\begin{align*}
|\widehat\EE_{-l}[Y]-\EE[Y]|
&\le\|\widehat{\mu}_{-l}-\mu\|_{L_2(\PP_{\bm W})}+C\sum_k t_k\\
&=o_\PP(n^{-1/4}).
\end{align*}
The preceding decomposition is exactly
\begin{align*}
&R(\widehat{\mathbb{P}}_{-l},\mathbb{P})_{\text{denominator}}\\
&\quad=(\widehat\EE_{-l}[Y]-\EE[Y])^2+B_{\widehat{\upsilon}_{-l}}\\
&\qquad-2\widehat\EE_{-l}[Y] B_{\widehat{\mu}_{-l}}.
\end{align*}
Since $\|\widehat{\upsilon}_{-l}-\upsilon\|_2=o_\PP(n^{-1/4})$ and $|\upsilon|\le C^2$,
$\|\widehat{\upsilon}_{-l}\|_2=O_\PP(1)$. The weighted subset bound therefore applies
to $\widehat{\upsilon}_{-l}$ without a sup-norm assumption, and gives
$R(\widehat{\mathbb{P}}_{-l},\mathbb{P})_{\text{denominator}}=o_\PP(n^{-1/2})$.
The actual fitted second moment also satisfies
\begin{align*}
|\widehat\EE_{-l}[Y^2]-\EE[Y^2]|
&\le\|\widehat{\upsilon}_{-l}-\upsilon\|_2\prod_k(1+e_k)+C^2\sum_k t_k\\
&=o_\PP(n^{-1/4}).
\end{align*}
It follows that $\var(Y)_{\widehat\PP_{-l}}-\var(Y)=o_\PP(n^{-1/4})$. Bounded regression learners,
their quarter rates, and the same product-law bounds give the quarter rate of
each fitted regression-square numerator. In particular, the fitted denominator
converges to $\var(Y)>0$.

It remains to connect the component expansions to the direct quotient
correction in \Cref{algo:cross.fitting}. On fold $l$, use the uncorrected component plug-ins $\Theta_{\widehat\PP_{-l}}$ and $\var(Y)_{\widehat\PP_{-l}}$.
Locally in this proof, let $u_l=\PP_n^l\{\theta_{\widehat\PP_{-l}}\}$ and $v_l=\PP_n^l\{\eta_{\widehat\PP_{-l}}\}$ denote their empirical score corrections.
The preceding population remainders and centered empirical bounds imply
\begin{align*}
\Theta_{\widehat\PP_{-l}}+u_l-\Theta
&=(\PP_n^l-\PP)\theta+o_\PP(n^{-1/2}),\\
\var(Y)_{\widehat\PP_{-l}}+v_l-\var(Y)
&=(\PP_n^l-\PP)\eta+o_\PP(n^{-1/2}).
\end{align*}
Together with the plug-in quarter rates, these identities give empirical score corrections
$u_l,v_l=o_\PP(n^{-1/4})$. Whenever the denominators are nonzero, the exact identity
\begin{align*}
&\frac{\Theta_{\widehat\PP_{-l}}+u_l}{\var(Y)_{\widehat\PP_{-l}}+v_l}\\
&\quad-\left(\frac{\Theta_{\widehat\PP_{-l}}}{\var(Y)_{\widehat\PP_{-l}}}
+\frac{u_l}{\var(Y)_{\widehat\PP_{-l}}}-\frac{\Theta_{\widehat\PP_{-l}} v_l}{\{\var(Y)_{\widehat\PP_{-l}}\}^2}\right)\\
&\qquad=\frac{v_l(\Theta_{\widehat\PP_{-l}} v_l-\var(Y)_{\widehat\PP_{-l}} u_l)}
        {\{\var(Y)_{\widehat\PP_{-l}}\}^2(\var(Y)_{\widehat\PP_{-l}}+v_l)}\\
&\qquad=o_\PP(n^{-1/2})
\end{align*}
shows that the direct correction has the same first-order expansion as the
ratio of corrected components. The delta method gives the foldwise true
influence function $\varphi=\theta/\var(Y)-\Theta\eta/\{\var(Y)\}^2$.
Using the actual weights in the algorithm,
\[
\widehat\xi-\xi
=\sum_{l=1}^L\frac{n_l}{n}(\PP_n^l-\PP)\varphi
  +o_\PP(n^{-1/2})
=\frac1n\sum_{i=1}^n\varphi(Y_i,\bm W_i;\PP)+o_\PP(n^{-1/2}).
\]
No independence between the overlapping training folds is used: there are
finitely many negligible fold remainders, and the leading sum is over the
original i.i.d.\ observations. The i.i.d.\ CLT proves the stated normal limit.

Finally, this argument also justifies the actual variance estimator in
\Cref{algo:cross.fitting}. The preceding component-score consistency and
convergence of the fitted quotient coefficients imply
$\|\varphi_{\widehat\PP_{-l}}-\varphi\|_{L_2(\PP)}\to_\PP0$.
Put $Z_i=\xi+\varphi(Y_i,\bm W_i;\PP)$. Conditional Markov's inequality
on each held-out fold, together with the plug-in consistency, gives
\[
\frac1n\sum_{i=1}^n(\widehat\xi_i-Z_i)^2\to_\PP0.
\]
For this step one may bound the conditional probability by the minimum of one
and the conditional squared error divided by a fixed threshold; its expectation
tends to zero. Thus no unscaled expected-risk assumption is required.
The law of large numbers for $Z_i$ and $Z_i^2$, followed by Cauchy--Schwarz,
shows that the empirical first and second moments of $\widehat\xi_i$ converge
to those of $Z_i$. Hence
\[
\widehat\sigma^2=\frac1n\sum_i(\widehat\xi_i-\widehat\xi)^2
\to_\PP\var\{\varphi(Y,\bm W;\PP)\}.
\]
When this limit is positive, Slutsky's theorem yields studentization and the
coverage asserted in \Cref{coro:coverage}. The same variance-consistency
argument applies to the baseline influence function in \eqref{eq:IF.np}.

\end{proof}

\subsection{Proof of \Cref{prop:sharp.null}}\label{appe:sec:proof:sharp:null}
\begin{proof}
In this proof, let $D=Y(\bm W)-Y(\bm W'_{\calS},\bm W_{-\calS})$. The two outcomes have the same distribution by \Cref{eq:super.population,assu:independent}, so $\EE[D]=0$. Since $\var(Y)>0$,
\begin{align}\label{eq:sharp.null}
\xi(\vee_{k\in\calS}W_k)=0
\quad\Longleftrightarrow\quad \var(D)=0
\quad\Longleftrightarrow\quad D=0\ \text{almost surely}.
\end{align}
For discrete assignments, each supported pair has positive probability, giving the stated equality almost surely over potential outcomes at that pair. For general assignments the conclusion is almost sure, not pointwise outside the assignment law.
\end{proof}

\subsection{Proof of \Cref{prop:coverage.degenerate}}\label{appe:sec:proof:prop:coverage.degenerate}
\begin{proof}
Under the sharp null, the full dataset $(\bm W_{\calS},\bm W_{-\calS},Y)$ is invariant in law under permutations of the rows of $\bm W_{\calS}$: the treatment block is independent of the retained treatments and null outcomes, and the units are i.i.d. Conditional on its permutation orbit and on $(\bm W_{-\calS},Y)$, the observed treatment array is uniform on that orbit. Independent uniform draws from the full permutation group therefore make the observed and $B$ resampled statistics jointly exchangeable, using the same analysis rule for all datasets. The plus-one rank $p$-value, with weak inequalities for ties, consequently satisfies $\PP(P\leq\alpha)\leq\alpha$.
\end{proof}

\subsection{Proof of \Cref{prop:coverage.sequential}}\label{appe:sec:proof:prop:coverage.sequential}
\begin{proof}
In this proof, let $R_n$ be first-stage rejection and $I_n$ the second-stage interval in \eqref{eq:confidence.interval}. At total explainability $\xi(\vee_{k\in\calS}W_k)=0$, noncoverage is contained in $R_n$, so its probability is at most $\alpha$ by \Cref{prop:coverage.degenerate}. At each fixed positive total explainability $\xi(\vee_{k\in\calS}W_k)>0$, noncoverage of the sequential procedure is contained in
$R_n^c\cup\{\xi(\vee_{k\in\calS}W_k)\notin I_n\}$.
The first probability tends to zero by power consistency and the second has limit superior at most $\alpha$ by \Cref{coro:coverage}. No independence between the stages is required.
With the fallback $[0,\infty)$ on nonrejection, noncoverage at a positive alternative is contained in $\{\xi(\vee_{k\in\calS}W_k)\notin I_n\}$, so the power condition is unnecessary. At the null, noncoverage remains contained in $R_n$ and is bounded by $\alpha$, not necessarily zero.
\end{proof}

\section{Additional tables}\label{appe:sec:table}

\begin{table}[tbp]
\centering
\renewcommand{\arraystretch}{1.15}
\begin{tabular}{@{}l p{0.72\linewidth}@{}}
\toprule
\textbf{Symbol} & \textbf{Meaning} \\
\midrule
$[K]$ & Index set $\{1,\ldots,K\}$ for the $K$ factors. \\
$Y$ & Outcome. \\
$\bm W=(W_1,\ldots,W_K)$ & Vector of factors. \\
$\bm W_{\mathcal S}$ & Subvector $\{W_k: k\in\mathcal S\}$ for $\mathcal S\subseteq[K]$. \\
$W_k$ & $k$-th component of $\bm W$ (for $k\in[K]$). \\
$\mathbb P,\ \widehat{\mathbb P}$ & True / estimated joint law of $(Y,\bm W)$. \\
$\mathbb P_{\bm W},\ \widehat{\mathbb P}_{\bm W}$ & True / estimated marginal law of $\bm W$. \\
$\mathbb P_{W_k},\ \widehat{\mathbb P}_{W_k}$ & True / estimated marginal law of $W_k$. \\
$\mathbb P_{Y\mid \bm W},\ \widehat{\mathbb P}_{Y\mid \bm W}$ & True / estimated conditional law of $Y$ given $\bm W$. \\
$p,\ \hat p$ & True / estimated joint density of $(Y,\bm W)$ w.r.t.\ a chosen dominating measure (when it exists). \\
$p_{\bm W},\ \widehat p_{\bm W}$ & True / estimated marginal density of $\bm W$. \\
$p_{W_k},\ \widehat p_{W_k}$ & True / estimated marginal density of $W_k$. \\
$p_{Y\mid \bm W},\ \widehat p_{Y\mid \bm W}$ & True / estimated conditional density of $Y$ given $\bm W$. \\
$\mathbb E_{\mathbb P}[\cdot],\ \mathbb E_{\widehat{\mathbb P}}[\cdot]$ & Expectation under $\mathbb P$ / $\widehat{\mathbb P}$. \\
$\mu(\bm w)$ & Conditional mean $\mathbb E_{\mathbb P}[Y\mid \bm W=\bm w]$. \\
$\upsilon(\bm w)$ & Conditional second moment $\mathbb E_{\mathbb P}[Y^2\mid \bm W=\bm w]$. \\
\midrule
$\mathcal P$ & Statistical model. \\
$\mathcal P_{\mathrm{np}}$ & Nonparametric model. \\
$\mathcal P_{\mathrm{ind}}$ & Independence submodel for $\bm W$. \\
$\dot{\mathcal P}_{\mathbb P}$ & Tangent space at $\mathbb P\in\mathcal P$. \\
$\xi$ & Causal estimand of interest. \\
$\xi\!\left(\bigvee_{k\in\mathcal S} W_k\right)$, $\mathcal S\subseteq[K]$ & Causal explainability attributable to the union of factors in $\mathcal S$. \\
$\xi\!\left(W_k \wedge W_{k'}\right)$ & Causal explainability attributable to the interaction of $W_k$ and $W_{k'}$. \\
$\varphi_{\IF}^\xi,\ \varphi_{\EIF}^\xi,\ \varphi^\xi$ & Influence function, efficient influence function, and a generic influence function/estimator-specific score for $\xi$. \\
\bottomrule
\end{tabular}
\caption{Table of notation. When unambiguous, subscripts/superscripts may be omitted.}
\label{tab:notation}
\end{table}

\begin{table}[tbp]
    \centering
    \begin{tabular}{l|l}
    \toprule
    Expression & Formula\\ \midrule
       $\widehat{\EE}[Y\mid \bm W=\bm w]$  & $\widehat{\mu}(\bm w)$ \\
       $\widehat{\EE}[Y\mid \bm W_{-\calS}=\bm w_{-\calS}]$  & $\int \widehat{\mu}(\bm w) d\widehat{\PP}_{\bm W_{\calS}}(\bm w_{\calS})$ \\    
         $\widehat{\EE}[\widehat{\EE}^2[Y\mid \bm W_{-\calS}]]$ & $\int \left(\int \widehat{\mu}(\bm w) d\widehat{\PP}_{\bm W_{\calS}}(\bm w_{\calS})\right)^2 d\widehat{\PP}_{\bm W_{-\calS}}(\bm w_{-\calS})$\\
       $\widehat{\mathbb{E}}\left[\widehat{\mathbb{E}}[Y\mid \bm W]\cdot\widehat{\mathbb{E}}[Y\mid \bm W_{-\calS}]\mid W_i=w_i\right]$, $i \in \calS$  & $\int \widehat{\mu}(\bm w) \cdot \left(\int \widehat{\mu}(\bm w) d\widehat{\PP}_{\bm W_{\calS}}(\bm w_{\calS})\right)d\widehat{\PP}_{\bm W_{-i}}(\bm w_{-i})$\\
       $\widehat{\mathbb{E}}\left[\widehat{\mathbb{E}}^2[Y\mid \bm W_{-\calS}]\mid W_j=w_j\right]$, $j \notin \calS$  & $\int \left(\int \widehat{\mu}(\bm w) d\widehat{\PP}_{\bm W_{\calS}}(\bm w_{\calS})\right)^2 d\widehat{\PP}_{\bm W_{-j}}(\bm w_{-j})$\\
       $\widehat{\EE}_{-\ell}[Y]$, $\varphi_{\EIF}$ & $\int\widehat\mu(\bm w)\,d\widehat{\PP}_{\bm W}(\bm w)$ \\
       $\widehat{\EE}_{-\ell}[Y^2]$, $\varphi_{\EIF}$ & $\int\widehat\upsilon(\bm w)\,d\widehat{\PP}_{\bm W}(\bm w)$ \\
       $\var(Y)_{\widehat{\PP}_{-\ell}}$, $\varphi_{\IF}$ & $(n-n_\ell)^{-1}\sum_{i\in\calI_{-\ell}}(Y_i-\widehat{\EE}_{-\ell}[Y])^2$ \\
       $\var(Y)_{\widehat{\PP}_{-\ell}}$, $\varphi_{\EIF}$ & $\int \widehat{\upsilon}(\bm w) d\widehat{\PP}_{\bm W}(\bm w) -\left(\int \widehat{\mu}(\bm w) d\widehat{\PP}_{\bm W}(\bm w)\right)^2$ \\ \bottomrule
       \end{tabular}
    \caption{Formulas for the one-step correction estimator.
The marginal-integration formulas use the fitted product law.
Fitted regressions and treatment marginals are trained without fold $\ell$; their fold subscripts are omitted inside the integrals.}
    \label{tab:EIF}
\end{table}

\clearpage

\section{Additional empirical studies}
\label{app:simulation}

\subsection{Additional simulations with true nuisances}\label{app:simulation.noise.magnitude}

We implement the two methods using the true nuisance functions. 
In \Cref{fig:simulation2}, we vary the interaction coefficient $\sigma$.

\begin{figure}[htbp]
    \centering
    \begin{minipage}{\linewidth}
        \centering
        \begin{minipage}{0.32\textwidth}
            \centering
            \includegraphics[clip, trim = 0cm 0cm 0cm 1cm, width=\linewidth]{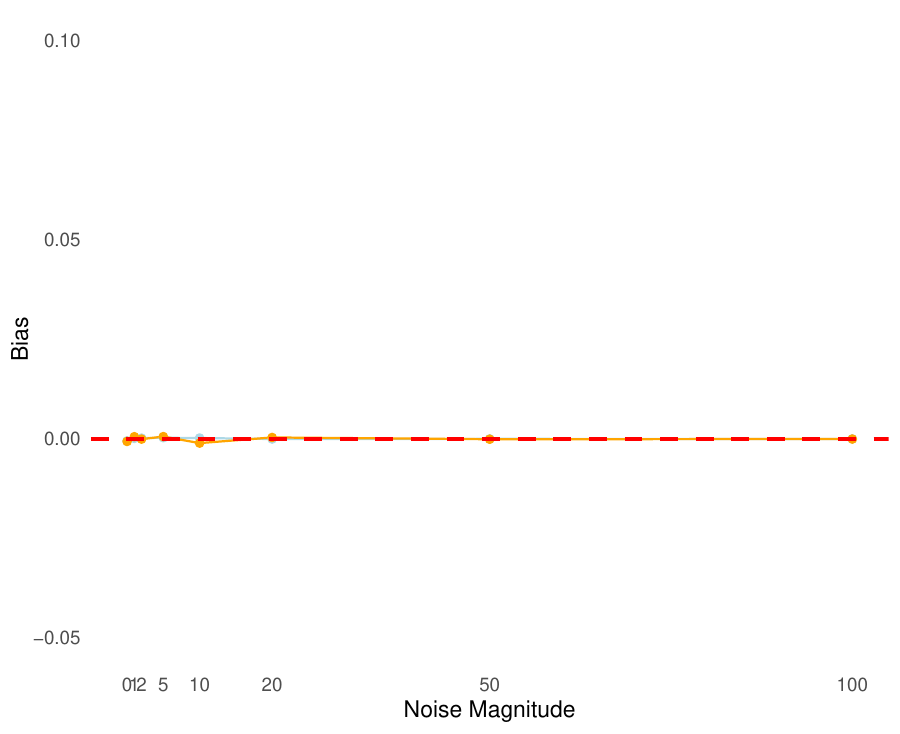}
        \end{minipage}
        \hfill
        \begin{minipage}{0.32\textwidth}
            \centering
            \includegraphics[clip, trim = 0cm 0cm 0cm 1cm, width=\linewidth]{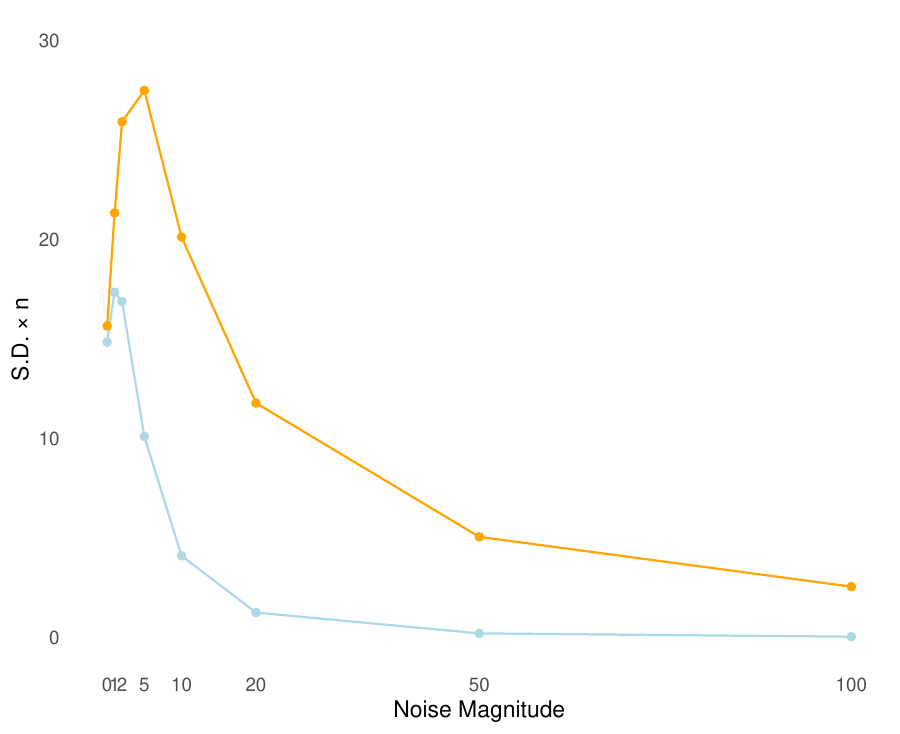}
        \end{minipage}
        \hfill
        \begin{minipage}{0.32\textwidth}
            \centering
            \includegraphics[clip, trim = 0cm 0cm 0cm 1cm, width=\linewidth]{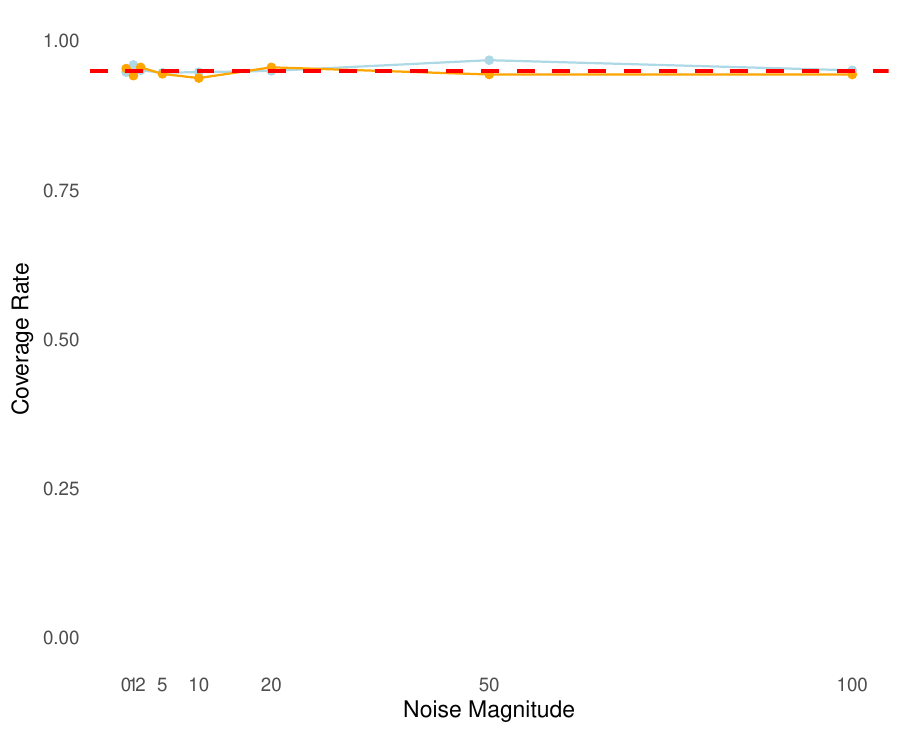}
        \end{minipage}
        \par\smallskip
        \centering
        \small{(d) $\xi( W_3)$; True nuisance functions; Varying $\sigma$; Sample size = 1000}
    \end{minipage}
    \vspace{1em}
    \begin{minipage}{\linewidth}
        \centering
        \begin{minipage}{0.32\textwidth}
            \centering
            \includegraphics[clip, trim = 0cm 0cm 0cm 1cm, width=\linewidth]{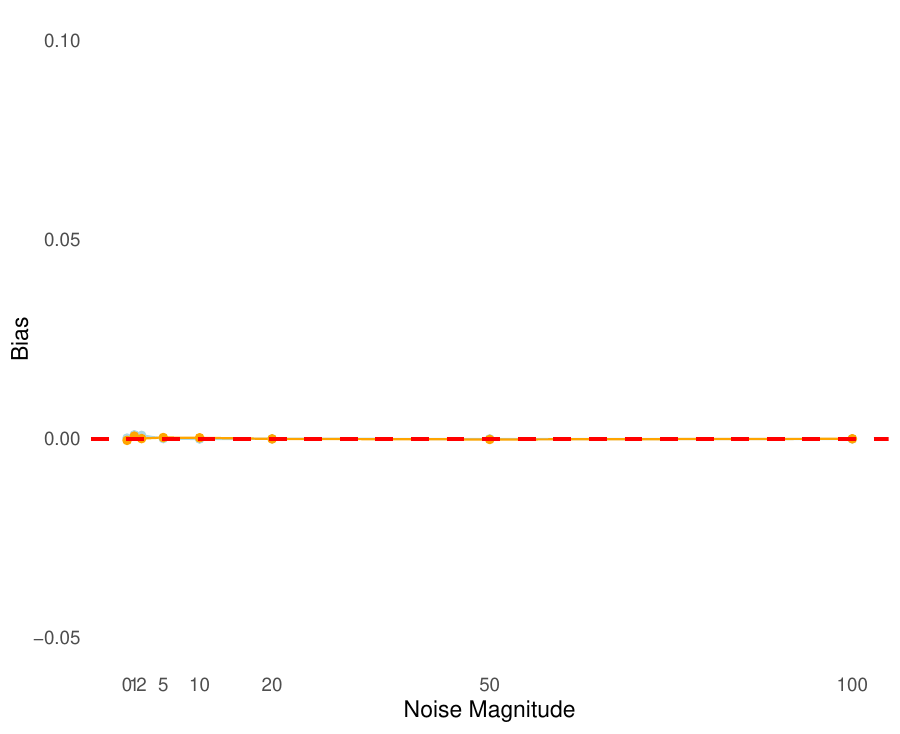}
        \end{minipage}
        \hfill
        \begin{minipage}{0.32\textwidth}
            \centering
            \includegraphics[clip, trim = 0cm 0cm 0cm 1cm, width=\linewidth]{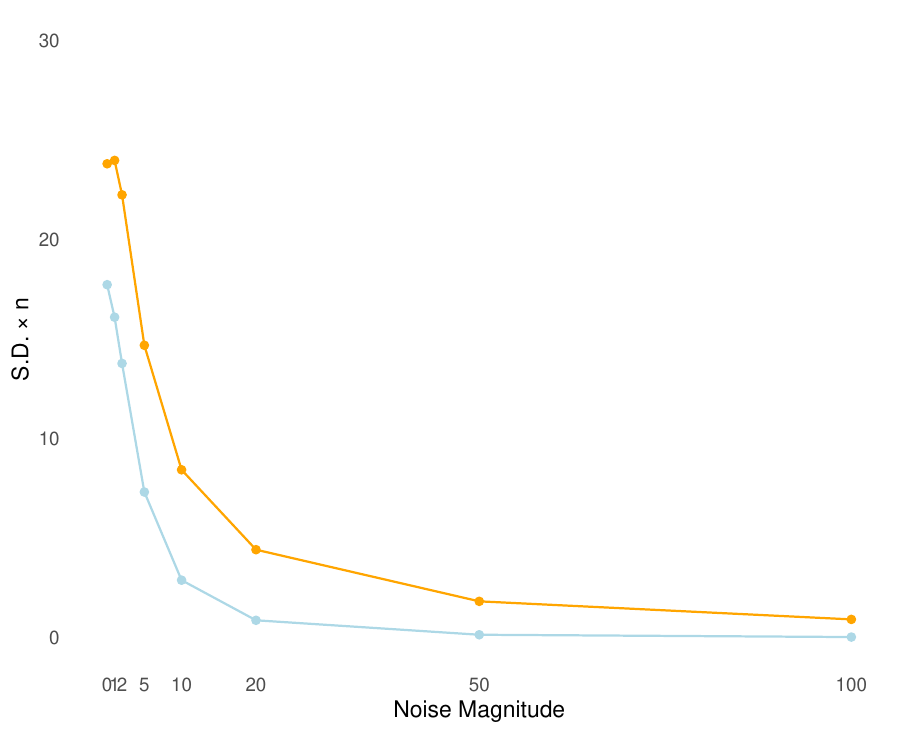}
        \end{minipage}
        \hfill
        \begin{minipage}{0.32\textwidth}
            \centering
            \includegraphics[clip, trim = 0cm 0cm 0cm 1cm, width=\linewidth]{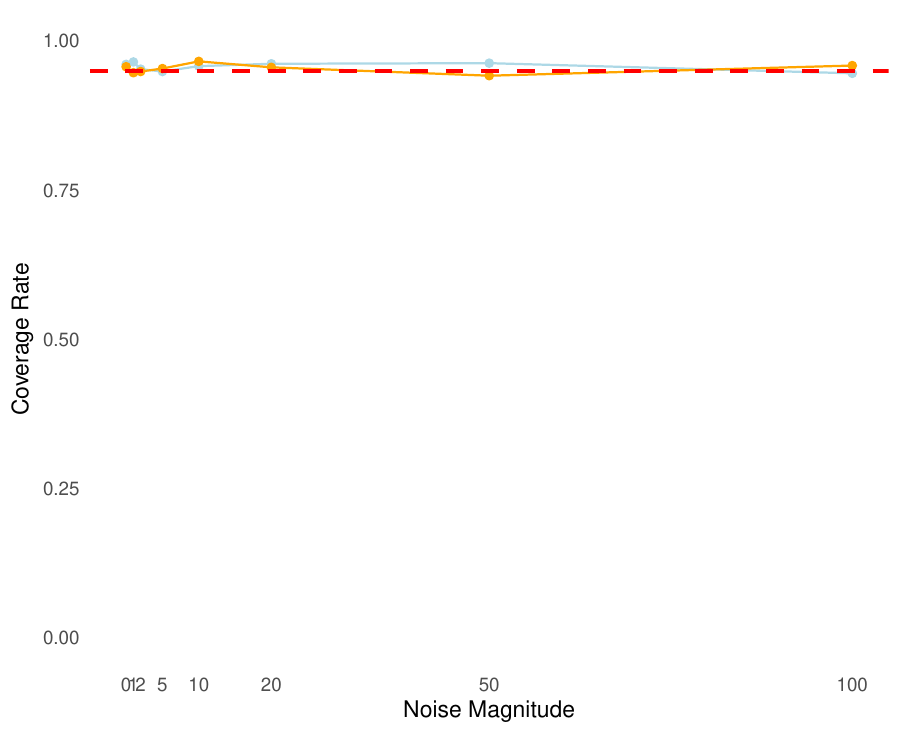}
        \end{minipage}
        \par\smallskip
        \centering
        \small{(e) $\xi(W_1 \vee W_3)$; True nuisance functions; Varying $\sigma$; Sample size = 1000}
    \end{minipage}
    \vspace{1em}

    \begin{minipage}{\linewidth}
        \centering
        \begin{minipage}{0.32\textwidth}
            \centering
            \includegraphics[clip, trim = 0cm 0cm 0cm 1cm, width=\linewidth]{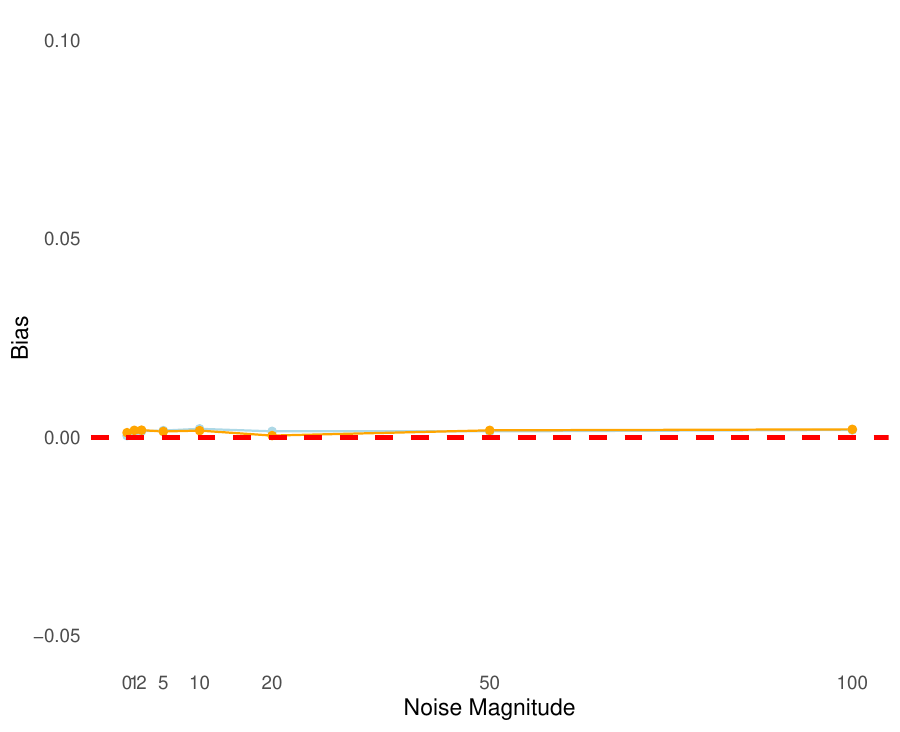}
        \end{minipage}
        \hfill
        \begin{minipage}{0.32\textwidth}
            \centering
            \includegraphics[clip, trim = 0cm 0cm 0cm 1cm, width=\linewidth]{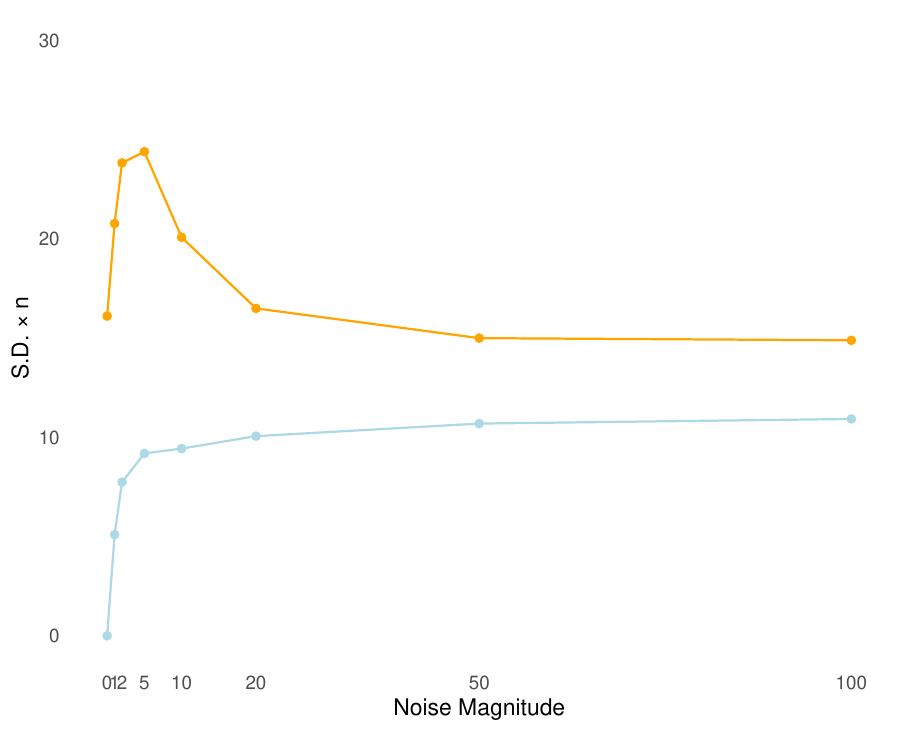}
        \end{minipage}
        \hfill
        \begin{minipage}{0.32\textwidth}
            \centering
            \includegraphics[clip, trim = 0cm 0cm 0cm 1cm, width=\linewidth]{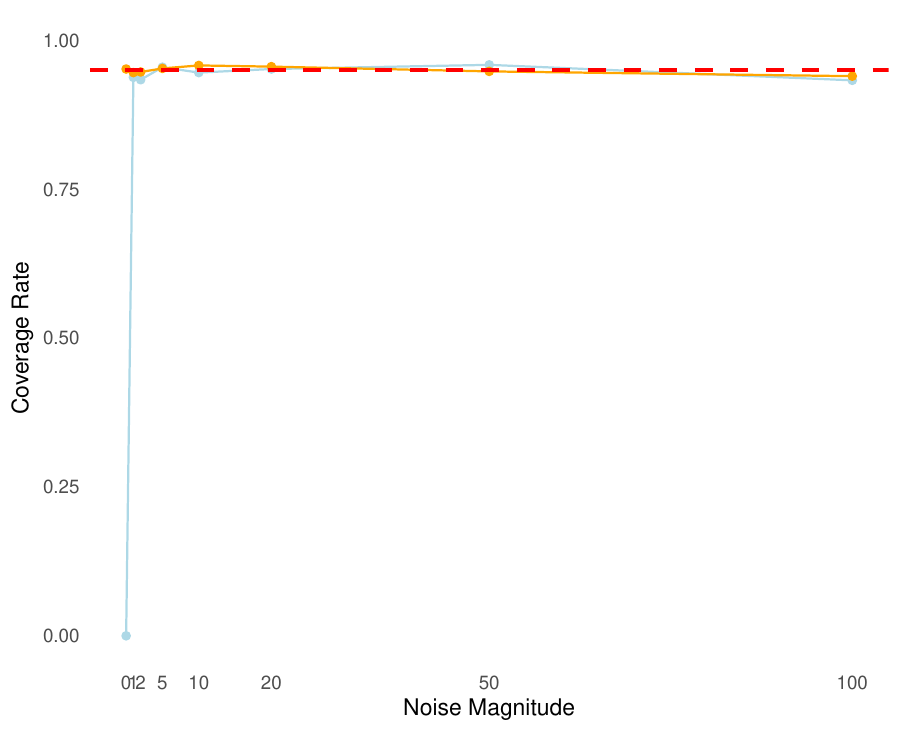}
        \end{minipage}
        \par\smallskip
        \centering
        \small{(f) $\xi(W_1 \wedge W_3)$; True nuisance functions; Varying $\sigma$; Sample size = 1000}
    \end{minipage}

  \caption{Comparison of bias (left panel), estimated standard deviation times sample size (mid panel), and coverage rate (right panel; significance level $\alpha = 0.05$) with varying $\sigma$.
    $\varphi_{\EIF}$-based method in blue, $\varphi_{\IF}$-based method in gold.
    True nuisance functions are used. Results aggregated over $1000$ trials.}
    \label{fig:simulation2}
\end{figure}

\subsection{Additional simulations with estimated nuisances}\label{app:estimated}

We compared $\varphi_{\EIF}$-based estimator and the $\varphi_{\IF}$-based estimator for total explainability in \Cref{sec:simulation}, and here we focus on the more challenging problem of interaction explainability. 
Owing to the increased complexity of nuisance function estimation inherent in interaction effects, both methods exhibit coverage below the target level. 
Recall that $\varphi_{\EIF}$-based estimator requires estimating more complicated nuisance components to exploit independence; in this setting, the theoretical gain in asymptotic variance is offset by the increased variability from nuisance estimation. As a result, the $\varphi_{\EIF}$-based estimator appears less favorable in terms of both MSE and coverage.

\begin{table}[ht]
\centering
\caption{Performance comparison of $\varphi_{\EIF}$ and $\varphi_{\IF}$-based estimators with estimated nuisance function regarding interaction expliainability. Aggregated over $100$ trials.}
\label{tab:eif_nif}
\begin{tabular}{cccc}
\toprule
Method & Bias & SD $\times n$ & Coverage ($\alpha=0.05$) \\
\midrule
$\varphi_{\EIF}$  & -0.11 & 2.87 &  0.79 \\
$\varphi_{\IF}$  & 0.07 & 2.58 & 0.83 \\
\bottomrule
\end{tabular}
\end{table}

\section{Details of the immigration experiment}\label{app:rda}

The attributes include each immigrant’s gender, education level, employment plans, job experience, profession, language skills, country of origin, reasons for applying, and prior trips to the United States. Respondents' characteristics such as age, education, ethnicity, gender, and ethnocentrism are also available in the survey.

One thing should be noted is that there are two restrictions imposed on the possible combinations of immigrant attributes. First, those immigrant profiles
who were fleeing persecution were restricted to come from countries where such an application was
plausible (e.g., Iraq, Sudan). Second, those immigrants occupying high-skill occupations (e.g.,
research scientists, doctors) were constrained to have at least two years of college education. In this case, the distribution of immigrants’ occupations is dependent on their education, but conditionally independent of the other attributes and the other profiles in the same choice task given their education \parencite{hainmueller2015conjoint}. However, as this is an artificial experiment, it is very unclear whether the attributes are \textit{causally} dependent or not. The reason is that the restrictions are put in the design stage, so attributes are simultaneously randomized at the moment the profile is generated. To avoid the problem of unknown directions in a directed acyclic graph, we choose to collapse the two pairs of dependent variables into two variables —— \textit{Origin \& Application reason} and \textit{Education \& Profession}. Another thing is that although it is a randomized experiment, the authors did not provide any information on the randomization mechanism. So we do the analysis assuming that the propensity scores are unknown.

\begin{table}[ht]
\centering
\caption{Conditional Randomization Test \parencite{ham2024using}}
\begin{tabular}{ll}
\hline
\textbf{Interaction} & \textbf{$p$ value} \\
\hline
Gender \& Job Experience & 0.9108911  \\
Job Plan \& Gender & 0.6831683 \\
Job Plan \& Language & 0.0990099 \\
Job Plan \& Job Experience & 0.00990099** \\
Job Plan \& Prior trips to U.S. & 0.5049505 \\
\hline
\end{tabular}
\end{table}

\begin{table}[ht]
\centering
\caption{Average Component Interaction Effects (ACIE) \cite{hainmueller2015conjoint}}
\resizebox{\textwidth}{!}{%
\begin{tabular}{l r r r r}
\toprule
\textbf{Attribute} & \textbf{Estimate} & \textbf{Std. Err} & \textbf{z value} & \textbf{Pr($>|z|$)} \\
\midrule
\multicolumn{5}{l}{\textit{Gender:Job Experience}} \\
\quad male:1-2 years & -0.0017453 & 0.023181 & -0.07529 & 0.93998 \\
\quad male:3-5 years & -0.0215082 & 0.024037 & -0.89481 & 0.37089 \\
\quad male:5+ years & -0.0088778 & 0.024508 & -0.36224 & 0.71717 \\
\addlinespace %
\multicolumn{5}{l}{\textit{Gender:Job Plan}} \\
\quad male:contract with employer & 0.0034950 & 0.023374 & 0.14953 & 0.88114 \\
\quad male:interviews with employer & -0.0185813 & 0.024573 & -0.75616 & 0.44955 \\
\quad male:no plans to look for work & 0.0042381 & 0.022811 & 0.18579 & 0.85261 \\
\addlinespace
\multicolumn{5}{l}{\textit{Job Plan:Language}} \\
\quad contract with employer:broken English & 0.0477691 & 0.031921 & 1.49649 & 0.134527 \\
\quad interviews with employer:broken English & 0.0652677 & 0.033518 & 1.94725 & 0.051505 \\
\quad no plans to look for work:broken English & 0.0117892 & 0.032696 & 0.36057 & 0.718423 \\
\quad contract with employer:tried English but unable & 0.0100675 & 0.032462 & 0.31013 & 0.756463 \\
\quad interviews with employer:tried English but unable & 0.0274033 & 0.032797 & 0.83555 & 0.403405 \\
\quad no plans to look for work:tried English but unable & 0.0069235 & 0.031912 & 0.21696 & 0.828243 \\
\quad contract with employer:used interpreter & 0.0150713 & 0.032272 & 0.46700 & 0.640497 \\
\quad interviews with employer:used interpreter & 0.0499209 & 0.032747 & 1.52444 & 0.127399 \\
\quad no plans to look for work:used interpreter & 0.0402530 & 0.031193 & 1.29046 & 0.196892 \\
\bottomrule
\end{tabular}%
}
\par\medskip %
\footnotesize
\centering
Number of Obs. = 13960, Number of Respondents = 1396. \\
Significance codes: 0 *** 0.001 ** 0.01 * 0.05 
\end{table}

\begin{table}[ht]
\centering
\caption{Average Component Interaction Effects (ACIE) -- continued}
\resizebox{\textwidth}{!}{%
\begin{tabular}{l r r r r}
\toprule
\textbf{Attribute} & \textbf{Estimate} & \textbf{Std. Err} & \textbf{z value} & \textbf{Pr($>|z|$)} \\
\midrule
\multicolumn{5}{l}{\textit{Job Experience:Job Plan}} \\
\quad 1-2 years:contract with employer & 0.03528079 & 0.032295 & 1.092449 & 0.274636 \\
\quad 3-5 years:contract with employer & 0.00069369 & 0.033193 & 0.020899 & 0.983327 \\
\quad 5+ years:contract with employer & -0.02503323 & 0.033785 & -0.740960 & 0.458718 \\
\quad 1-2 years:interviews with employer & 0.01152705 & 0.033719 & 0.341853 & 0.732462 \\
\quad 3-5 years:interviews with employer & -0.02887657 & 0.034100 & -0.846821 & 0.397095 \\
\quad 5+ years:interviews with employer & -0.03515407 & 0.034809 & -1.009902 & 0.312542 \\
\quad 1-2 years:no plans to look for work & -0.00947529 & 0.031805 & -0.297922 & 0.765762 \\
\quad 3-5 years:no plans to look for work & -0.07701162 & 0.032256 & -2.387519 & 0.016963$^*$ \\
\quad 5+ years:no plans to look for work & -0.06210952 & 0.032943 & -1.885391 & 0.059377 \\
\addlinespace %
\multicolumn{5}{l}{\textit{Job Plan:Prior Entry}} \\
\quad contract with employer:once as tourist & -0.0306128 & 0.036545 & -0.83767 & 0.402218 \\
\quad interviews with employer:once as tourist & 0.0639714 & 0.037340 & 1.71320 & 0.086676 \\
\quad no plans to look for work:once as tourist & 0.0469911 & 0.037612 & 1.24936 & 0.211533 \\
\quad contract with employer:many times as tourist & -0.0034264 & 0.036436 & -0.09404 & 0.925077 \\
\quad interviews with employer:many times as tourist & 0.0075229 & 0.038365 & 0.19609 & 0.844544 \\
\quad no plans to look for work:many times as tourist & 0.0355192 & 0.037469 & 0.94797 & 0.343144 \\
\quad contract with employer:six months with family & 0.0191688 & 0.036835 & 0.52039 & 0.602792 \\
\quad interviews with employer:six months with family & 0.0508093 & 0.037914 & 1.34012 & 0.180206 \\
\quad no plans to look for work:six months with family & 0.0392413 & 0.036617 & 1.07166 & 0.283871 \\
\quad contract with employer:once w/o authorization & 0.0091775 & 0.036384 & 0.25224 & 0.800859 \\
\quad interviews with employer:once w/o authorization & 0.0302793 & 0.037539 & 0.80661 & 0.419892 \\
\quad no plans to look for work:once w/o authorization & 0.0620386 & 0.035214 & 1.76175 & 0.078112 \\
\bottomrule
\end{tabular}%
}
\par\medskip
\footnotesize
\centering
Number of Obs. = 13960, Number of Respondents = 1396. \\
Significance codes: 0 *** 0.001 ** 0.01 * 0.05 
\end{table}

\begin{table}[ht]
\centering
\caption{Average Component Interaction Effects (ACIE) -- continued}
\resizebox{\textwidth}{!}{%
\begin{tabular}{l r r r r}
\toprule
\textbf{Attribute} & \textbf{Estimate} & \textbf{Std. Err} & \textbf{z value} & \textbf{Pr($>|z|$)} \\
\midrule
\multicolumn{5}{l}{\textit{Job:Job Plan}} \\
\quad waiter:contract with employer & -0.0327390 & 0.047473 & -0.689639 & 0.490421 \\
\quad child care provider:contract with employer & -0.0900152 & 0.050735 & -1.774211 & 0.076028 \\
\quad gardener:contract with employer & 0.0311797 & 0.048505 & 0.642814 & 0.520345 \\
\quad financial analyst:contract with employer & -0.0484468 & 0.085065 & -0.569525 & 0.569000 \\
\quad construction worker:contract with employer & -0.0493627 & 0.049477 & -0.997686 & 0.318432 \\
\quad teacher:contract with employer & -0.0363747 & 0.050976 & -0.713561 & 0.475499 \\
\quad computer programmer:contract with employer & -0.1469473 & 0.085449 & -1.719704 & 0.085486 \\
\quad nurse:contract with employer & -0.0603685 & 0.047341 & -1.275183 & 0.202245 \\
\quad research scientist:contract with employer & -0.1356852 & 0.080691 & -1.681531 & 0.092660 \\
\quad doctor:contract with employer & -0.1538243 & 0.077780 & -1.977672 & 0.047966$^*$ \\
\quad waiter:interviews with employer & -0.0450223 & 0.049338 & -0.912537 & 0.361486 \\
\quad child care provider:interviews with employer & -0.0383598 & 0.049890 & -0.768895 & 0.441956 \\
\quad gardener:interviews with employer & -0.0413542 & 0.050545 & -0.818161 & 0.413265 \\
\quad financial analyst:interviews with employer & -0.0863427 & 0.092088 & -0.937614 & 0.348443 \\
\quad construction worker:interviews with employer & -0.0211739 & 0.051522 & -0.410968 & 0.681096 \\
\quad teacher:interviews with employer & 0.0301964 & 0.050484 & 0.598142 & 0.549746 \\
\quad computer programmer:interviews with employer & -0.1196731 & 0.088597 & -1.350757 & 0.176773 \\
\quad nurse:interviews with employer & -0.0116562 & 0.049813 & -0.234000 & 0.814985 \\
\quad research scientist:interviews with employer & -0.0679747 & 0.084104 & -0.808220 & 0.418964 \\
\quad doctor:interviews with employer & -0.1211354 & 0.082133 & -1.474867 & 0.140248 \\
\quad waiter:no plans to look for work & 0.0028488 & 0.047582 & 0.059871 & 0.952258 \\
\quad child care provider:no plans to look for work & 0.0035945 & 0.048850 & 0.073583 & 0.941343 \\
\quad gardener:no plans to look for work & 0.0303911 & 0.048098 & 0.631859 & 0.527479 \\
\quad financial analyst:no plans to look for work & 0.0443680 & 0.086945 & 0.510300 & 0.609842 \\
\quad construction worker:no plans to look for work & -0.0190767 & 0.048589 & -0.392611 & 0.694607 \\
\quad teacher:no plans to look for work & 0.0125936 & 0.049917 & 0.252292 & 0.800815 \\
\quad computer programmer:no plans to look for work & 0.0012924 & 0.083344 & 0.015506 & 0.987628 \\
\quad nurse:no plans to look for work & -0.0131314 & 0.047923 & -0.274010 & 0.784077 \\
\quad research scientist:no plans to look for work & 0.0331237 & 0.085202 & 0.388766 & 0.697449 \\
\quad doctor:no plans to look for work & -0.0040783 & 0.079754 & -0.051136 & 0.959217 \\
\bottomrule
\end{tabular}%
}
\par\medskip
\footnotesize
\centering
Number of Obs. = 13960, Number of Respondents = 1396. \\
Significance codes: 0 *** 0.001 ** 0.01 * 0.05 
\end{table}

\end{document}